\documentclass[a4paper,USenglish,cleveref]{lipics-v2019}
\hideLIPIcs
\nolinenumbers

\usepackage{newunicodechar}
\usepackage{silence}
\newunicodechar{♢}{\tikz \node[inner sep=1.5,draw,diamond] {};}
\newunicodechar{☆}{\tikz \node[inner sep=1,draw,star,star point ratio=2] {};}
\newunicodechar{△}{\kern 0.5pt\tikz \node[inner sep=1.2,draw,regular polygon,regular polygon sides=3] {};\kern 0.5pt}
\newunicodechar{⬜}{\kern 0.5pt\tikz \node[inner sep=1.7,draw,regular polygon,regular polygon sides=4] {};\kern 0.5pt}
\newunicodechar{◯}{\tikz[baseline=-3pt] \node[inner sep=1.7,draw,cloud,cloud puffs=4,cloud puff arc=190] {};}

\newunicodechar{⊥}{\bot}
\newunicodechar{•}{\item}
\newunicodechar{✓}{\checkmark}
\newunicodechar{✗}{\xmark}
\newunicodechar{…}{\dots}
\newunicodechar{≔}{\coloneqq}
\newunicodechar{⁻}{^-}
\newunicodechar{⁺}{^+}
\newunicodechar{₋}{_-}
\newunicodechar{₊}{_+}
\newunicodechar{ℓ}{\ell}
\newunicodechar{•}{\item}
\newunicodechar{…}{\dots}
\newunicodechar{≔}{\coloneqq}
\newunicodechar{≤}{\leq}
\newunicodechar{≥}{\geq}
\newunicodechar{≰}{\nleq}
\newunicodechar{≱}{\ngeq}
\newunicodechar{⊕}{\oplus}
\newunicodechar{⊗}{\otimes}
\newunicodechar{≠}{\neq}
\newunicodechar{¬}{\neg}
\newunicodechar{≡}{\equiv}
\newunicodechar{₀}{_0}
\newunicodechar{₁}{_1}
\newunicodechar{₂}{_2}
\newunicodechar{₃}{_3}
\newunicodechar{₄}{_4}
\newunicodechar{₅}{_5}
\newunicodechar{₆}{_6}
\newunicodechar{₇}{_7}
\newunicodechar{₈}{_8}
\newunicodechar{₉}{_9}
\newunicodechar{⁰}{^0}
\newunicodechar{¹}{^1}
\newunicodechar{²}{^2}
\newunicodechar{³}{^3}
\newunicodechar{⁴}{^4}
\newunicodechar{⁵}{^5}
\newunicodechar{⁶}{^6}
\newunicodechar{⁷}{^7}
\newunicodechar{⁸}{^8}
\newunicodechar{⁹}{^9}
\newunicodechar{∈}{\in}
\newunicodechar{∉}{\notin}
\newunicodechar{⊂}{\subset}
\newunicodechar{⊃}{\supset}
\newunicodechar{⊆}{\subseteq}
\newunicodechar{⊇}{\supseteq}
\newunicodechar{⊄}{\nsubset}
\newunicodechar{⊅}{\nsupset}
\newunicodechar{⊈}{\nsubseteq}
\newunicodechar{⊉}{\nsupseteq}
\newunicodechar{∪}{\cup}
\newunicodechar{∩}{\cap}
\newunicodechar{∀}{\forall}
\newunicodechar{∃}{\exists}
\newunicodechar{∄}{\nexists}
\newunicodechar{∨}{\vee}
\newunicodechar{∧}{\wedge}
\newunicodechar{ℝ}{\mathbb{R}}
\newunicodechar{𝔽}{\mathbb{F}}
\newunicodechar{ℕ}{\mathbb{N}}
\newunicodechar{𝔼}{\mathbb{E}}
\newunicodechar{𝟙}{\mathds{1}}
\newunicodechar{𝒪}{\mathcal{O}}
\newunicodechar{ℤ}{\mathbb{Z}}
\newunicodechar{·}{\cdot}
\newunicodechar{∘}{\circ}
\newunicodechar{×}{\times}
\newunicodechar{↑}{\uparrow}
\newunicodechar{↓}{\downarrow}
\newunicodechar{→}{\rightarrow}
\newunicodechar{←}{\leftarrow}
\newunicodechar{↪}{\hookrightarrow}
\newunicodechar{⇒}{\Rightarrow}
\newunicodechar{⇐}{\Leftarrow}
\newunicodechar{↔}{\leftrightarrow}
\newunicodechar{⇔}{\Leftrightarrow}
\newunicodechar{↦}{\mapsto}
\newunicodechar{∅}{\varnothing}
\newunicodechar{∞}{\infty}
\newunicodechar{≅}{\cong}
\newunicodechar{≈}{\approx}
\newunicodechar{ℓ}{\ell}
\newunicodechar{⌊}{\lfloor}
\newunicodechar{⌋}{\rfloor}
\newunicodechar{⌈}{\lceil}
\newunicodechar{⌉}{\rceil}

\newunicodechar{α}{\alpha}
\newunicodechar{β}{\beta}
\newunicodechar{γ}{\gamma}
\newunicodechar{Γ}{\Gamma}
\newunicodechar{δ}{\delta}
\newunicodechar{Δ}{\Delta}
\newunicodechar{ε}{\varepsilon}
\newunicodechar{ζ}{\zeta}
\newunicodechar{η}{\eta}
\newunicodechar{θ}{\theta}
\newunicodechar{Θ}{\Theta}
\newunicodechar{ι}{\iota}
\newunicodechar{κ}{\kappa}
\newunicodechar{λ}{\lambda}
\newunicodechar{Λ}{\Lambda}
\newunicodechar{μ}{\mu}
\newunicodechar{ν}{\nu}
\newunicodechar{ξ}{\xi}
\newunicodechar{Ξ}{\Xi}
\newunicodechar{π}{\pi}
\newunicodechar{Π}{\Pi}
\newunicodechar{ρ}{\rho}
\newunicodechar{σ}{\sigma}
\newunicodechar{Σ}{\Sigma}
\newunicodechar{τ}{\tau}
\newunicodechar{υ}{\upsilon}
\newunicodechar{ϒ}{\Upsilon}
\newunicodechar{φ}{\phi}
\newunicodechar{ϕ}{\varphi}
\newunicodechar{Φ}{\Phi}
\newunicodechar{χ}{\chi}
\newunicodechar{ψ}{\psi}
\newunicodechar{Ψ}{\Psi}
\newunicodechar{ω}{\omega}
\newunicodechar{Ω}{\Omega}

\usepackage{xcolor}
\usepackage{xspace}
\usepackage{booktabs}
\usepackage{amsmath,amsthm,amssymb}
\usepackage{dsfont}
\usepackage{bm}
\usepackage{enumerate}
\usepackage{changepage}
\usepackage{thmtools}
\usepackage{thm-restate}
\usepackage{hyperref}
\usepackage{sidecap}
\usepackage{stackengine}
\usepackage{tikz}
\usetikzlibrary{shapes.geometric,decorations.pathmorphing,arrows.meta,calc,patterns}
\usepackage{pgfplots}
\def\U{\mathcal{U}}
\def\O{\mathcal{O}}
\newcommand{\Oh}[1]{\mathcal{O}\!\left( #1\right)}

\newcommand{\Om}[1]{\mathrm{\Omega}\!\left(#1\right)}

\def\sgauss{\scalebox{0.8}{SGAUSS}\xspace}
\def\dsimp{d_{\mathrm{simp}}}
\def\E{\mathbb{E}}

\def\F{\mathbb{Z}}
\def\fpr{\ensuremath{ϕ}}
\def\refrel#1#2#3{\stackrel{\text{#2 \ref{#1}}}{#3}}
\def\myparagraph#1{\smallskip\noindent\textbf{#1}}
\newcommand{\onebit}{$1^+$-bit\xspace}
\newcommand{\onebits}{$1\kern-0.15em^+$\kern-0.25em-bit\xspace}

\newenvironment{substatement}{\begin{enumerate}[{\normalfont\bfseries(a)}]}{\end{enumerate}}

\newcommand{\frage}[1]{{{\color{blue}[\sf#1]}}}
\newcommand{\pesa}[1]{\frage{PS:#1}}
\newcommand{\lhs}[1]{\frage{\color{red}LH:#1}}

\newcommand{\stwa}[1]{\frage{\color{magenta}SW:#1}}

\newcommand{\nat}{\ensuremath{\mathbb{N}}}
\newcommand{\Def}{:=}
\newcommand{\set}[1]{\left\{ #1\right\}}
\newcommand{\ceil}[1]{\left\lceil #1\right\rceil}
\newcommand{\floor}[1]{\lfloor #1\rfloor}

\newcommand{\eps}{\varepsilon}

\usepackage[linesnumbered,commentsnumbered,ruled,vlined]{algorithm2e}
\DontPrintSemicolon
\crefname{algocf}{Algorithm}{Algorithms}

\newif\ifconfVersion
\confVersionfalse

\newif\iftr
\trtrue

\iftr
\usepackage{array}
\fi

\newif\iffinal
\finaltrue

\iffinal
\renewcommand{\frage}[1]{}
\else
\fi

\hypersetup{
  colorlinks=true,
  pdftitle={Fast Succinct Retrieval and Approximate Membership using Ribbon},
  pdfauthor={Peter C. Dillinger, Lorenz Hübschle-Schneider, Peter Sanders, and Stefan Walzer},
  pdfsubject={}
}

\title{Fast Succinct Retrieval and Approximate Membership using Ribbon}
\author{Peter C. Dillinger}{Facebook}{peterd@fb.com}{}{}
\author{Lorenz Hübschle-Schneider}{Karlsruhe Institute of Technology, Germany}{huebschle@4z2.de}{}{}
\author{Peter Sanders}{Karlsruhe Institute of Technology, Germany}{sanders@kit.edu}{}{}
\author{Stefan Walzer}{Cologne University}{walzer@cs.uni-koeln.de}{}{DFG grant WA 5025/1-1.}
\authorrunning{P. C. Dillinger, L. Hübschle-Schneider, P. Sanders, and S. Walzer}
\Copyright{P. C. Dillinger, L. Hübschle-Schneider, P. Sanders, and S. Walzer}


\supplement{The code and scripts used for our experiments are available under
  a permissive license at \href{https://github.com/lorenzhs/BuRR}{github.com/lorenzhs/BuRR} and
  \href{https://github.com/lorenzhs/fastfilter_cpp}{github.com/lorenzhs/fastfilter\_cpp}.}

\begin{document}
\maketitle
%!TEX root=./main.tex
\begin{abstract}
  A \emph{retrieval} data structure for a static function
  $f:S\rightarrow \{0,1\}^r$ supports queries that return $f(x)$ for
  any $x ∈ S$.
  Retrieval data structures  can be used to
  implement a static approximate membership query data structure (AMQ), i.e., a Bloom filter alternative, with false positive rate~$2^{-r}$.
  The information-theoretic lower bound for both tasks is
  $r|S|$ bits.  While \emph{succinct} theoretical constructions using
  $(1+o(1))r|S|$ bits were known, these could not achieve very small
  overheads in practice because they have an unfavorable space--time
  tradeoff hidden in the asymptotic costs or because small overheads
  would only be reached for physically impossible input sizes.  With
  \emph{bumped ribbon retrieval (BuRR)}, we present the first
  practical succinct retrieval data structure.  In an
  extensive experimental evaluation BuRR achieves space overheads
  well below 1\,\% while being faster than most previously used retrieval data
  structures (typically with space overheads at least an order of
  magnitude larger) and faster than classical Bloom filters (with
  space overhead $\geq 44\,\%$). This efficiency, including favorable
  constants, stems from a combination of simplicity, word parallelism,
  and high locality.

  We additionally describe \emph{homogeneous ribbon filter AMQs}, which
  are even simpler and faster at the price of slightly larger space
  overhead.
\end{abstract}

\clearpage
% \input{intro-v4}

% \input{intro}
%!TEX root=./main.tex

% Hack: \vec stretches rows sometimes
\def\svec#1{\smash{\vec{#1}}}
\section{Introduction}\label{s:intro}

A \emph{retrieval data structure} (sometimes called ``static function'') represents a function $f: S → \{0,1\}^r$ for a set $S \subseteq \U$ of $n$ keys from a universe $\U$ and $r \in \nat$. A query for $x \in S$ must return $f(x)$, but a query for $x \in \U \setminus S$ may return any value from $\{0,1\}^r$.

The information-theoretic lower bound for the space needed by such a data structure is $nr$ bits in the general case.%
\footnote{If $f$ has low entropy then \emph{compressed static functions} \cite{HKP:CompressedFunction:2009,BelazzouguiV13,GOV:retrieval-Compressed:2020} can do better and even machine learning techniques might help, see
e.g.\ \cite{VKKM:Learned:2020}.}
This significantly undercuts the $Ω((\log |\U|+r)n)$ bits%
\footnote{This lower bound holds when $|\U| = Ω(n^{1+δ})$ for $δ > 0$. The general bound is $\log \binom{|\U|}{n}+nr$ bits.}
needed by a dictionary, which must return “\textsf{None}” for $x ∈ \U \setminus S$. The intuition is that dictionaries have to store $f ⊆ \U × \{0,1\}^r$ as a set of key-value pairs while retrieval data structures, surprisingly, need not store the keys.
We say a retrieval data structure using $s$ bits has \emph{(space) overhead} $\frac{s}{nr}-1$. 

The starting point for our contribution is a \emph{compact} retrieval data structure from \cite{DW:One-Block-per-Row:2019}\pesa{for a journal paper we could take Martin on board and claim that we are also subsuming that paper?}, i.e.\ one with overhead $𝒪(1)$. After minor improvements, we first obtain \emph{standard ribbon retrieval}. All theoretical analysis assumes computation on a word RAM with word size $\Om{\log n}$ and that hash functions behave like random functions.\footnote{This is a standard assumption in many papers and can also be justified by standard constructions \cite{DR:Applications:2009}.} The \emph{ribbon width} $w$ is a parameter that also plays a role in following variants.
\begin{theorem}[{similar to \cite{DW:One-Block-per-Row:2019}}]
  \label{thm:standard-ribbon}
  For any $ε > 0$, an $r$-bit standard ribbon retrieval data structure with ribbon width $w = \frac{\log n}{ε}$ has construction time $𝒪(n/ε²)$, query time $𝒪(r/ε)$ and overhead $𝒪(ε)$.
\end{theorem}
We then combine standard ribbon retrieval with the idea of \emph{bumping}, i.e., a convenient subset $S' ⊆ S$ of keys is handled in the first \emph{layer} of the data structure and the small rest is \emph{bumped} to recursively constructed subsequent layers. \pesa{Reformulated ``The resulting \emph{bumped ribbon retrieval} (BuRR) data structure can be \emph{succinct}, i.e.\ can be configured to have an overhead of $o(1)$'' because that is a bit imprecise -- even standard ribbon can be succinct its just more expensive. I have concentrated the succinctness discussion after the statement of the theorem.}The resulting \emph{bumped ribbon retrieval} (BuRR) data structure has much smaller overhead for any given ribbon width $w$.
\begin{restatable}{theorem}{bumpedRibbonTheorem}%\pesa{slightly reformulated}
  \label{thm:bumped-ribbon}
  An $r$-bit BuRR data structure
  with ribbon width $w = \O(\log n)$ and $r = \O(w)$ has expected construction time $\O(nw)$,
    space overhead $\O(\frac{\log w}{rw²})$,
    and query time $\O(1+\frac{rw}{\log n})$.
\end{restatable}
\def\AP#1{\textbf{(\kern-1pt\scalebox{0.8}{A}\scalebox{0.9}{#1}\kern-0.5pt)}}
\def\AP#1{}
\begin{table}[tbh]
  \centering
  \def\g{\color{darkgray}}
  \small
  \begin{tabular}{cccc@{\,}c@{}c@{\,\,}cl}
    \toprule
    && Year & $t_{\mathrm{construct}}$ & $t_{\mathrm{query}}$ & \footnotesize{\stackanchor{multiplicative}{overhead}} & shard size & Solver\\
    \midrule
    \AP{1}&\cite{LMSS:Efficient_Erasure:2001}&2001
    &$\O(n\log k)$ & $\O(\log k)^\dagger$ & $\frac{1}{k}$&–&peeling\\
    \AP{2}&\cite{P:An_Optimal:2009}&2009
    &$\O(n)$ & $\O(1)$ & $\bm{\O(\frac{\log \log n}{(\log n)^{1/2}})}$&$\sqrt{\log n}$&lookup table\\
    % &\cite{ADR:Experimental:2009}&2009
    % &$\O(nC²)$ & $\O(k)$ & $\mathrm{e}^{-k}+Ω(C^{-1/2})$\\
    \AP{3}&\cite{BPZ:Practical:2013}&2013
    &$\O(n)$ & $\O(1)$ & $0.2218$&–&peeling\\
    & \cite{BelazzouguiV13} & 2013 & $𝒪(n)$ & $𝒪(1)$ & $\bm{𝒪(\frac{\log² \log n}{r\log n})}$ & $𝒪(\frac{\log² \log n}{r\log n})$ & –\\
    \AP{8}&\cite{Sanders:Retrieval-FingerPrinting:2014}&2014
    &$\O(n)$ & $\O(1)$ & $Ω(1/r)$&$\O(1)$&sorting/sharding\\
    \AP{4}&\cite{Vigna:Fast-Scalable-Construction-of-Functions:2016}&2016
    &$\O(nC²)$ & $\O(1)$ & $0.024 + \O(\frac{\log n}{C})$&$C$&structured Gauss\\
    \begin{tikzpicture}[overlay]
      \node[font=\scriptsize,rotate=90](v) at (-2em,2pt) {Standard Ribbon};
      \draw[->] (v) -- (0,2pt);
    \end{tikzpicture}
    \AP{5}&\cite{DW:One-Block-per-Row:2019}&2019
    &$\O(n/ε²)$ & $\O(r/ε)$ & $ε$&–&Gauss\\
    \AP{6}&\cite{DW:One-Block-per-Row:2019}&2019
    &$\O(n/ε)$ & $\O(r)$ & $ε+\O(\frac{\log n}{n^{ε}})$&$n^ε$&Gauss\\
    \AP{7}&\cite{DW:Retrieval-log-extra-bits:2019}&2019
    &$\O(nC²)$ & $\O(r)$ & $\bm{\O(\frac{\log n}{C})}$&$C$&structured Gauss\\
    \AP{9}&\cite{W:SpatialCoupling:2021}&2021
    &$\O(nk)$ & $\O(k)$ & $(1+o_k(1))\mathrm{e}^{-k}$&–&peeling\\
    \midrule
    \multicolumn{3}{c}{BuRR}
    & $\O(nw)$  & $\O(1+\frac{rw}{\log n})$ & $\bm{\O(\frac{\log w}{rw²})}$&–&on-the-fly Gauss\\
    \multicolumn{3}{c}{$↪$ with $w = Θ(\log n)$: }
    & $\O(n\log n)$  & $\O(r)$ & $\bm{\O(\frac{\log \log n}{r\log²n})}$& –& on-the-fly Gauss\\
    \bottomrule
  \end{tabular}\\
  $\dagger$ Expected query time. Worst case query time is $\O(D)$.
  \caption{Performance of various $r$-bit retrieval data structures with $r = 𝒪(\log n)$. \textbf{Bold} overhead indicates that the data structure is (or can be configured to be) succinct. The parameters $k ∈ ℕ$ and $ε > 0$ are constants with respect to $n$. The parameter $C ∈ ℕ$ is typically $n^α$ for constant $α ∈ (0,1)$.}
  \label{tab:retrieval-comparison}

  % \pesa{Express result of A7 (and A4?) in terms of $w$ for better comparability?}\stwa{Don't like the idea, because these parameters tend to have vastly different values.}
  % \pedi{For brr row, would it be better to say query time is $\O(\frac{rw}{\log n} + 1)$ so that we aren't relying on $w = \O(\log n)$? Some other rows are assuming $r = \O(\log n)$ for $\O(1)$ query time. We could call these $\O(\frac{r}{\log n} + 1)$. Then I guess the construction times would need to be elaborated. Perhaps we can just list the assumptions on $r$ and $w$ (both logarithmic at worse) used in the reported times?}
  % \stwa{Tried to address comments.}
\end{table}

\def\mysmash#1{\scalebox{0.9}{$\smash{#1}$}}

\pesa{concentrated succinctness discussion here:}In particular, BuRR can be configured to
be \emph{succinct}, i.e., can be configured to have an overhead of $o(1)$ while retaining
constant access time for small $r$. Construction time is slightly superlinear.
Note that succinct retrieval data structures were known before, even with asymptotically optimal construction and query times of $𝒪(n)$ and $𝒪(1)$, respectively \cite{P:An_Optimal:2009,BelazzouguiV13}.
Seeing the advantages of BuRR requires a closer look. Details are given in \cref{s:related}, but the gist can be seen from \cref{tab:retrieval-comparison}: Among the \pesa{added word}previous succinct retrieval data structures (overheads set in bold font), \pesa{added word}only \cite{DW:Retrieval-log-extra-bits:2019} can achieve small overhead in a \pesa{slightly reformulated}\emph{tunable} way, i.e., independently of $n$ using an appropriate tuning parameter $C= \omega(\log n)$. However, this approach suffers from comparatively high constructions times. \cite{P:An_Optimal:2009} and \cite{BelazzouguiV13} are not tunable and only \emph{barely} succinct with significant overhead in practice. A quick calculation to illustrate\label{page:overhead-calculation-porat}: Neglecting the factors hidden by $𝒪$-notation, the overheads are \mysmash{\frac{\log \log n}{\sqrt{\log n}}} and \mysmash{\frac{\log²\log n}{r\log n}}, which is at least $75\%$ and $7\%$ for $r = 8$ and any $n ≤ 2^{64}$. A similar estimation for BuRR with $w = Θ(\log n)$ suggests an overhead of \mysmash{\frac{\log \log n}{r\log²n}} $≈ 0.1\%$ already for $r = 8$ and $n = 2^{24}$. Moreover, by tuning the ribbon width $w$, a wide range of trade-offs between small overhead and fast running times can be achieved.

Overall, we believe that asymptotic analyses struggle to tell the full story due to the extremely slow decay of some “$o(1)$” terms.
\stwa{Removed the added sentence: \pesa{It therefore seems that a tunability feature independent of $n$,
concrete nonasymptotic estimates of space overhead, and/or actual experiments
should be equally important criteria for compact data structures as only the asymptotic notion of succinctness.}Reason: We have no nonasymptotic estimates of space overhead (except for the experiments). And tunability of $w$ was mentioned a sentence ago.}
We therefore accompany the theoretical account with experiments comparing BuRR to other efficient (compact or succinct) retrieval data structures. We do this in the use case of data structures for approximate membership and also invite competitors not based on retrieval into the ring such as (blocked) Bloom filters and Cuckoo filters. % and Xor filters.

\myparagraph{Data structures for approximate membership.}
Retrieval data structures are an important basic tool for building compressed data structures. Perhaps the most widely used application is associating an $r$-bit fingerprint with each key from a set $S ⊆ \U$, which allows implementing an \emph{approximate  membership query data structure} (\emph{AMQ}, aka \emph{Bloom filter replacement} or simply \emph{filter}) that supports membership queries for $S$ with \emph{false positive rate} $ϕ = 2^{-r}$. A membership query for a key $x ∈ \U$ will simply compare
the fingerprint of $x$ with the result returned by the retrieval data structure for $x$. The values will be the same if $x ∈ S$.
Otherwise, they are the same only with probability $2^{-r}$.

In addition to the AMQs following from standard ribbon retrieval and BuRR, we also present homogeneous ribbon filters, which are not directly based on retrieval.

\begin{restatable}{theorem}{homogRibbonTheorem}
  \label{thm:homog-ribbon}
  Let $r ∈ ℕ$ and $ε ∈ (0,\frac{1}{2}]$. There is $w ∈ ℕ$ with $\frac{w}{\max(r,\log w)} = 𝒪(1/ε)$ such that the homogeneous ribbon filter with ribbon width $w$ has false positive rate $\fpr ≈ 2^{-r}$ and space overhead $𝒪(ε)$. On a word RAM with word size $≥ w$ expected construction time is $\O(n/ε)$ and query time is $\O(r)$.
\end{restatable}

\begin{figure}[hbt]
  \centering
\includegraphics[page=8,width=\columnwidth]{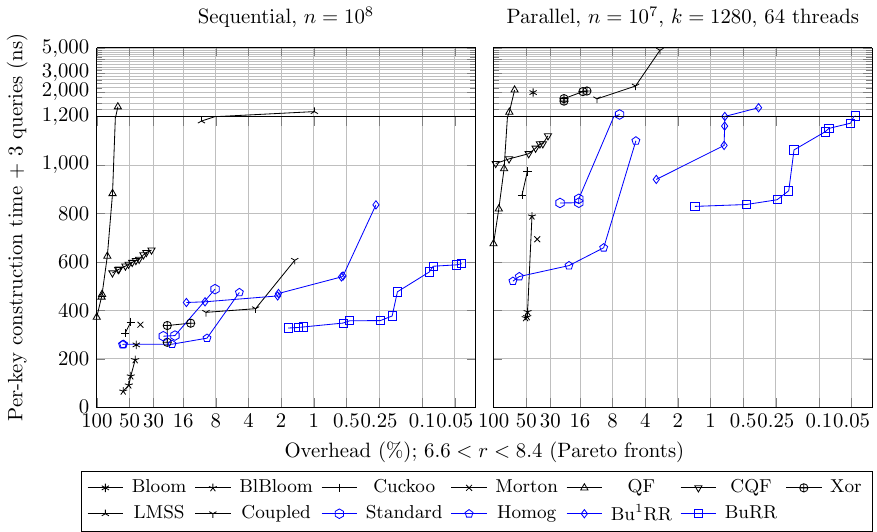}
  \caption{\label{fig:scatter1}%
    % \pesa{order legend so that similar techniques come together, e.g. Bloom, BlockedBloom, Cuckoo, Morton, QF, CQF, Xor, LMSS, Coupled, Standard, Homog, BuRR, Bu1RR}
    Performance--overhead trade-off for measured
    false-positive rate in 0.003--0.01 (i.e., $r\approx 8$),
    for different AMQs and inputs. Ribbon-based data structures are in {\color{blue}blue}. For each category of approaches, only
    variants are shown that are not Pareto-dominated by variants in the same
    category.  Sequential benchmarks use a single filter of size $n$
    while the parallel benchmark uses $1280$ filters of size $n$ and utilizes 64 cores.
    Logarithmic vertical axis above 1200\,ns.}
\end{figure}

\myparagraph{Experiments.}
\Cref{fig:scatter1} shows some of the results explained in detail later in the paper. In the depicted parallel setting, ribbon-based AMQs (blue) are the fastest static AMQs when an overhead less than $≈44\%$ is desired (where “fastest” considers a somewhat arbitrary weighting of construction and query times). The advantage is less pronounced in the sequential setting.

\myparagraph{Why care about space?} 
Especially in AMQ applications, retrieval data structures occupy a
considerable fraction of RAM in large server farms continuously drawing many megawatts of power. Even small
reductions (say 10\,\%) in their space consumption thus translate into
considerable cost savings. % savings in hardware costs.
Whether or not these space savings should be pursued at the price of increased access costs depends on the number of queries per second.
The lower the access frequency, the more worthwhile it is to occasionally spend increased access costs for a permanently lowered memory budget.
Since the false-positive rate also has an associated cost (e.g.\ additional accesses to disk or flash) it is also subject to tuning.
% Sophisticated implementations use multiple variants of compressed data
% structures at once based on known access frequencies of different
% parts of the database \cite{MRF14}.
The entire set of Pareto-optimal variants with respect to tradeoffs between space, access time, and FP rate is relevant for applications.
For instance, sophisticated implementations of LSM-trees use multiple variants of AMQs
at once based on known access frequencies \cite{DAI:Monkey:2017}. Similar ideas have been used in compressed data bases \cite{MRF14}.

\myparagraph{Outline.}
The paper is organized as follows (section numbers in parentheses). After important preliminaries (\ref{sec:prelim}), we explain our data structures and algorithms in broad strokes (\ref{sec:ribbon-intro}) and summarize our experimental findings (\ref{sec:exp-summary}).
\ifconfVersion
We then summarize related work (\ref{s:related}).
In the full paper \cite{dillinger2021fast}\pesa{check} we give a detailed theoretical analysis, extensively describe the design space of BuRR, and discuss additional experiments.
\else
We then fill in the details. We summarize related work (\ref{s:related}), including that from \cref{tab:retrieval-comparison}. In theory-oriented sections (\ref{s:ribbonRetrieval}-\ref{sec:bumped-ribbon}) we first analyze general aspects of the “ribbon” approach (\ref{s:ribbonRetrieval} and \ref{sec:ribbon-analysis}) and then prove theorems on standard ribbon (\ref{sec:architectures}), homogeneous ribbon (\ref{sec:homogeneous-ribbon}) and BuRR (\ref{sec:bumped-ribbon}). Algorithm engineers will be interested in a discussion of the design decisions that have to be made when implementing BuRR (\ref{s:designspace}) and a precise description of the many experiments we made (\ref{s:exp}).
\fi

\section{Linear Algebra Based Retrieval Data Structures and \texorpdfstring{\sgauss}{SGAUSS}}
\label{sec:prelim}

A simple, elegant and highly successful approach for compact and succinct retrieval 
uses linear algebra over the finite
field $\F_2=\{0,1\}$
\cite{DP:Succinct:2008,Vigna:Fast-Scalable-Construction-of-Functions:2016,ADR:Experimental:2009,P:An_Optimal:2009,CKRT:The_Bloomier:2004,BPZ:Practical:2013,DW:Retrieval-log-extra-bits:2019,DW:One-Block-per-Row:2019}.
Refer to
\cref{s:related} for a discussion of alternative and complementary techniques.

The train of thought is this:
A natural idea would be to have a hash function point to a location where the key's information is stored while the key itself need not be stored.
This fails because of hash collisions.
We therefore allow the information for each key to be dispersed over several locations.
%While this gets rid of the collision problem, it is not obvious that this idea is sound otherwise.
%% We get rid of the key-column by xoring
%% \emph{several} entries of the value-column -- the stored information is dispersed over the table and the hash function allows us to reconstruct it for each key $x\in S$.
%
Formally we store a table $Z ∈ \{0,1\}^{m × r}$ with $m\geq n$
entries of $r$ bits each and to define $f(x)$ as the bit-wise xor of a set of
table entries whose positions $h(x) ⊆ [m]$ are determined by a hash function $h$.\footnote{In this paper, $[k]$ can stand for $\{0,\ldots,k-1\}$ or $\{1,\ldots,k\}$ (depending on the context), and $a..b$ stands for $\{a,\ldots,b\}$.}
This
can be viewed as the matrix product $\svec{h}(x)Z$ where $\svec{h}(x) ∈ \{0,1\}^m$
is the characteristic (row)-vector of $h(x)$.
For given
$h$, the main task in building the data structure is to find the right
table entries such that $\svec{h}(x)Z = f(x)$ holds for every key $x$.
This is equivalent to
solving a system of linear equations $AZ=\mathbf{b}$ where $A=(\svec{h}(x))_{x
  ∈ S}\in \{0,1\}^{n\times m}$ and $\mathbf{b}=(f(x))_{x ∈ S}\in
\{0,1\}^{n\times r}$.
Note that
rows in the constraint matrix $A$ correspond to keys in the input set $S$.
In the following, we will thus switch between the terms ``row'' and ``key''
depending on which one is more natural in the given context.

An encouraging observation is that even for $m=n$, the system $AZ=\mathbf{b}$
is solvable with constant probability
 if the rows of $A$ are chosen uniformly at random
\cite{Cooper:Rank-Of-Random-Matrices:2000,P:An_Optimal:2009}.
% \pesa{check these references}. \stwa{looks good to me.}
With linear query time and cubic construction time, we can thus
achieve optimal space consumption. For a practically useful approach, however, we want
the $1$-entries in $\svec{h}(x)$ to be sparse and highly localized to allow 
cache-efficient queries in (near) constant time and
we want a (near) linear time algorithm for solving $AZ=\mathbf{b}$. This is possible if $m > n$.
%\stwa{slight rewording. added sentence for smoother transition to next paragraph.}

A particularly promising approach in this regard is \sgauss from \cite{DW:One-Block-per-Row:2019}
that chooses the $1$-entries within a narrow range.
% Higher locality is possible by choosing a larger number of nonzeroes in a narrow range.
% Specifically, the \sgauss approach \cite{DW:One-Block-per-Row:2019}
Specifically, it chooses $w$ random bits $c(x) ∈ \{0,1\}^w$ and a random \emph{starting position}
$s(x)\in[m-w-1]$, i.e.,
$\svec{h}(x) = 0^{s(x)-1}c(x)0^{m-s(x)-w+1}$.
For $m=(1+\eps)n$ some value $w=\Oh{\log(n)/\eps}$
suffices to make the system $AZ=\mathbf{b}$ solvable with high probability.
We call $w$ the \emph{ribbon width}
because after sorting the rows of $A$ by $s(x)$ we obtain a matrix which is not technically a band matrix, but which likely has all $1$-entries within a narrow \emph{ribbon} close to the diagonal.
% \pesa{sth about having $\svec{h}(x)[s(x)]=1$?}\lhs{I'll mention it in \cref{s:designspace}} \stwa{sounds good, commented out this discussion.}
The solution $Z$ can then be found in time
$\Oh{n/\eps^2}$ using Gaussian elimination \cite{DW:One-Block-per-Row:2019} and bit-parallel row operations; see also \cref{fig:ribbon-matrix}~(a).

%----------------------------------------------------------------------

\section{Ribbon Retrieval and Ribbon Filters}
\label{sec:ribbon-intro}
% \stwa{added subsection. ok?}\pesa{Just one numbered subsection is not good. But what about the new substructure with 3 subsections?}

We advance the linear algebra approach to the point where space overhead is almost eliminated while keeping or improving
the running times of previous constructions.

\begin{figure}[ht]
    \centering
    \begin{tabular}{c@{\!\!\!}cc@{\!\!\!}c}
    \textbf{(a)} && \textbf{(b)}\\[-15pt]
    &\includegraphics[page=1]{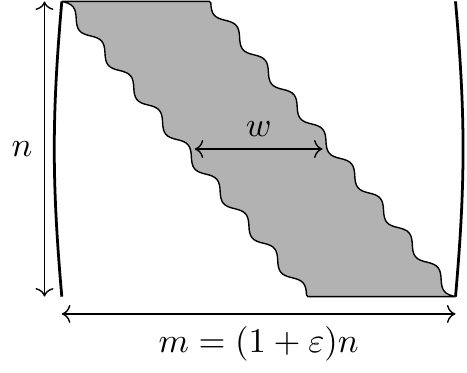}
    &&\includegraphics[page=4]{figures/MatrixPictures.pdf}
    \end{tabular}
    \caption[fragile]{\textbf{(a)} Typical shape of the random matrix $A$ with rows ${(\vec{h}(x))_{x ∈ S}}$ sorted by starting positions. The shaded “ribbon” region contains random bits. Gaussian elimination never causes any fill-in outside of the ribbon.\\%\lhs{change font back from ESA template}
    \textbf{(b)} Shape of the linear system ${M}$ in row echelon form maintained using Boolean banding on the fly.
    In gray we visualize the insertion of a key $x$ where \textsf{(i)} $\svec{h}(x)$ has its left-most $1$ in position $s(x) = 2$, \textsf{(ii)} after xoring the second row of $M$ to $\svec{h}(x)$, the left-most $1$ is in position $5$ and \textsf{(iii)} xoring the fifth row as well, the left-most $1$ is in position $6$. The resulting row fills the previously empty sixth row of $M$ and $f(x) ⊕ b₂ ⊕ b₅$ is added as right hand side.
    }
    \label{fig:ribbon-matrix}
\end{figure}

\myparagraph{Ribbon solving.}
Our first %added by stwa, commented out pesa: “minor but central”
contribution is a simple algorithm we could not resist to also call \emph{ribbon} as in \emph{\textbf{R}apid \textbf{I}ncremental \textbf{B}oolean \textbf{B}anding \textbf{ON} the fly}.
It maintains a system of linear equations in row echelon form as shown in \cref{fig:ribbon-matrix} \textbf{(b)}. It does so \emph{on-the-fly}, i.e.\ while equations arrive one by one in arbitrary order. For each index $i$ of a column there may be at most one equation that has its leftmost one in column $i$.
When an equation with row vector $a$ arrives and its slot is already taken by a row $a'$, then ribbon
performs the row operation $a\leftarrow a ⊕ a'$, which eliminates the $1$ in position $i$, and continues with the modified row. An invariant is that rows have all their nonzeroes in a range of size $w$,
which allows to process rows with a small number of bit-parallel word operations.
This insertion process is \emph{incremental} in that insertions do not modify existing rows. This improves performance and allows to cheaply roll back the most recent insertions which will be exploited below.
It is a non-trivial insight that the order in which equations are added does not significantly affect the expected number of row operations.

When all rows processed we perform back-substitution to compute the
solution matrix $Z$.  At least for small~$r$,
\emph{interleaved representation} of $Z$ works well,
where blocks of size $w\times r$ of $Z$ are stored column-wise.
A query for $x$ can then retrieve one bit of $f(x)$ at
a time by applying a population count instruction
to pieces of rows retrieved from at most two of these blocks.
This is particularly
advantageous for negative queries to AMQs (i.e.\ queries of elements not in the set), where only two bits need to be
retrieved on average.
\ifconfVersion
More details are given in the full paper.
\else
More details are given in \cref{s:ribbonRetrieval}.
\fi

%
% Our first contribution is \emph{Ribbon (\textbf{R}apid
%   \textbf{I}ncremental \textbf{B}oolean \textbf{B}anding \textbf{ON}
%   the fly)}, a new algorithm\pesa{discuss: I am not sure what the
%   increment compared to \cite{BGV:OnTheFlyGauss:2010} is. The theory
%   paper is not clear to me in that respect.} for solving such
% band-like systems.  Ribbon has  an \emph{incremental on-the-fly} property that is
% crucial for several subsequent constructions: Ribbon inserts one row $a$
% at a time into a system of equations in row echelon form.%
% \footnote{A matrix $A$ is in row echelon form if for every row $i$, if
%   $a_{ij}$ is the first nonzero in row $i$ then all coefficients
%   $a_{i'j'}$ with $i'\geq i$, $j'\leq j$, and $(i',j')\neq (i,j)$ are
%   zero. This form eventually allows solution of the system by
%   backsubstitution.}
% The algorithm tries to place $a$ at the position of its first nonzero.
% If this position is already taken by a row $a'$, Ribbon
% performs the row operation $a\leftarrow a+a'$.
% A crucial difference to Gaussian elimination is that row operations are only performed
% on the inserted row. This improves locality and  makes it possible to easily
% backtrack the solution process to also remove previously inserted
% rows. Ribbon maintains the invariants that rows have all their nonzeroes in a range of size $w$ and
% thus allows to process rows with a small number of bit-parallel word operations.
% More details are given in \cref{s:ribbonRetrieval}.

\subsection{Standard Ribbon}
When employing no further tricks, we obtain \emph{standard ribbon retrieval}, which is essentially the same data structure as in \cite{DW:One-Block-per-Row:2019} except with a different solver that is faster by noticeable constant factors.
A problem is that $w$ has to become impractically large when $n$ is large and $\eps$ is small. For example, in our experiments the smallest overhead we
could achieve for $n=10^6$ and (already quite expensive) $w=128$ is
around $3.3\,\%$ (for construction success rate $50\%$).
To some degree this can be mitigated by sharding techniques \cite{W:Thesis}, but in this paper we pursue a more ambitious route.
% \stwa{rewritten: separate paragraph, better glue to next paragraph, connection to ESA paper. ok?}
% \pesa{add a sentence on mitigation of that
%   problem by sharding?}\lhs{This paragraph ends abruptly. Suggestion: To mitigate this, we introduce BuRR, which gives the solver more flexibility by allowing it to remove ranges of problematic rows.}

% The resulting \emph{standard ribbon retrieval} has the problem that
% $w$ has to become impractically large when $n$ is large and $\eps$
% is small. For example, in our experiments the smallest overhead we
% could achieve for $n=10^6$ and (already quite expensive) $w=128$ were
% around $5.8\,\%$.  \pesa{add a sentence on mitigation of that
%   problem by sharding?}\lhs{This paragraph ends abruptly. Suggestion: To mitigate this, we introduce BuRR, which gives the solver more flexibility by allowing it to remove ranges of problematic rows.}

\subsection{Bumped Ribbon Retrieval}
Our main contribution is \emph{bumped ribbon retrieval (BuRR)}, which
reduces the required ribbon width to a constant that only depends on the targeted space efficiency.
BuRR is based on two ideas.

\myparagraph{Bumping.}
The ribbon solving approach
manages to insert most rows (representing most keys of $S$) even when
$w$ is small.  Thus, by eliminating those rows/keys that cause a
linear dependency, we obtain a compact retrieval data structure for
a large subset of $S$. The remaining keys are \emph{bumped}, meaning they are handled by a fallback
data structure which, by recursion, can be a BuRR data structure
again. We show that only $𝒪(\frac{n \log w}{w})$ keys need to be bumped in expectation.  Thus, after
a constant number of layers (we use $4$), a less ambitious retrieval data structure can be used to handle
the few remaining keys without bumping.

The main challenge is that we need additional metadata to encode which keys are bumped.
The basic \emph{bumped retrieval} approach is adopted from the
updateable retrieval data structure \emph{filtered retrieval (FiRe)}
\cite{Sanders:Retrieval-FingerPrinting:2014}. To shrink the input size by a moderate
constant factor, FiRe needs a constant number of bits per key
(around 4). This leads to very high space overhead for small
$r$.
A crucial observation
for BuRR is that bumping can be done with coarser than per-key granularity.
We will bump keys based on their starting position and say \emph{position~$i$ is bumped} to indicate that all keys with $s(x) = i$ are bumped.
Bumping by position is sufficient because linear dependencies
in $A$ are largely unrelated to the actual bit patterns~$c(x)$ but
mostly caused by fluctuations in the number
of keys mapped to different parts of the matrix~$A$.
By selectively bumping
ranges of positions in overloaded parts of the system, we can
obtain a solvable system. Furthermore, our analysis shows that we can
drastically limit the spectrum of possible bumping ranges, see below.

\myparagraph{Overloading.}
Besides metadata, space overhead results from the $m-n+n_b$ excess
slots of the table where $n_b$ is the number of bumped keys.
Trying
out possible values of $\eps = \frac{m-n}{n} > 0$ one sees that the overhead due to
excess slots is always $\Om{1/w}$ and will thus dominate the
overhead due to metadata.  However, we show that by choosing $\eps<0$ (of order $-ε = 𝒪(\frac{\log w}{w})$), i.e., by
\emph{overloading} the table, we can almost completely eliminate
excess table slots so that the minuscule amount of metadata becomes
the dominant remaining overhead.
There are many ways to decide and encode which keys are bumped.
Here, we outline a simple variant that achieves very good performance in
practice and is a generalization of the theoretically analyzed
approach.
\ifconfVersion
We expand on the much larger design space of BuRR in the full paper.
\else
We expand on the much larger design space of BuRR in \cref{s:designspace}.
\fi

\myparagraph{Deciding what to bump.}
We subdivide the possible starting positions
into \emph{buckets} of width $b=\Oh{w^2/\log w}$ and allow to bump a
single initial range of each bucket.
% We first scan $S$ and generate a sequence of pairs $\mathtt{S}=\langle(\mathrm{MHC}(x),f(x)): x\in S\rangle$.
% where $\mathrm{MHC}(x)$ is an $\Oh{\log n}$ bit  \emph{master hash code} uniquely identifying key $x$.
% All further hash values are subsequently derived from the MHCs.
The keys (or more precisely pairs of hashes and the value to be retrieved) are sorted according to the bucket addressed by the starting
position $s(x)$. We use a fast in-place
integer sorter for this purpose \cite{AWFS20}.  Then buckets are processed
one after the other from \emph{left to right}. Within a bucket, however,
keys are inserted into the row echelon form from \emph{right to
  left}.
The reason for this is that insertions of the previous bucket may have “spilled over” causing additional load on the left of the bucket – an issue we wish to confront as late as possible. See also \cref{fig:fillrtlnew}.
% \begin{SCfigure}
%   \centering
%   \includegraphics{../fig_burr}
%   \caption{Insertion of a key $x$ into the system.  Already filled rows
%     are marked light gray for keys belonging to the previous bucket, and dark
%     gray for those inserted so far from the current bucket.  White boxes
%     preceding stored coefficients mark coefficients that were reduced to 0 and
%     show the keys' starting positions.  The reduced coefficients and function
%     value of $x$ are eventually stored in the blue-shaded slot.  Two empty slots
%     remain in the bucket.  \label{fig:fillrtl}}
% \end{SCfigure}

\begin{figure}[htbp]
  \centering
\includegraphics[page=5,scale=0.9]{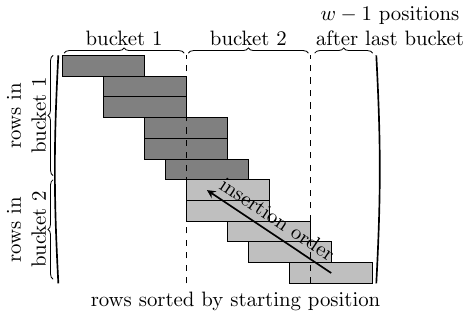}
\raisebox{1.5cm}{\scalebox{1.5}{$\rightsquigarrow$}}\hspace{5mm}
\includegraphics[page=6,scale=0.9]{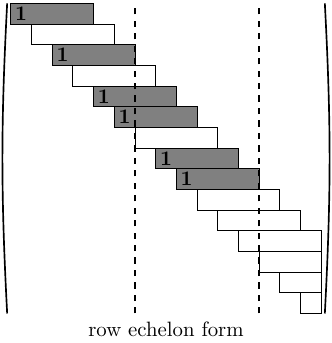}
  \caption{Illustration of BuRR construction with $n = 11$ keys, $m = 2b+w-1 = 15$ table positions, ribbon width $w = 4$ and bucket size $b = 6$. Keys of the first bucket were successfully inserted into row echelon form with two insertions “overflowing” into the second bucket. Insertions of the second bucket's rows will be attempted next, in the indicated order.}
  \label{fig:fillrtlnew}
\end{figure}

If all keys of a bucket can be successfully inserted, no
keys of the bucket are bumped. Otherwise, suppose the first failed insertion for a bucket $[i,i+b)$ concerns a key where $s(x) = i+k$ is the $k$-th position of the bucket. We \emph{could} decide to bump all keys $x'$ of the bucket with $s(x') ≤ i+k$, which would require storing the \emph{threshold} $k$ using $𝒪(\log w)$ bits and which would yield an overhead of $𝒪(\log²(w)/w²)$ due to metadata. \emph{Instead}, to reduce this overhead to $𝒪(\log(w)/w²)$, we only allow a constant number of threshold values. This means that we find the smallest threshold value $t$ with $t ≥ k$ representable by metadata and bump all keys $x'$ with $s(x') ≤ i+t$. This requires rolling back the insertions of keys $x'$ with $s(x') ∈ [k,t]$ by clearing the most recently populated rows from the row echelon form.
One good
compromise between space and speed stores 2 bits per bucket encoding
the threshold values $\{0,\ell, u, b\}$, for suitable $\ell$ and $u$. The special case $\ell=u = \frac{3}{8}w$ is
used in our analysis. Another slightly more compact variant ``\onebit'' stores one
bit encoding threshold values from the set $\{0,t\}$, for a suitable $t$, and additionally stores a hash table of
exceptions for thresholds $>t$.

\myparagraph{Running times.}
With these ingredients we obtain \cref{thm:bumped-ribbon} stated on page \pageref{thm:bumped-ribbon}.
It implies constant query time\footnote{%
  It should be noted that the proof invokes a lookup table in one case to speed up the computation of a matrix vector product. In \cref{s:related}, we argue that lookup tables should be avoided in practice. Technically, our \emph{implementation} therefore has a query time of $𝒪(r)$.%
} if $rw=\Oh{\log n}$ and linear
construction time if $w\in\Oh{1}$. For wider ribbons,
construction time is slightly superlinear. However, in practice this does
not necessarily mean that BuRR is slower than other approaches with
asymptotically better bounds as the factor $w$ involves operations
with very high locality. An analysis in the external memory model
reveals that BuRR construction is possible with a single scan of the
input and integer sorting of $n$ objects of size $\Oh{\log n}$
bits, see
\ifconfVersion
the full paper for details.
\else
\cref{ss:parallel}.
\fi

\subsection{Homogeneous Ribbon Filter}
%\stwa{moved this paragraph so the proof outline can refer to it.}
For the application of ribbon
% \lhs{Ribbon or ribbon? Use uniformly}\pesa{in the intro I now use capital R if Ribbon stands alone but small r if it is combined with user things, e.g. standard, retrieval, filter,\ldots I am uncertain whats linguistically best.}\stwa{We're on the safe side with lower case everywhere.}
to AMQs, we can also compute a \emph{uniformly random}
solution of the \emph{homogeneous} equation system $AZ=0$, i.e., we
compute a retrieval data structure that will retrieve $0^r$ for all
keys of $S$ but is unlikely to produce $0^r$ for other
inputs. Since $AZ=0$ is always solvable, there is no need for
bumping. The crux is that the false positive rate is no longer $2^{-r}$ but higher.
\ifconfVersion
In the full paper
\else
In \cref{sec:homogeneous-ribbon}
\fi
we show that with table size
$m=(1+\eps)n$ and $ε = Ω(\frac{\max(r,\log w)}{w})$ the difference is negligible, thereby showing \cref{thm:homog-ribbon}.
Homogeneous ribbon AMQs are simpler and faster
than BuRR but have higher space overhead. Our experiments indicate that
together, BuRR and homogeneous ribbon AMQs cover a large part of the
best tradeoffs for static AMQs.

\subsection{Analysis outline}
To get an intuition for the relevant linear systems, it is useful to consider two simplifications. First, assume that $\vec{h}(x)$ contains a block of $w$ uniformly random \emph{real} numbers from $[0,1]$ rather than $w$ random bits. Secondly, assume that we sort the rows by starting position and use Gaussian elimination rather than ribbon to produce a row echelon form. In \cref{fig:simplified-diagonal}\,\textbf{(a)} we illustrate for such a matrix with $×$-marks where the pivots would be placed and in yellow the entries that are eliminated (with one row operation each); both with probability $1$, i.e.\ barring coincidences where a row operation eliminates more than one entry. The $×$-marks trace a diagonal through the matrix except that the green column and the red row are skipped because the end of the (gray) area of nonzeroes is reached. “Column failures” correspond to free variables and therefore unused space. “Row failures” correspond to linearly dependent equations and therefore failed insertions. This view remains largely intact when handling \emph{Boolean} equations in \emph{arbitrary} order except that the \emph{ribbon diagonal}, which we introduce as an analogue to the trace of pivot positions, has a more abstract meaning and probabilistically suffers from row and column failures depending on its \emph{distance} to the ribbon border.

\begin{figure}[bth]
    \centering
    \begin{tabular}{c@{\!}cc@{}c}
    \textbf{(a)} && \textbf{(b)}\\[-15pt]
    &\includegraphics[width=0.25\textwidth,page=1]{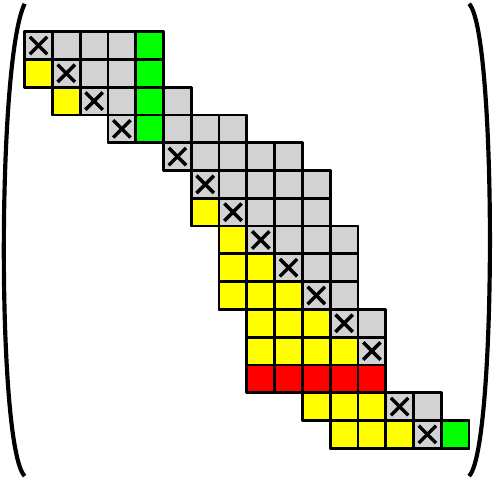}
    &&
    \includegraphics[width=0.25\textwidth,page=2]{figures/ribbon-diagonal.pdf}
    \raisebox{5em}{\Large $\longrightarrow$}
    \raisebox{1em}{\includegraphics[width=0.25\textwidth,page=3]{figures/ribbon-diagonal.pdf}}
    \end{tabular}
    \caption[fragile]{\textbf{(a)}
  The simplified ribbon diagonal (made up of $×$-marks) passing through $A$.\\
    \textbf{(b)} The idea of BuRR: When starting with an “overloaded” linear system and removing sets of rows strategically, we can often ensure that the ribbon diagonal does not collide with the ribbon border (except possibly in the beginning and the end).}
  \label{fig:simplified-diagonal}
    \label{fig:burr-idea}
\end{figure}

The idea of standard ribbon is to give the gray ribbon area an expected slope of less than $1$ such that row failures are unlikely. BuRR, as illustrated in \cref{fig:burr-idea}\,\textbf{(b)} largely avoids both failure types by using a slope bigger than $1$ but removing ranges of rows in strategic positions. Homogeneous ribbon filters, despite being the simplest approach, have the most subtle analysis as both failure types are allowed to occur. While row failures cannot cause outright construction failure, they are linked to a compromised false positive rate in a non-trivial way.
Our proofs involve mostly simple techniques as would be used in the analysis of linear probing, which is unsurprising given that \cite{DW:One-Block-per-Row:2019} has already established a connection to Robin Hood hashing. We also profit from queuing theory via results we import\pesa{reformulate for the journal version?} from \cite{DW:One-Block-per-Row:2019}.

\subsection{Further results}
We have several further results around variants of BuRR that we summarize here.
\ifconfVersion
See the full paper for detail.
\fi

Perhaps most interesting is {\bf bump-once ribbon retrieval
  (Bu${}^1$RR)}, which improves~the worst-case query time by
guaranteeing that each key can be retrieved from one out of two
layers -- its \emph{primary layer} or the next one.
The primary layer of the keys is now distributed over all layers (except for the last).
When building a layer, the keys bumped from the previous layer
are inserted into the row echelon form first.
The layer sizes have to be chosen in such a way that no bumping is
needed for these keys with high probability.
Only then
are the keys with the current layer as their primary layer
inserted -- now allowing bumping.
\ifconfVersion
\else
See \cref{ss:twoLayer} for details.
\fi

For building large retrieval data structures, {\bf parallel
  construction} is important. Doing this directly is difficult for
ribbon retrieval since there is no efficient way to parallelize
back-substitutions. However, we can partition the equation system into
parts that can be solved independently by bumping $w$ consecutive
positions. Note that this can be done transparently to the query
algorithm by using the bumping mechanism that is present anyway.
\ifconfVersion
\else
See \cref{ss:parallel} for details.
\fi

For large $r$, we accelerate queries by working with {\bf sparse bit
  patterns} that set only a small fraction of the $w$ bits in the
window used for BuRR.  In some
sense, we are covering here the middle ground between ribbon and
spatial coupling \cite{W:SpatialCoupling:2021}. Experiments indicate
that setting 8 out of 64 bits indeed speeds up queries for
$r\in\{8,16\}$ at the price of increased (but still small)
overhead. Analysis and further exploration of this middle ground may
be an interesting area for future work.

% \pesa{removed paragraph on implementation with other base numbers.}
%% From an algebraic point of view, retrieval based on linear algebra can
%% be done on arbitrary finite fields. By considering the prime
%% factorization of a base $u$, we can use this to encode numbers from
%% $[u]$ in a retrieval data structure. With BuRR, this would result in a
%% {\bf succinct representation of base $u$ numbers}. Doing arithmetics modulo small prime numbers in a word-parallel
%% way comes at a cost in running time, but it has been done in the context of retrieval. \cite{Vigna:Fast-Scalable-Construction-of-Functions:2016,GOV:retrieval-Compressed:2020}.
%% \pesa{which paper are you referring to? also in the SEA 2016 (or the corresponding Arxiv paper), I cannot find such a place.}\stwa{Search for add\_mod3\_step in these papers.}

% \pesa{omitted paragraph on smash for now} \stwa{gone from theory now as well ✓}

%% For {\bf small $n$}, the way we currently describe ribbon retrieval\pesa{check}
%% has an additive space overhead of $w-1$ table entries because keys
%% with $s(x)=(1+\eps)n-1$ might be placed anywhere in the next $w$
%% positions of the table. In
%% \cref{sec:architectures} we
%% \pesa{todo. Stefan/PeterD can you check and complete this paragraph?}

%----------------------------------------------------------------------

\section{Summary of Experimental Findings}
\label{sec:exp-summary}

We performed extensive experiments to evaluate our ribbon-based data structures and competitors. We summarize our findings here with details provided in
\ifconfVersion
the full paper.

\myparagraph{Implementation Details.}  We implemented BuRR in C++
using template parameters that allow us to navigate a large part of
the design space mapped in the full paper.
Input keys themselves are only hashed once to a 64-bit \emph{master-hash-code (MHC)}
that is subsequently used when further hash values are needed.
For this, fast linear congruential mapping is used.
\stwa{Removed the following sentence, because we have not introduced alternatives to these points: “When not
otherwise mentioned, our default configuration is BuRR with
left-to-right processing of buckets, aggressive right-to-left
insertion within a bucket, and threshold-based bumping.”}
The table is stored in an \emph{interleaved fashion}, i.e., it is organized as $rm/w$ words of $w$ bits each
where word $i$ represents bit $i\bmod r$ of $w$ subsequent table
entries.  This organization allows the extraction
of one retrieved bit from two adjoining machine words using
population-count instructions.  Interleaved representation is
advantageous for uses of BuRR as an AMQ data structure since a negative
query only has to extract two bits in expectation. Moreover, the
implementation directly works for any value of $r$.
The default data
structure has four layers, the last of which uses
$w'\Def\min(w,64)$ and $\eps\geq0$, where $\eps$ is increased in increments of 0.05 until no keys are bumped.  For \onebit, we
choose $t \Def \lceil -2\eps b + \sqrt{b/(1+\eps)}/2 \rceil$ and
$\eps \Def -2/3 \cdot w/(4b+w)$.  For 2-bit, parameter tuning showed that
 $\ell\Def\ceil{(0.13-\eps/2)b}, u\Def\ceil{(0.3-\eps/2)b}$, and
$\eps \Def -3/w$~work well for $w=32$; for $w \geq 64$, we use
$\ell=\ceil{(0.09-3\eps/4)b}$, $u=\ceil{(0.22-1.3\eps)b}$, and $\eps \Def -4/w$.

In addition, there is a prototypical implementation of Bu$^1$RR from
\cite{Ribbon:Arxiv:2021}. Both BuRR and Bu$^1$RR build
on the same software for ribbon solving from
\cite{Ribbon:Arxiv:2021}. %\pesa{any more details needed?}
For validation we extend the experimental setup used for Cuckoo and Xor
filters~\cite{Lemire:fastfilter:2020}, with our code and scripts available at
\href{https://github.com/lorenzhs/fastfilter_cpp}{github.com/lorenzhs/fastfilter\_cpp} and \href{https://github.com/lorenzhs/BuRR}{github.com/lorenzhs/BuRR}.

\myparagraph{Experimental Setup.}\label{exp:hw} %
All experiments were run on a machine with an AMD EPYC 7702 processor with 64
cores, a base clock speed of 2.0\,GHz, and a maximum boost clock speed of
3.35GHz.  The machine is equipped with 1\,TiB of DDR4-3200 main memory and runs
Ubuntu 20.04.  We use \texttt{clang++} 11.0 with
optimization flags \texttt{-O3 -march=native}.  During sequential experiments,
only a single core was used at any time to minimize interference.

We looked at different input sizes $n\in\set{10^6,10^7,10^8}$.  Like most
studies in this area, we first look at a {\bf sequential} workload on a powerful
processor with a considerable number of cores.  However, this seems unrealistic
since in most applications, one would not let most cores lay bare but use
them. Unless these cores have a workload with very high locality this would have
a considerable effect on the performance of the AMQs. We therefore also look at
a scenario that might be the opposite extreme to a sequential unloaded
setting. We run the benchmarks on all available hardware threads in {\bf
  parallel}.  Construction builds many AMQs of size $n$ in parallel. Queries
select AMQs randomly.
% \pesa{or was it round-robin?}\lhs{no, random is correct}
This emulates a large AMQ that is parallelized using sharding and
puts very high
load on the memory system.

\myparagraph{Experimental Results.}
\else
\cref{s:exp}.
\fi
Two preliminary remarks are in order: Firstly, since every retrieval data structure can be used as a filter but not vice versa, our experiments are for filters, which admits a larger number of competitors. Secondly, to reduce complexity (for now), our speed ranking
considers the sum of construction time per key and three query times.\footnote{%
Queries measured in three settings: Positive keys, negative keys and a mixed data set (50\,\% chance of being positive). The latter is not an average of the first two due to branch mispredictions.}
% sw: elaborated for clarity.
% For each key, one negative query, one positive query, and one
% with a 50\,\% chance of being positive).

% In \cref{s:exp}, we compare several implementation variants of
% BuRR and homogeneous ribbon filters with a wide range of existing AMQ
% filters.
\ifconfVersion
\myparagraph{Space Overhead of BuRR}
\pgfplotscreateplotcyclelist{eps}{
  red, mark=* \\% uncompressed B=128
  black, mark=* \\% uncompressed B=256
  blue, mark=square \\% 1-Bit B=64
  red, mark=square \\% 1-Bit B=128
  blue, mark=asterisk \\% 2-Bit B=64
  red, mark=asterisk \\% 2-Bit B=128
}
\pgfplotscreateplotcyclelist{eps2}{
%  green!50!black, mark=* \\% uncompressed B=512
  black, mark=* \\% uncompressed B=256
  black, mark=square \\% 1-Bit B=256
  black, mark=asterisk \\% 2-Bit B=256
  red, mark=* \\% uncompressed B=128
  red, mark=square \\% 1-Bit B=128
  red, mark=asterisk \\% 2-Bit B=128
  blue, mark=square \\% 1-Bit B=64
  blue, mark=asterisk \\% 2-Bit B=64
}
%
% IMPORT-DATA epstuning ../../ribbon/logs/epstuning_v5
% SQL CREATE TEMPORARY TABLE modename (mode INT, name VARCHAR);
% SQL INSERT INTO modename VALUES (0, "plain"), (1, "\onebits"), (2, "2-bit")
\begin{figure}\centering
  \begin{tikzpicture}
    \begin{axis}[
      xlabel={Overloading factor $\eps = 1 - m/n$},
      ylabel={Fraction of empty slots $e$},
      width=\textwidth,%16.5cm,
      height=10cm,
      ymode = log,
      xmin=-20,
      xmax=0,
      mark size=1.5pt,
      legend style={font=\small},
      legend pos=north west,
      cycle list name=eps2,
      every axis legend/.append style={nodes={left}},
      xticklabel={\pgfmathprintnumber\tick\%},
      minor tick num=1,
      ]
%% MULTIPLOT(mode, B|ptitle) SELECT 100*eps AS x, avg(tlemptyfrac) AS y,
%% name || ', $b\!=\!' || B || '$' as ptitle, MULTIPLOT
%% FROM epstuning NATURAL JOIN modename WHERE mode < 3 AND sparse=0 AND L=64 AND B<512
%% GROUP BY MULTIPLOT,x ORDER BY B DESC, mode, x
\addplot coordinates { (-12.5,0.00344107) (-12.1875,0.00305143) (-11.875,0.00270312) (-11.5625,0.00239274) (-11.25,0.00211286) (-10.9375,0.0018644) (-10.625,0.00164167) (-10.3125,0.00144626) (-10.0,0.00127343) (-9.6875,0.00111884) (-9.375,0.000982776) (-9.0625,0.00086344) (-8.75,0.000757241) (-8.4375,0.000665932) (-8.125,0.000585121) (-7.8125,0.000516006) (-7.5,0.000454459) (-7.1875,0.000403004) (-6.875,0.000358929) (-6.5625,0.000323967) (-6.25,0.000295019) (-5.9375,0.000276897) (-5.625,0.000267358) (-5.3125,0.000268961) (-5.0,0.000283565) (-4.6875,0.00031269) (-4.375,0.0003609) (-4.0625,0.000433411) (-3.75,0.000537686) (-3.4375,0.00067923) (-3.125,0.000868863) (-2.8125,0.00112058) (-2.5,0.0014437) (-2.1875,0.00185595) (-1.875,0.00237278) (-1.5625,0.00301205) (-1.25,0.00379161) (-0.9375,0.00472636) (-0.625,0.00582658) (-0.3125,0.00710256) (0.0,0.00855855) };
\addlegendentry{plain, $b\!=\!256$};
\addplot coordinates { (-7.84314,0.025898) (-7.64706,0.0235997) (-7.45098,0.0214507) (-7.2549,0.0194218) (-7.05882,0.0175356) (-6.86274,0.0157801) (-6.66667,0.0141471) (-6.47059,0.012635) (-6.27451,0.0112538) (-6.07843,0.0100022) (-5.88235,0.00885793) (-5.68627,0.00783061) (-5.4902,0.00691131) (-5.29412,0.0060919) (-5.09804,0.00538075) (-4.90196,0.00475378) (-4.70588,0.00423376) (-4.5098,0.00377811) (-4.31373,0.00341486) (-4.11765,0.00312841) (-3.92157,0.00290111) (-3.72549,0.00274523) (-3.52941,0.00264115) (-3.33333,0.00259433) (-3.13725,0.0025985) (-2.94118,0.00265606) (-2.7451,0.00275117) (-2.54902,0.00288103) (-2.35294,0.00306267) (-2.15686,0.00328121) (-1.96078,0.00354027) (-1.76471,0.00383912) (-1.56863,0.00418917) (-1.37255,0.00458848) (-1.17647,0.00503615) (-0.980392,0.00554291) (-0.784314,0.0061064) (-0.588235,0.00673203) (-0.392157,0.00742479) (-0.196078,0.00818777) (0.0,0.00892927) };
\addlegendentry{\onebits, $b\!=\!256$};
\addplot coordinates { (-12.5,0.0474202) (-12.1875,0.0435703) (-11.875,0.0399044) (-11.5625,0.0380378) (-11.25,0.0346695) (-10.9375,0.0329653) (-10.625,0.0298657) (-10.3125,0.0269841) (-10.0,0.0255071) (-9.6875,0.022913) (-9.375,0.0216163) (-9.0625,0.0193261) (-8.75,0.0172274) (-8.4375,0.0161603) (-8.125,0.0143429) (-7.8125,0.0134387) (-7.5,0.0118986) (-7.1875,0.0105289) (-6.875,0.00982617) (-6.5625,0.00871386) (-6.25,0.00814382) (-5.9375,0.00723434) (-5.625,0.0064863) (-5.3125,0.00606431) (-5.0,0.00549379) (-4.6875,0.00524442) (-4.375,0.00485587) (-4.0625,0.00458073) (-3.75,0.00453195) (-3.4375,0.00442091) (-3.125,0.0045383) (-2.8125,0.00461032) (-2.5,0.0048019) (-2.1875,0.00518652) (-1.875,0.00559684) (-1.5625,0.00620988) (-1.25,0.00686461) (-0.9375,0.00766728) (-0.625,0.00868672) (-0.3125,0.00978406) (0.0,0.0111192) };
\addlegendentry{2-bit, $b\!=\!256$};
\addplot coordinates { (-18.75,9.76875e-07) (-18.125,7.6e-07) (-17.5,6.325e-07) (-16.875,5.4375e-07) (-16.25,4.2875e-07) (-15.625,3.4625e-07) (-15.0,2.53125e-07) (-14.375,2.2625e-07) (-13.75,2.05e-07) (-13.125,2.0625e-07) (-12.5,1.9625e-07) (-12.1875,1.95625e-07) (-11.875,2.06875e-07) (-11.5625,2.20625e-07) (-11.25,2.40625e-07) (-10.9375,3.0e-07) (-10.625,3.5125e-07) (-10.3125,4.49375e-07) (-10.0,5.73125e-07) (-9.6875,7.60625e-07) (-9.375,1.0275e-06) (-9.0625,1.44438e-06) (-8.75,2.01938e-06) (-8.4375,2.81e-06) (-8.125,3.96e-06) (-7.8125,5.70375e-06) (-7.5,8.16e-06) (-7.1875,1.1675e-05) (-6.875,1.6445e-05) (-6.5625,2.34469e-05) (-6.25,3.30681e-05) (-5.9375,4.65306e-05) (-5.625,6.52725e-05) (-5.3125,9.10244e-05) (-5.0,0.000126447) (-4.6875,0.000173787) (-4.375,0.000238588) (-4.0625,0.000325136) (-3.75,0.000441188) (-3.4375,0.000593733) (-3.125,0.000794131) (-2.8125,0.0010523) (-2.5,0.00138282) (-2.1875,0.00180172) (-1.875,0.00232416) (-1.5625,0.00296993) (-1.25,0.00375275) (-0.9375,0.00469249) (-0.625,0.00579612) (-0.3125,0.00707529) (0.0,0.008536) };
\addlegendentry{plain, $b\!=\!128$};
\addplot coordinates { (-14.8148,0.0216292) (-14.4444,0.0187468) (-14.0741,0.0161362) (-13.7037,0.0137642) (-13.3333,0.011647) (-12.963,0.00976829) (-12.5926,0.00811006) (-12.2222,0.00665114) (-11.8519,0.00539643) (-11.4815,0.00432627) (-11.1111,0.00341416) (-10.7407,0.0026554) (-10.3704,0.00203986) (-10.0,0.00153513) (-9.62963,0.00114352) (-9.25926,0.000846284) (-8.88889,0.000623753) (-8.51852,0.000470966) (-8.14815,0.000365955) (-7.77778,0.000300624) (-7.40741,0.000265219) (-7.03704,0.000250851) (-6.66667,0.000253808) (-6.2963,0.000269763) (-5.92593,0.000338304) (-5.55556,0.000384674) (-5.18518,0.000443703) (-4.81481,0.000527281) (-4.44444,0.00063379) (-4.07407,0.000772338) (-3.7037,0.000953613) (-3.33333,0.00118595) (-2.96296,0.00148201) (-2.59259,0.00186467) (-2.22222,0.0023489) (-1.85185,0.00295963) (-1.48148,0.00372113) (-1.11111,0.00466447) (-0.740741,0.00581115) (-0.37037,0.00718342) (0.0,0.00880585) };
\addlegendentry{\onebits, $b\!=\!128$};
\addplot coordinates { (-12.5,0.000455558) (-12.1875,0.000417065) (-11.875,0.000383734) (-11.5625,0.000351142) (-11.25,0.00026694) (-10.9375,0.000246483) (-10.625,0.000227725) (-10.3125,0.000212185) (-10.0,0.000167732) (-9.6875,0.000186363) (-9.375,0.000151176) (-9.0625,0.000145274) (-8.75,0.000144576) (-8.4375,0.000142104) (-8.125,0.000132879) (-7.8125,0.000138738) (-7.5,0.000152489) (-7.1875,0.000163181) (-6.875,0.000179032) (-6.5625,0.000206453) (-6.25,0.000238142) (-5.9375,0.000271866) (-5.625,0.000324495) (-5.3125,0.000395069) (-5.0,0.000465569) (-4.6875,0.000576842) (-4.375,0.000681295) (-4.0625,0.000837634) (-3.75,0.001018) (-3.4375,0.00124162) (-3.125,0.00151732) (-2.8125,0.00187091) (-2.5,0.00226493) (-2.1875,0.00277111) (-1.875,0.00336516) (-1.5625,0.00407586) (-1.25,0.00489406) (-0.9375,0.00587375) (-0.625,0.00698788) (-0.3125,0.00823936) (0.0,0.00966829) };
\addlegendentry{2-bit, $b\!=\!128$};
\addplot coordinates { (-26.6667,0.0117626) (-26.0,0.00937658) (-25.3333,0.00835214) (-24.6667,0.00649915) (-24.0,0.00494558) (-23.3333,0.00367588) (-22.6667,0.00264974) (-22.0,0.0018548) (-21.3333,0.00124943) (-20.6667,0.000811399) (-20.0,0.000696296) (-19.3333,0.000429037) (-18.6667,0.000252317) (-18.0,0.000140941) (-17.3333,7.48113e-05) (-16.6667,3.76444e-05) (-16.0,1.90831e-05) (-15.3333,9.58312e-06) (-14.6667,9.79312e-06) (-14.0,6.16188e-06) (-13.3333,4.8075e-06) (-12.6667,4.49125e-06) (-12.0,4.57937e-06) (-11.3333,5.21813e-06) (-10.6667,6.46938e-06) (-10.0,1.06469e-05) (-9.33333,1.39481e-05) (-8.66667,1.99506e-05) (-8.0,2.9075e-05) (-7.33333,4.44138e-05) (-6.66667,6.9685e-05) (-6.0,0.000111784) (-5.33333,0.000183428) (-4.66667,0.000331454) (-4.0,0.000545313) (-3.33333,0.00090276) (-2.66667,0.0014821) (-2.0,0.00240582) (-1.33333,0.00381582) (-0.666667,0.00586993) (0.0,0.00868792) };
\addlegendentry{\onebits, $b\!=\!64$};
\addplot coordinates { (-18.75,2.169e-05) (-18.125,1.74838e-05) (-17.5,2.95006e-05) (-16.875,2.58669e-05) (-16.25,4.23537e-05) (-15.625,3.70863e-05) (-15.0,5.95669e-05) (-14.375,5.11656e-05) (-13.75,8.01362e-05) (-13.125,0.000124503) (-12.5,0.00011199) (-12.1875,0.000102417) (-11.875,0.000167561) (-11.5625,0.000154661) (-11.25,0.000144189) (-10.9375,0.000133428) (-10.625,0.00022282) (-10.3125,0.000207287) (-10.0,0.000192996) (-9.6875,0.000179489) (-9.375,0.000283045) (-9.0625,0.000264823) (-8.75,0.000255051) (-8.4375,0.000243194) (-8.125,0.00037187) (-7.8125,0.000350787) (-7.5,0.000333814) (-7.1875,0.000316779) (-6.875,0.000481061) (-6.5625,0.000476229) (-6.25,0.000465172) (-5.9375,0.000461648) (-5.625,0.000669364) (-5.3125,0.000673083) (-5.0,0.000690216) (-4.6875,0.000980121) (-4.375,0.00103607) (-4.0625,0.00109756) (-3.75,0.0012033) (-3.4375,0.00163678) (-3.125,0.00180167) (-2.8125,0.00203891) (-2.5,0.00237753) (-2.1875,0.0030698) (-1.875,0.00355425) (-1.5625,0.00415344) (-1.25,0.00489496) (-0.9375,0.00605374) (-0.625,0.00706369) (-0.3125,0.00826849) (0.0,0.00962377) };
\addlegendentry{2-bit, $b\!=\!64$};

    \end{axis}
  \end{tikzpicture}
  \caption{Fraction of empty slots for various configurations of bumped ribbon
    retrieval with $w=64$, depending on the overloading factor $\eps$.%\pesa{made figure a bit higher to better resolve lines}
  }
  \label{fig:epsilon64}
\end{figure}
\Cref{fig:epsilon64}
plots the fraction $e$ of empty slots of BuRR for
$w=64$ and several combinations of bucket size $b$ and different
threshold compression schemes.  Similar plots are given in the full paper for $w=32$, $w=128$, and for $w=64$
with sparse coefficients.  Note that (for an infinite number of
layers), the overhead is about
$o=e+\mu/(rb(1-e))$ where $r$ is the
number of retrieved bits and $\mu$ is the number of metadata bits per
bucket. Hence, at least when $\mu$ is constant, the overhead is a
monotonic function in $e$ and \mbox{thus minimizing $e$ also minimizes overhead.}

We see that for small $|\eps|$, $e$ decreases exponentially. For
sufficiently small $b$, $e$ can get almost arbitrarily small. For fixed
$b>w$, $e$ eventually reaches a local minimum because with
threshold-based compression, a large overload enforces large
thresholds ($>w$) and thus empty regions of buckets.  Which actual
configuration to choose depends primarily on $r$.  Roughly, for larger
$r$, more and more metadata bits (i.e., small $b$, higher resolution
of threshold values) can be invested to reduce~$e$.  For fixed~$b$
and threshold compression scheme, one can choose $\eps$ to
minimize $e$. One can often choose a larger $\eps$ to get slightly
better performance due to less bumping with little impact on $o$.
Perhaps the most delicate tuning parameters are the thresholds to use
for 2-bit and \onebit compression\ifconfVersion\else(see \cref{ss:thresholdCompression})\fi.
Indeed, in \cref{fig:epsilon64} \onebit compression has lower $e$ than
2-bit compression for $b=64$ but higher $e$ for $b=128$. We expect
that 2-bit compression could always achieve smaller $e$ than \onebit
compression, but we have not found choices for the threshold values
that always ensure this.
\ifconfVersion
\else
\Cref{tab:configs} in \cref{app:exp} summarizes key parameters
of some selected BuRR configurations.
\fi
\fi
% \pesa{todo for Lorenz:add some lines to the table that encompass sparse
% coefficients and Bu$^1$RR}\lhs{sparse done, don't have data for Bu$^1$RR}

\enlargethispage{0.35em}% lhs: enlarged page to make footnote fit
\myparagraph{Ribbon yields the fastest static AMQs for overhead $\bm{< 44\%}$.} 
\begin{figure}\centering
  \includegraphics[page=1,width=\textwidth]{%
    \ifconfVersion%
    heatmap.png%
    \else%
    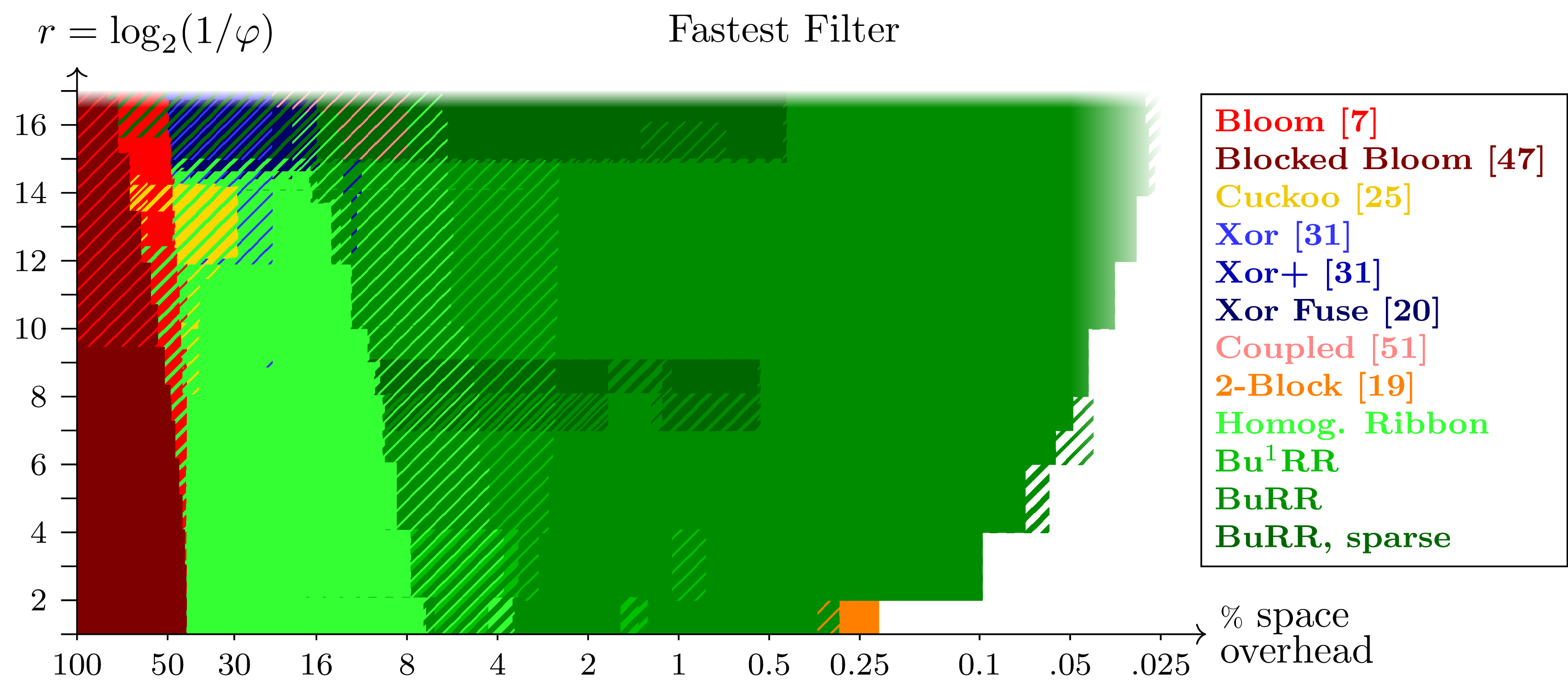%
    \fi%
  }
  \caption{\label{fig:heatmap}Fastest AMQ\pesa{todo: adapt citations once final} category for different choices of
    overhead and false-positive rate $ϕ=2^{-r}$. Shaded regions indicates
a dependency on the input type.  Ranking metric:
    construction time per key plus time for three queries, of which one is positive, one negative, and one mixed \mbox{(50\,\% chance of either).\stwa{Is this for sequential or parallel workloads?}\lhs{both, one of the diagonal shadings is for parallel}}}
\end{figure}
Consider \Cref{fig:scatter1} on page \pageref{fig:scatter1}, where
% \pesa{quotient filter also with larger overhead?}\lhs{I added 70\,\%, it didn't change much}
we show the tradeoff between space overhead and computation cost
for a range of AMQs for false positive rate $ϕ \approx 2^{-8}$
(i.e., $r=8$ for BuRR) and large inputs.%
\footnote{Small deviations of parameters are necessary because not all
  filters support arbitrary parameter choices. Also note that
  different filters have different functionality: (Blocked) Bloom
  allows dynamic insertion, Cuckoo, Morton and Quotient additionally
  allow deletion and counting. Xor \cite{BPZ:Practical:2013,DW:DensePeelable:2019,GL:XorFilters:2020}, Coupled \cite{W:SpatialCoupling:2021},
  LMSS \cite{LMSS:Efficient_Erasure:2001} and all ribbon variants are static retrieval data structures.}
In the parallel workload on the right all cores access many AMQs randomly.
% Of most interest are \emph{Pareto-optimal} configurations,
% i.e., those not dominated by other configurations with respect to both
% space and time. → sw: I find this sentence confusing. Does it describe that we only show pareto fronts per category? In that case “Pareto optimal” in the next sentence refers to something else. I found the sentence can be safely omitted.

Only three AMQs have Pareto-optimal configurations for this case: BuRR for
space overhead below 5\,\% (actually achieving between 1.4\,\% and 0.2\,\%
for a narrow time range of 830--890\,ns),
homogeneous ribbon for space overhead below 44\,\%
(actually achieving between 20\,\% and 10\,\% for a narrow time range
580--660\,ns), and
\emph{blocked Bloom filters} \cite{Putze:Efficient-Bloom-Filters:2009}
with time around 400\,ns at the price of space overhead of around
50\,\%. All other tried AMQs are dominated by homogeneous ribbon and
BuRR. Somewhat surprisingly, this even includes plain Bloom filters \cite{B:Space:1970}
which are slow because they incur several cache faults for each
insertion and positive query. Since plain Bloom filters are
extensively used in practice (often in cases where a static interface
suffices), we conclude that homogeneous ribbon and BuRR are fast enough
for a wide range of applications, opening the way for substantial
space savings in those settings.
BuRR is at least twice as fast as all tried retrieval data structures.%
\ifconfVersion
\footnote{FiRe \cite{Sanders:Retrieval-FingerPrinting:2014} is likely
  to be faster but has two orders of magnitude higher overhead; see the full paper for more details.}
\else
\footnote{FiRe \cite{Sanders:Retrieval-FingerPrinting:2014} is likely
  to be faster but has two orders of magnitude higher overhead; see \cref{s:exp} for more details.}
\fi
% \stwa{What about BPZ-based
%   retrieval? This is the same as a Xor filter.}\lhs{see \texttt{plot\_scatter\_xor.tex} and comment in \texttt{appendix-experiments.tex} line 142, using BPZ in parallel is slow}%
%
The filter data structures that support counting and deletion (Cuckoo filters \cite{FAK:CuckooFilterBetter:2013} and the related Morton filters \cite{BJ:MortonFilters:2020} as well as the quotient filters QF \cite{MSW:ConcQuotientFilter:2020} and CQF \cite{BFJKKMMSS:QuotientFilters:2012})
% \frage{check for Cuckoo and Morton}\lhs{both support deletions}
are slower than the best static AMQs.

The situation changes slightly when going to a sequential workload with large inputs as shown on the left of \Cref{fig:scatter1}.%
% \footnote{Arguably, on modern machines a \emph{sequential} workload
%   that leaves most of the area of a processor idle is
%   highly inefficient and thus unrealistic. However, this benchmark approximates a parallel
%   workload that has high locality and thus puts little load on the
%   memory subsystem. In other words, \cref{fig:scatter1} shows two
%   extremes in a spectrum of conceivable workloads.\lhs{we've said this before, redundant}\stwa{Agreed. Removed.}}
 Blocked Bloom and BuRR are still the best filters for large and small overhead,
respectively.  But now homogeneous ribbon and (variants of) the hypergraph peeling based Xor filters \cite{GL:XorFilters:2020,DW:DensePeelable:2019} share the
middle-ground of the Pareto curve between them. Also, plain Bloom
filters are almost dominated by Xor filters with half the overhead.  The reason is
that modern CPUs can handle several main memory accesses
in parallel. This is very helpful for Bloom and Xor, whose queries do little else than
computing the logical (x)or of a small number of randomly chosen memory cells.
Nevertheless, the faster variants of BuRR are only moderately slower than
Bloom and Xor filters while having at least an order of magnitude smaller overheads.

% lhs: changed back \paragraph to \paragraph because otherwise the spacing
% before the list looks off, and it doesn't change the page breaks
\myparagraph{Further Results.} Other claims supported by our data are:
% \Cref{s:exp} describes further situations. We summarize the
% findings here.
\begin{itemize}
  \setlength{\itemsep}{0pt}
  • \textbf{Good ribbon widths are $w = 32$ and $w = 64$.} Ribbon widths as small as
$w=16$ can achieve small overhead but at least on 64-bit
processors, $w\in\{32,64\}$ seems most sensible. The case $w=32$ is only 15--20\,\%
faster than $w=64$
while the latter has about four times less overhead.
Thus the case $w=64$ seems the most favorable one. This confirms that the
linear dependence of the construction time on $w$ is to some extent hidden behind the cache faults which are similar for both values of $w$ (this is in line with our analysis in the external memory model).
%% Ribbon width $w=128$  allows further reduction of overhead to far below 1\,\%
%% but the performance penalty is much larger than the step from $w=32$ to $w=64$,
%% which is to be expected on a 64-bit architecture.
  • \textbf{Bu$\bm{{}^1}$RR is slower than BuRR} by about 20\,\%, which may be a reasonable price
  for better worst-case query time in some real-time applications.%
\ifconfVersion
\footnote{Part of the performance difference might be due to implementation details; see the full paper.}
\else  
\footnote{Part of the performance difference might be due to implementation details; see \cref{ss:twoLayer}.}
\fi
  • \textbf{The $\bm{1^+}$-bit variant of BuRR is smaller but slower}
than the variant with 2-bit metadata per bucket, as expected, though not by a large margin.
% though only slightly. the variant with 2-bit metadata per bucket
% is slightly faster than the $1^+$-bit variant which in turn allows
% slightly smaller overheads for small $r$\lhs{why just for small $r$?}.
% \stwa{I don't see the “for small r” qualification in the data either, so I guess Lorenz is right?}
• \textbf{Smaller inputs and smaller $\bm{r}$ change little.}
For inputs that fit into cache, the Pareto
curve is still dominated by blocked Bloom, homogeneous ribbon, and BuRR,
but the performance penalty for achieving low overhead increases.
For $r = 1$ we have data for additional competitors. GOV \cite{GOV:retrieval-Compressed:2020}, which relies on structured Gaussian elimination,
% \pesa{a word on what it does?} \stwa{done}
is several times slower than BuRR and exhibits an unfavorable time--overhead tradeoff.  2-block \cite{DW:Retrieval-log-extra-bits:2019} uses two small dense blocks of nonzeroes and
can achieve very small overhead at the cost of prohibitively expensive construction.
• \textbf{For large $\bm{r}$, Xor filters and Cuckoo filters
come into play.}
\Cref{fig:heatmap} shows the fastest AMQ depending on overhead and false positive rate $ϕ = 2^{-r}$ up to $r = 16$. While blocked Bloom, homogeneous ribbon, and BuRR cover most of the area,
they lose ground for large $r$ because their running time depends on $r$. Here Xor filters and Cuckoo filters make an appearance.
% \stwa{I removed a sentence which did not make sense to me:}
% For large $r$, blocked Bloom and homogeneous ribbon lose ground since their false-positive rate becomes less favorable.
% \stwa{My objection: $r = \log(1/ϕ)$ \emph{determines} the fpr. Something else must become less favourable when we make $r$ large.}
% BuRR takes over most of the ground for small overheads while Xor filters and Cuckoo filters
% come into play for larger overheads.
% A few things change for different operation mixes:
• \textbf{Bloom filters and Ribbon filters are fast for \emph{negative queries}} where, on average, only two bits need to be retrieved to prove that a key is not in the set. This improves the relative standing  of plain Bloom filters on large and parallel workloads with mostly negative queries. %\pesa{should we also have a scatter plot for negative queries?}\lhs{added, \cref{fig:scatterQueryNeg}}.
• \textbf{Xor filters \cite{GL:XorFilters:2020} and Coupled \cite{W:SpatialCoupling:2021} have fast queries} since they can exploit parallelism in memory accesses. They suffer, however, from slow construction on large sequential inputs due to poor locality, and exhibit poor query performance when accessed from many threads in parallel.  For small $n$, large $r$, and overhead between 8\,\% and 20\,\%, Coupled becomes the fastest AMQ.
\end{itemize}

\section{Related Results and Techniques}
\label{s:related}

We now take the time to review some related work on retrieval including all approaches listed in \cref{tab:retrieval-comparison}.

\myparagraph{Related Problems.}
An important application of retrieval besides AMQs is encoding perfect \emph{hash functions} (PHF), i.e.\ an injective function $p: S → [(1+ε)|S|]$ for given $S ⊆ \U$. Objectives for $p$ are compact encoding, fast evaluation and small $ε$. 
Consider a result from cuckoo hashing \cite{FPSS:Space_Efficient:2005,FP:Sharp:2012,L:A_New_Approach:2012}, namely that given four hash functions $h₁,h₂,h₃,h₄ : S → [1.024|S|]$ there exists, with high probability, a choice function $f : S → [4]$ such that $x ↦ h_{f(x)}(x)$ is injective.
A $2$-bit retrieval data structure for $f$ therefore gives rise to a perfect hash function~\cite{BPZ:Practical:2013}, see also \cite{CKRT:The_Bloomier:2004}.
Retrieval data structures can also be used to directly store compact names of
objects, e.g., in column-oriented databases
\cite{Sanders:Retrieval-FingerPrinting:2014}. This takes more space
than perfect hashing but allows to encode the ordering of the keys
into the names.

In retrieval for AMQs and PHFs the stored values $f(x) ∈ \{0,1\}^r$ are uniformly random. However, some authors consider applications where $f(x)$ has a skewed distribution and the overhead of the retrieval data structure is measured with respect to the \emph{$0$-th order empirical entropy} of $f$ \cite{HKP:CompressedFunction:2009,BelazzouguiV13,GOV:retrieval-Compressed:2020}.
Note that once we can do 1-bit retrieval with low overhead, we can use that to store data with
prefix-free \emph{variable-bit-length encoding} (e.g. Huffman or Golomb codes).
We can store the $k$-th bit of $f(x)$ as
data to be retrieved for the input tuple $(x,k)$. This can be further
improved by storing $R$ 1-bit retrieval data structures where $R = \max_{x ∈ S}|f(x)|$ \cite{HKP:CompressedFunction:2009,BelazzouguiV13,GOV:retrieval-Compressed:2020}.
By interleaving these data structures, one can make queries almost as fast
as in the case of fixed $r$.

\myparagraph{More Linear Algebra based approaches.}
It has long been known that some matrices with random entries are likely to have full rank, even when sparse \cite{Cooper:Rank-Of-Random-Matrices:2000} and density thresholds for random $k$-XORSAT formulas to be solvable – either at all \cite{DM:The_3-XORSAT:2002,DGMMPR:Tight:2010} or with a linear time peeling algorithm \cite{Molloy05:Cores-in-random-hypergraphs,Luczak:A-simple-solution} – have been determined.

Building on such knowledge, a solution to the retrieval problem was identified by Botelho, Pagh and Ziviani \cite{BPZ:Simple:2007,B:Near-Optimal,BPZ:Practical:2013} in the context of perfect hashing. In our terminology, their rows $\vec{h}(x)$ contain $3$ random $1$-entries per key which makes $AZ = \bm{b}$ solvable with peeling, provided $m > 1.22n$.

Several works develop the idea from \cite{BPZ:Practical:2013}. In \cite{Vigna:Fast-Scalable-Construction-of-Functions:2016,GOV:retrieval-Compressed:2020} only $m > 1.089n$ is needed in principle (or $m > 1.0238n$ for $|\vec{h}(x)| = 4$) but a Gaussian solver has to be used. More recently in the \emph{spatial coupling} approach
\cite{W:SpatialCoupling:2021} $\svec{h}(x)$ has $k$ random $1$-entries within a small window, achieving space overhead $\approx e^{-k}$ while still allowing a peeling solver.
With some squinting, a class of linear erasure correcting codes from \cite{LMSS:Efficient_Erasure:2001} can be interpreted as a retrieval data structure of a similar vein, where $|\vec{h}(x)| ∈ \{5,…,k\}$ is random with expectation $\O(\log k)$.

Two recent approaches also based on sparse matrix solving are \cite{DW:Retrieval-log-extra-bits:2019,DW:One-Block-per-Row:2019} where $\vec{h}(x)$ contains two blocks or one block of random bits. Our ribbon approach builds on the latter.

We end this section with a discussion of seemingly promising techniques and give reasons why we choose not to use them in this paper.
\ifconfVersion
Some more details are also discussed in the experimental section of the full paper.
\else
Some more details on methods used in the experiments are also discussed in \cref{s:exp}.
\fi

\myparagraph{Shards.}
A widely used technique in hashing-based data structures is to use a
splitting hash function to first divide the input set into many much
smaller sets (shards, buckets, chunks, bins,\ldots) that can then be
handled separately
\cite{GOV:retrieval-Compressed:2020,BBOVV:Cache-Oblivious-Peeling:14,DW:Retrieval-log-extra-bits:2019,DW:One-Block-per-Row:2019,ADR:Experimental:2009,P:An_Optimal:2009}. This
incurs only linear time overhead during preprocessing and constant
time overhead during a query, and allows to limit the impact of
superlinear cost of further processing to the size of the shard.  Even
to ribbon, this could be used in multiple ways.  For example, by
statically splitting the table into pieces of size $n^{ε}$ for
standard ribbon, one can achieve space overhead
$ε+\Oh{n^{-ε}}$, preprocessing time $\Oh{n/ε}$,
and query time $\Oh{r}$ \cite{DW:One-Block-per-Row:2019}. This is, however, underwhelming on reflection. Before arriving at the current form of BuRR, we
designed several variants based on sharding but never achieved better
overhead than $\Om{1/w}$. The current overhead of $\Oh{\log w/w^2}$
comes from using the splitting technique in a ``soft'' way -- keys
are assigned to buckets for the purpose of defining bumping
information but the ribbon solver may implicitly allocate them to
subsequent buckets.

\myparagraph{Table lookup.}
The first asymptotically efficient succinct retrieval data structure
we are aware of \cite{P:An_Optimal:2009} uses two levels of sharding
to obtain very small shards of size $\Oh{\sqrt{\log n}}$ with small
asymptotic overhead.  It then uses dense random matrices per shard to
obtain per-shard retrieval data structures.  This can be done in
constant time per shard by tabulating the solutions of all possible
matrices. This leads to a multiplicative overhead of $\Oh{\log\log
  n/\sqrt{\log n}}$. 
% Already done on Page 2.
% However, note that for $n\leq 2^{64}$, $\log\log
% n/\sqrt{\log n}\geq 3/4$ so that we do not get a meaningful bound for
% realistic input sizes.
% 
Belazzougui and Venturini \cite{BelazzouguiV13} use slightly larger
shards of size $\Oh{(1+\log\log(n)/r)\log\log(n)/\log n}$. Using
carefully designed random lookup tables they show that linear
construction time, constant lookup time, and overhead $\Oh{(\log\log
  n)^2/\log n}$ is possible. We discussed on page \pageref{page:overhead-calculation-porat} why we suspect large overhead for \cite{P:An_Optimal:2009} and \cite{BelazzouguiV13} in practice.
% This even applies to result sets with
% arbitrary distributions.  Although their overhead is asymptotically
% smaller than the $\Oh{\log\log n/\sqrt{\log n}}$ above, the numerical
% test for $n\leq 2^{64}$ and $r=1$ yields 0.657 which is also not
% useful for practical input sizes.

% Recall that the overhead of BuRR is $\Oh{\log(w)/(rw^2)}$ which can be made arbitrarily small independent of $n$. Even if we set $w=\log n$ for better comparability, we get
% $\Oh{\log\log(n)/(r\log^2n)}$ which is asymptotically better (and a numerical test gives 0.4\,\%
% even for a realistic $n = 2^{32}$).

In general, lookup tables are often problematic for compressed data
structures in practice -- they cause additional space overhead and
cache faults. Even if the table is small and fits into cache, this may
yield efficient benchmarks but can still cause cache faults in
practical workloads where the data structure is only a small part in a
large software system with a large working set.

\myparagraph{Cascaded bumping.}
Hash tables consisting of multiple shrinking
levels are also used in \emph{multilevel adaptive hashing}
\cite{BK:Multilevel:1990} and \emph{filter hashing} \cite{FPSS:Space_Efficient:2005}.
While similar to BuRR in this sense, they do not maintain bumping information. This is fine for storing key-value pairs because all levels can be searched for a requested key. But it is unclear how the idea would work in the context of retrieval, i.e.\ without storing keys.

% \input{from-retrieval-to-filters}

%!TEX root=./main.tex

\section{Ribbon Insertions}\label{s:ribbonRetrieval}

In this section we enhance the \sgauss construction for retrieval from \cite{DW:One-Block-per-Row:2019} with a new solver called \textbf{R}apid \textbf{I}ncremental \textbf{B}oolean \textbf{B}anding \textbf{ON} the fly (\emph{Ribbon}), which is the basis of all ribbon variants considered later.
% , including BuRR in \cref{sec:bumped-ribbon}.
% and homogeneous ribbon filters in \cref{sec:homogeneous-ribbon}.

\myparagraph{The \sgauss construction.} For a parameter $w \in\nat$ that we call the \emph{ribbon width}, the vector $\svec{h}(x) ∈ \{0,1\}^m$ is given by a random \emph{starting position} $s(x) ∈ [m-w-1]$ and a random \emph{coefficient vector} $c(x) ∈ \{0,1\}^w$ as $\svec{h}(x) = 0^{s(x)-1}c(x)0^{m-s(x)-w+1}$. Note that even though $m$-bit vectors like $\svec{h}(x)$ are used to simplify mathematical discussion, such vectors can be represented using $\log(m)+w$ bits.
The matrix $A$ with rows $(\svec{h}(x))_{x ∈ S}$ sorted by $s(x)$ has all of its $1$-entries in a “ribbon” of width $w$ that randomly passes through the matrix from the top left to the bottom right, as in \cref{fig:ribbon-matrix} \textbf{(a)}.
The authors show:
\begin{theorem}[{\cite[Thm 2]{DW:One-Block-per-Row:2019}}]
    \label{thm:sgauss}
    For any constant $0 < ε < \frac{1}{2}$ and $\frac{n}{m} = 1-ε$ there is a suitable choice for $w = Θ(\smash{\frac{\log n}{ε}})$ such that with high probability the linear system $(\svec{h}(x)·Z = f(x))_{x ∈ S}$ is solvable for any $r ∈ ℕ$ and any $f : S → \{0,1\}^r$. Moreover, after sorting $(\svec{h}(x))_{x ∈ S}$ by $s(x)$, Gaussian elimination can compute a solution $Z$ in expected time $\O(n/ε²)$.
\end{theorem}

%% \begin{figure}[h]
%%     \centering
%%     \begin{tabular}{c@{\!\!\!}cc@{\!\!\!}c}
%%     (a) & & (b) \\[-15pt]
%%     &\includegraphics[page=1]{figures/MatrixPictures.pdf}
%%     &&\includegraphics[page=4]{figures/MatrixPictures.pdf}
%%     \end{tabular}
%%     \caption[fragile]{$\textbf{(a)$} Typical shape of the random matrix $A$ with rows ${(\vec{h}(x))_{x ∈ S}}$ sorted by starting positions. The shaded “ribbon” region contains random bits. Gaussian elimination never causes any fill-in outside of the ribbon.\\
%%     $\textbf{(b)}$ Shape of the linear system ${M}$ central to Boolean banding on the fly.}
%%     \label{fig:ribbon-matrix}
%% \end{figure}

\myparagraph{Boolean banding on the fly.} For ribbon retrieval we use the same hash function $\svec{h}$ as in \sgauss except that we force coefficient vectors $c(x)$ to start with $1$, which slightly improves presentation and prevents construction failures caused by single keys with $c(x) = 0^w$.
% \pedi{Note: we need to guarantee $c(x) ≠ 0$ for $w = o(\log n)$, right, to avoid block of all zeros? (Does that change how we phrase this?)}
% \footnote{In asymptotic considerations this change is inconsequential (and mildly annoying). For better alignment with \cite{DW:One-Block-per-Row:2019} our theorems still assume that $c(x)$ is uniformly distributed in $\{0,1\}^w$.\label{fn:no-leading-1-in-theory}}
The main difference lies in how we solve the linear system. The \emph{insertion phase} maintains a system $M$ of linear equations in row echelon form using on-the-fly Gaussian elimination~\cite{BGV:OnTheFlyGauss:2010}. This system is of the form shown in \cref{fig:ribbon-matrix} \textbf{(b)} and has $m$ rows that we also call \emph{slots}.
The $i$-th slot contains a $w$-bit vector $c_i ∈ \{0,1\}^w$ and $b_i ∈ \{0,1\}^r$. Logically, the $i$-th slot is either empty ($c_i = 0^w$) or specifies a linear equation $c_i · Z_{[i,i+w)} = b_i$ where $c_i$ starts with a $1$.
% sw: Didn't like the following phrase because it expresses an informal intention (we haven't defined reduced form) while the rest is pure math. Also “reduced form” is mentioned two sentences before already.
% to keep $M$ a reduced form.
With $Z_{[i,i+w)} ∈ \{0,1\}^{w × r}$ we refer to rows $i,…,i+w-1$ of $Z ∈ \{0,1\}^{w × r}$. We ensure $c_i · Z_{[i,i+w)}$ is well-defined even when $i + w - 1 > m$ with the invariant that $c_i$ never selects ``out of bounds'' rows of $Z$.

We consider the equations $\svec{h}(x)·Z = f(x)$ for $x ∈ S$ one by one, in arbitrary order, and try to integrate each into $M$ using \cref{algo:ribbon-insertion}, which we explain now.

\SetKwFor{Loop}{loop}{}{again}
\SetKwData{shift}{shift}
\SetKwFunction{ctz}{ctz}
\begin{algorithm}[h]
$(i,c,b) ← (s(x),c(x),f(x))$\;
 \Loop{}{
  \If(\tcp*[h]{slot $i$ of $M$ is empty}){$M.c[i] = 0$}{
   $(M.c[i],M.b[i]) ← (c,b)$\;
   \Return \textsc{success}% (inserted)
  }
  $(c,b) ← (c ⊕ M.c[i],b ⊕ M.b[i])$\;
  \If{$c = 0$}{
   \lIf{$b = 0$}{ \Return \textsc{redundant}} % { \Return \textsc{success} (redundant) }
   \lElse{ \Return \textsc{failure}} %  (inconsistent)
  }
  $j ← \mathrm{findFirstSet}(c)$ \tcp{a.k.a. BitScanForward}
  $i ← i + j$\;
  $c ← c >> j$ \tcp{logical shift last toward first}
 }
 \caption[fragile]{\hbox{Adding a key's equation to the linear system $M$.}}
 \label{algo:ribbon-insertion}
\end{algorithm}

A key's equation may be modified several times before it can be added to $M$, but a loop invariant is that its form is
\begin{equation}
    c · Z_{[i,i+w)} = b \text{ for some $i ∈ [m]$, $c ∈ 1∘\{0,1\}^{w-1}$, $b ∈ \{0,1\}^r$.}\label{eq:form-of-equation}
\end{equation}
The initial equation $\svec{h}(x)·Z = f(x)$ of key $x ∈ S$ has this form with $i = s(x)$, $c = c(x)$ and $b = f(x)$.
\begin{description}
    •[Case 1:] In the simplest case, slot $i$ of $M$ is empty and we can store \cref{eq:form-of-equation} in it.
    •[Case 2:] Otherwise slot $i$ of $M$ is occupied by an equation $c_i · Z_{[i,i+w)} = b_i$.  We perform a \emph{row operation} to obtain the new equation
    \begin{equation}
        c' · Z_{[i,i+w)} = b' \text{ with $c' = c ⊕ c_i$ and $b' = b ⊕ b_i$,}\label{eq:form-of-eq-modified}
    \end{equation}
    which, in the presence of the equation in slot $i$ of $M$, puts the same constraint on $Z$ as \cref{eq:form-of-equation}.
    Both $c$ and $c_i$ start with $1$, so $c'$ starts with $0$. We consider the following sub-cases.
    \begin{description}
        •[Case 2.1:] $c' = 0^w$ and $b' = 0^r$.
        The equation is void and can be ignored. This happens when the original equation of $x$ is implied by equations previously added to $M$.
        •[Case 2.2:] $c' = 0^w$ and $b' ≠ 0^r$.
        The equation is unsatisfiable. This happens when the key's original equation is inconsistent with equations previously added to $M$.
        •[Case 2.3:] $c'$ starts with $j > 0$ zeroes followed by a $1$. Then \cref{eq:form-of-eq-modified} can be rewritten as $c'' · Z_{[i',i'+w)} = b'$ where $i' = i + j$ and $c''$ is obtained from $c'$ by discarding the $j$ leading zeroes of $c$ and appending $j$ trailing zeroes.%
        \footnote{Note that in the bit-shift of \cref{algo:ribbon-insertion} the roles of “leading” and “trailing” may seem reversed because the least-significant “first” bit of a word is conventionally thought of as the “right-most” bit.}
    \end{description}
\end{description}
Termination is guaranteed since $i$ increases with each loop iteration.

\myparagraph{“On-the-fly” and “incremental.”} The insertion phase of Ribbon is \emph{on-the-fly} \cite{BGV:OnTheFlyGauss:2010}, i.e.\ maintains a row echelon form as keys arrive. This allows us to determine the longest prefix $(x₁,…,x_n)$ of a sequence $S = (x₁,x₂,x₃,…)$ of keys for which construction succeeds: Simply insert keys until the first failure. We say the insertion phase is \emph{incremental} since an insertion may lead to a new row in $M$ but does not modify existing rows. This allows us to easily undo the most recent successful insertions by clearing the slots of $M$ that were filled last. These properties are not shared by \sgauss and will be exploited by BuRR in \cref{sec:bumped-ribbon}.

\ifdefined\pagelimithack\relax\else
\myparagraph{Efficiency.} Running times of \sgauss and Ribbon are tied in $\O$-notation. However, Ribbon improves upon \sgauss in constant factors for the following reasons:
\begin{itemize}
    • Ribbon need not pre-sort the keys by $s(x)$.
% PD: changed "sort" to "pre-sort" because RIBBON arguably sorts on-the-fly
    • During construction, \sgauss explicitly maintains for each row the index of the left-most~$1$. Ribbon represents these implicitly, saving significant amounts memory.
    % This is because \sgauss does not compute an echelon\pesa{this term is used for the first time here? Better explain for non-LA experts?} form but only ensures that in each row the left-most $1$-entry — the pivot — is the bottom-most $1$-entry of its column.
    • \sgauss performs some number $D$ of elimination steps, which, depending on some bit, turn out to be xor-operations or no-ops. Ribbon on the other hand performs roughly $D/2$ bit shifts and $D/2$ (unconditional) xor operations. Though the details are complicated, intuition on branching complexity seems to favor Ribbon.
\end{itemize}
\fi

\def\Obs#1{{\small\textbf{\color{darkgray}\sffamily(\kern-0.5ptO#1\kern-0.5pt)}}}
\def\Assump#1{{\small\textbf{\color{darkgray}\sffamily(\kern-0.5ptM#1\kern-0.5pt)}}}
% \pesa{changed the Claim-macro so that it works on my system. curiously, Assump works but we should use uniform styling?}

\section{Analysis of Ribbon Insertions}
\label{sec:ribbon-analysis}

Given a set $S$ of $n$ keys we wish to analyze the process of inserting these keys into the system $M$ using \cref{algo:ribbon-insertion}. In particular, we are interested in the number of successful and failed insertions, the set of occupied slots in $M$ and the total running time.
Recall that $A ∈ \{0,1\}^{n × m}$ contains the rows $\svec{h}(x)$ for $x ∈ S$ sorted by $s(x)$, see \cref{fig:ribbon-matrix} \textbf{(a)}. Our analysis considers the \emph{ribbon diagonal}, which is a line passing through $A$. We begin with an instructive simplification.

\subsection{A Warm Up: The Simplified Ribbon Diagonal}
\label{sec:intuition}

We make the following two assumptions:
\begin{description}
	\item[\Assump{1}] Keys are inserted in the order they appear in $A$ (sorted by $s(x)$). This ensures that the insertion of each key $x ∈ S$ fails or succeeds within the first $w$ steps because no $1$-entries can exist in $M$ beyond column $s(x)+w-1$.
	\item[\Assump{2}] Inserting $x ∈ S$ fills the first free slot $i ∈ [s(x),s(x)+w-1]$ unless all of these slots are occupied, in which case the insertion fails. This ignores the role of $c(x)$.
\end{description}
\Cref{fig:simplified-diagonal} visualizes the process with an $×$ in position $(j,i)$ if the insertion of the $j$-th key fills slot $i$ of $M$. These points approximately trace a diagonal line from top left to bottom right and we call it the \emph{simplified ribbon diagonal} $\dsimp$. We make the following observations:
\begin{description}
	\item[\Obs{1}] If $\dsimp$ were to cross the bottom border of the ribbon, it skips a column (shown in green). Column $i$ is skipped if and only if slot $i$ of $M$ remains empty.
	\item[\Obs{2}] If $d$ were to cross the right border of the ribbon, it skips a row (shown in red). Row $j$ is skipped if and only if the $j$-th key is not inserted successfully.
	\item[\Obs{3}] The area enclosed between $d$ and the left border of the ribbon (shown in yellow) is an upper bound on the number of row operations performed during successful insertions.
\end{description}

%!TEX root=./main.tex

\subsection{The Ribbon Diagonal}
\label{sec:ribbon-diagonal}

A formal analysis can salvage much of the intuition from the simplified model. First, we show that \Assump{1}, though not \Assump{2}, can be made without loss of generality. For an adjusted definition of the ribbon diagonal, we then prove probabilistic versions of \Obs{1}, \Obs{2} and \Obs{3}. The following notation will be useful.
\begin{itemize}
	• $S_i = \{x ∈ S \mid s(x) ≤ i\}$ and $s_i = |S_i|$, for $i ∈ [m]$.
	• $S' ⊆ S$ is the set of keys not inserted successfully. Moreover, $S'_i = S_i ∩ S'$ and $s_i' = |S'_i|$.
	• $r_i$, for $i ∈ [m]$, is the rank of the first $i$ columns of $A$.
	• $P_M$ is the set of slots of $M$ that end up being filled.
\end{itemize}
\subparagraph{On \Assump{1}: The order of keys is irrelevant.}
Since $M$ arises from $A$ by row operations, which do not affect ranks of sets of columns, we conclude that $r_i$ is the rank of the first $i$ columns of $M$, regardless of the order in which keys are handled. From the form of $M$ (see \cref{fig:ribbon-matrix} \textbf{(b)}) it is clear that $i ∈ P_M ⇔ r_i = r_{i-1}+1$. Therefore, the set $P_M$ and thus the number $n-|P_M| = |S'|$ of failed insertions is also invariant.

Assuming all insertions are successful, the number of row operations performed for key $x$ is at most the distance of $s(x)$ to the slot $i(x) ∈ P_M$ that is filled. An invariant upper bound $Δ$ on the number of row operations, which are the dominating contribution to construction time, is then
\begin{equation}
	Δ := \sum_{x ∈ S} (i(x) - s(x)) = \sum_{i ∈ P_M} i - \sum_{x ∈ S}s(x).\label{eq:row-op-upper-bound}
\end{equation}
% Unfortunately, in the presence of failures both $S'$ (though not $|S'|$) and the time spent on keys in $S'$ depends on the order of the keys.
Except for the time related to failed insertions, which we have to bound separately, we can derive everything we want from $S$ and the invariants $P_M$, $|S'|$. We can therefore assume \Assump{1}.

\subparagraph{Definition and properties of $d$.}
Given \Assump{1}, we formally define the \emph{ribbon diagonal} $d$ as the following set of matrix positions in $A$.
\[d = \{(d_i,i) \mid i ∈ [m]\} \text{ where } d_i = r_i + s'_{i-w+1}.\]
It is useful to imagine the “default case” to be $r_{i} = r_{i-1} + 1$ and $s'_{i} = s'_{i-1}$. We then have $d_i = d_{i-1}+1$ and the ribbon diagonal indeed moves {diagonally down and to the right}. An {empty slot $i ∉ P_M$ correspond to a right-shift} (due to $r_{i} = r_{i-1}$) and a {failed insertion of a key with $s(x) = i-w+1$ correspond to a down-shift} (due to $s'_{i-w+1} > s'_{i-w}$).

Let us first check that $d$ is actually within the ribbon. More precisely:
\begin{lemma}
	\label{lem:diagonal-in-ribbon}
	For any $i ∈ [m]$, $d_i$ is not below the \emph{bottom ribbon border} $s_i$ and at most one position above the \emph{top ribbon border} $s_{i-w}+1$.%
\end{lemma}
% \footnote{%
	% We admit that intuition is stretched here because the ribbon might be disconnected if $w$ consecutive starting positions have no key assigned to them. The “top border” can then be one position \emph{below} the bottom border with a “ribbon height” of zero. There is little to be learned by dwelling on such strange (low probability) cases.
% }
\begin{proof}
	The first claim holds because
	\[d_i = r_i + s'_{i-w+1} ≤ r_i + s'_i = |P_M ∩ [1,i]| + s'_i ≤ s_i\]
	where the last step uses that each key in $S_i$ can fail to be inserted or fill a slot in $M$, but not both. The latter is true because
	\[d_i = |P_M ∩ [1,i]| + s'_{i-w+1} ≥ (s_{i-w+1}-s'_{i-w+1}) + s'_{i-w+1} = s_{i-w+1} ≥ s_{i-w}.\]
	where the first “$≥$” uses that the first $s_{i-w+1}$ rows cause $s_{i-w+1}-s'_{i-w+1}$ slots with index at most $i$ to be filled.
\end{proof}

The first part of \cref{lem:diagonal-in-ribbon} ensures that the height $h_i := s_i - d_i$ of the ribbon diagonal above the bottom ribbon border is non-negative. It plays a central role in the precise versions of \Obs{1} to \Obs{3} we prove next. The main adjustment we have to make is that $d$ is probabilistically repelled when \emph{close} to the ribbon border, while $\dsimp$ only responds to outright collisions.

\begin{lemma}[Precise version of \Obs{1}]
	\label{lem:no-position-empty}
	% Let $h_{i-1} := s_{i-1} - d_{i-1}$ be the \emph{height} of the ribbon diagonal above the bottom ribbon border in column $i-1$.
	We have $\Pr[i ∉ P_M \mid h_{i-1} = k] ≤ 2^{-k}$ for any $k ∈ ℕ₀$.
\end{lemma}
\begin{proof}
	A useful alternative way to think about \cref{algo:ribbon-insertion} uses language from linear probing: A key $x$ \emph{probes} slots $s(x),s(x)+1,…,s(x)+w-1$ one by one. When probing an empty slot, $x$ is inserted into that slot with probability $\frac{1}{2}$, otherwise it keeps probing.\footnote{This uses that for any key $x$ and $i ∈ [s(x)+1,s(x)+w-1]$ the random coefficient $a_i$ that $x$ has for slot $i$ remains fully random until slot $i$ is reached, since the bits that are added to $a_i$ during row operations are uncorrelated with $a_i$.}
	Now consider slot $i$.
	Of the $s_{i-1}$ keys with starting position at most $i-1$, precisely $r_{i-1}$ are successfully inserted to slots in $[1,i-1]$ and $s'_{i-w}$ insertions fail without probing slot $i$. Therefore $s_{i-1}-s'_{i-w}-r_{i-1} = s_{i-1}-d_{i-1} = h_{i-1}$ keys probe slot $i$. So conditioned on $h_{i-1} = k$, slot $i$ remains empty with probability at most $2^{-k}$.
\end{proof}

\begin{lemma}[Precise version of \Obs{2}]
	\label{lem:no-failed-insertion}
	Let $x$ be a key with $s(x) = i$.
	\begin{substatement}
		• Let $i' ∈ [i,i+w]$ be the column of $A$ where the ribbon diagonal passes the row of $x$. Assume $i'-i = w-k$, i.e. $i'$ is $k$ positions left of the right ribbon border. Conditioned on this, $\Pr[x ∈ S'] ≤ 2^{-k}$.
		• A simple variant of this claim is: If $h_i ≤ w-k$ for some $k ∈ ℕ$ then $\Pr[x ∈ S'] ≤ 2^{-k}$.
	\end{substatement}
	
\end{lemma}
\begin{proof}
	\begin{substatement}
		• Let $i = s(x)$. We may assume that $x$ is the last key with starting position $i$ as this can only increase $\Pr[x∈ S']$. This means $x$ corresponds to row $s_i$ and hence $d_{i'} ≥ s_i$. Of the $s_i-1$ keys that are handled before $x$, exactly $r_{i-1}$ are inserted before slot $i$ and at least $s_{i-1}'$ were not inserted successfully. The number of slots in $[i,m]$ that are occupied when $x$ is handled is therefore at most
		\begin{align*}
			s_i-1-r_{i-1}-s_{i-1}' ≤ d_{i'}-1-r_{i-1}-s_{i-1}' = r_{i'}+s_{i'-w+1}-1-r_{i-1}-s_{i-1}'\\
			≤ r_{i'}-r_{i-1}-1 ≤ i'-i = w-k.
		\end{align*}
		This means at least $k$ slots within $[i,i+w-1]$ are empty. The probability that $x$ cannot be inserted despite probing these $k$ slots is $2^{-k}$.
		• The assumption gives an alternative way to derive the same intermediate step:
		\[ s_i-1-r_{i-1}-s_{i-1}' ≤ s_i - r_i - s_{i-w+1}' = s_i - d_i = h_i ≤ w-k.\qedhere\]
	\end{substatement}
\end{proof}

\def\op{\mathrm{op}}
\begin{lemma}[Precise version of \Obs{3}]
 	\label{lem:time-bound-insertions}
	Let $\op₊$ be the number of row operations performed during \emph{successful} insertions. We have $\op₊ ≤ n(w-1)$ (trivially) as well as $\op₊ ≤ \sum_{i ∈ [m]} h_i$.
\end{lemma}

\begin{proof}
    First assume that all insertions succeed and consider \cref{eq:row-op-upper-bound}. Since $i(x) - s(x) ≤ w-1$ holds for all $x ∈ S$ the trivial bound $\op₊ ≤ Δ ≤ n(w-1)$ follows.
    Now consider the right hand side of \cref{eq:row-op-upper-bound}. The sum $\sum_{x ∈ S} (s(x)-1)$ can be interpreted as the area in $A$ (i.e.\ the number of matrix positions) left of the ribbon. Moreover we have
    \[
        \sum_{i ∈ P_M} (i-1) = \sum_{i ∈ [m]} |P_M ∩ [i+1,m]| = \sum_{i∈[m]} n -r_i = \sum_{i∈[m]} n -d_i.
    \]
    which is the area below the ribbon diagonal. This makes $Δ$ the area enclosed between the ribbon diagonal and the lower ribbon border. A column-wise computation of this area yields $\op₊ ≤ Δ = \sum_{i ∈ [m]} h_i$ as desired.

    Contrary to our initial assumption, there may be keys that fail to be inserted. But our bounds remain valid in the presence of such keys: The number $\op₊$ only counts operations made for successfully inserted keys and hence does not change. Our bounds $n(w-1)$ and $\sum_{i ∈ [m]} h_i$ are easily seen to increase by $w-1$ and $w$, respectively, for each additional “failing” key we take into account.
\end{proof}

% Old version, did not deal gracefully with the failing insertions.
% Let us revisit \cref{eq:row-op-upper-bound} assuming no insertions fail. The sum $\sum_{x ∈ S}s(x)$ can be interpreted as the area between the left ribbon border and the left border of $A$. Moreover, $P_M = \{i ∈ [m] \mid d_i = d_{i-1}+1\}$ so $\sum_{i ∈ P_M} i$ is the area between the ribbon diagonal and the left border of $A$.
% This makes $Δ$ the area between the ribbon diagonal and the left ribbon border.
% In this sense, \cref{eq:row-op-upper-bound} is a formalisation of \Obs{3}.
% \begin{lemma}[Precise version of \Obs{3}]
% 	\label{lem:time-bound-insertions}
% 	If all insertions succeed, then $Δ$ bounds the number of row operations during all insertions. A trivial bound on $Δ$ is $Δ ≤ n(w-1)$.
% \end{lemma}
% The trivial bound uses that under \Assump{1} we have $i(x)-s(x) ≤ w-1$ for each $x ∈ S$.

%%% Local Variables:
%%% mode: latex
%%% TeX-master: "main"
%%% End:

%%% Local Variables:
%%% mode: latex
%%% TeX-master: "main"
%%% End:

%!TEX root=./main.tex

\subsection{Chernoff Bounds}

The following lemma will play a role in \cref{sec:homogeneous-ribbon,sec:bumped-ribbon}.

\begin{lemma}
	\label{lem:chernoff}
	Let $(X_j)_{j ∈ [N]}$ be i.i.d.\ indicator random variables, $X := \sum_{j ∈ [N]} X_j$ and $μ := \E X$.
	\begin{substatement}
		• For $δ ∈ [0,1]$ we have
		$\Pr[|X - μ| ≥ δμ] ≤ 2\exp(-δ²μ/3)$.
		• There exists $C > 0$ such that for any $w ∈ ℕ$ and $μ ≤ \frac{2w²}{C\log w}$ we have
		$\Pr[|X - μ| ≥ \tfrac{w}{8}] = \O(w^{-5})$.
	\end{substatement}
\end{lemma}

\begin{proof}
	\begin{substatement}
		• This combines standard Chernoff bounds on the probability of $\{X ≥ (1+δ)μ\}$ and $\{X ≤ (1-δ)μ\}$ as found for instance in \cite[Chapter 4]{MU:Probability:2017}.
		• We set $δ = \frac{w}{8μ}$ and apply \textbf{(a)}. This gives
		\[ \Pr[|X - μ| ≥ \tfrac{w}{8}] ≤ 2\exp(-δ²μ/3) = 2\exp(-\tfrac{w²}{192μ}) ≤ 2\exp(-\tfrac{C \log w}{384}) = 2w^{-C/384}. \]
		Choosing $C = 1920$ achieves the desired bound.\footnote{We do not attempt to optimise $C$ here. In practice much smaller values of $C$ are sufficient, see \cref{s:exp}}
		\qedhere
	\end{substatement}
\end{proof}

%%% Local Variables:
%%% mode: latex
%%% TeX-master: "main"
%%% End:

% \input{ribbon-analysis} % included in ribbon-intuition
%!TEX root=./main.tex
\section{Analysis of Standard Ribbon Retrieval}
\label{sec:architectures}

% We now present two additional Ribbon architectures and sketch how an analysis in terms of the ribbon diagonal from \cref{sec:ribbon-analysis} might proceed.

By \emph{standard ribbon} we mean the original design from \cite{DW:One-Block-per-Row:2019}, except that we use our improved solver. We sketch an implementation in \cref{algo:standard-ribbon} and recall the broad strokes of the analysis from \cite{DW:One-Block-per-Row:2019} which will help us to analyze homogeneous ribbon filters in \cref{sec:homogeneous-ribbon}.

Given $n ∈ ℕ$ keys we allocate a system $M$ of size $m = n/(1-ε)+w-1$ and try to insert all keys using \cref{algo:ribbon-insertion}. If any insertion fails, the entire construction is restarted with new hash functions. Otherwise, we obtain a solution $Z$ to $M$ in the \emph{back substitution phase}. The rows of $Z$ are obtained from bottom to top. If slot $i$ of $M$ contains an equation then this equation uniquely determines row $i$ of $Z$ in terms of later rows of $Z$. If slot $i$ of $M$ is empty, then row $i$ of $Z$ can be initialized arbitrarily.

\SetKwInput{Input}{Input}
\SetKwInput{Parameters}{Parameters}
\SetKw{DownTo}{down to}
\SetKw{Restart}{restart}
\begin{algorithm}[h]
    \Input{$f : S → \{0,1\}^r$ for some $S ⊆ \U$ of size $n$.}
    \Parameters{$w ∈ ℕ$, $ε > 0$.}
    $m ← n/(1-ε)+w-1$; allocate system $M$ of size $m$\;
    pick hash functions $s : \U → [m-w+1],\ c : \U → \{0,1\}^w$\;
    \For{$x ∈ S$}{
        $\mathrm{ret} ← \mathrm{insert}(x)$ \tcp{using \cref{algo:ribbon-insertion}}
        \If{$\mathrm{ret} = \text{\upshape\scshape failure}$}{
            \Restart\;
        }
    }
    $Z ← 0^{m × r}$\;
    \For{$i = m$ \DownTo $1$}{
        $Z_i ← M.c[i] · Z_{i..i+w-1}$ \tcp{back substitution}
    }
    \Return $(s,c,Z)$\;
    \caption[fragile]{The construction algorithm of standard ribbon.}
    \label{algo:standard-ribbon}
\end{algorithm}

The expected ``slope'' of the ribbon is $1-ε$, giving us reason to hope that the ribbon diagonal will stick to the left ribbon border making failures unlikely.

% For $\dsimp$ to hit the right border and therefore for an insertion to fail by \Obs{3}, there would have to be a range $[i,i+k)$ into which $X ≥ k+w$ starting positions of keys fall.
% We can write $X$ as a sum of $n$ independent indicator random variables (one for each key) with $\E[X] = k(1-ε)$. Using the Chernoff bound from \cref{lem:chernoff} \textbf{(a)} with $δ = Δ/\E[X]$ where $Δ = kε+w$ yields
% \[\Pr[X ≥ \E[X]+Δ] ≤ \exp(-\tfrac{Δ²}{3\E[X]}) = \exp(-Θ(\tfrac{k²ε²+w²}{k})) = \exp(-Θ(kε²+\tfrac{w²}{k})).\]
% This bound is weakest for $k = \frac{w}{ε}$ (because both summands coincide) namely it is $\exp(-Θ(εw))$. To exclude the existence of an overfull range with high probability, $w = Ω(\frac{\log n}{ε})$ is required – precisely the value from \cite{DW:One-Block-per-Row:2019} (here: \cref{thm:sgauss}). A full analysis in terms of the (non-simplified) ribbon diagonal is not fundamentally different.

\begin{lemma}
	\label{lem:standard-ribbon}
	% In the standard ribbon setting,
	The heights $h_i := s_i - d_i$ for $i ∈ [m]$ satisfy:
	\begin{substatement}
		• $𝔼[h_i] ≤ 𝒪(1/ε)$
		• $∀k ∈ ℕ: \Pr[h_i > k] = \exp(-Ω(εk))$.
	\end{substatement}
\end{lemma}

\begin{proof}[Proof idea from \cite{DW:One-Block-per-Row:2019}.]
	\def\Bin{\mathrm{Bin}}
	\def\Po{\mathrm{Po}}
	By definition of $h_i$, $d_i$ and $r_i$ we have
	\begin{align*}
		h_i - h_{i-1} &= (s_i - s_{i-1}) - (r_i - r_{i-1}) - (s'_{i-w+1}-s'_{i-w})\\
		&≤ (s_i - s_{i-1}) - (r_i - r_{i-1}) = (s_i - s_{i-1}) - 𝟙_{i ∈ P_M}.
	\end{align*}
	The number $s_i - s_{i-1}$ of keys with starting position $i$ has distribution $\Bin(n,\frac{1}{m-w+1})$ which is approximately $\Po(1-ε)$.
	By \cref{lem:no-position-empty} we have $\Pr[i ∉ P_M] ≤ 2^{-h_{i-1}}$. Roughly speaking this means that $\Pr[𝟙_{i ∈ P_M}≠1]$ is negligible as soon as $h_{i-1}$ rises to a value large enough to threaten the upper bounds we intend to prove. A coupling argument then allows us to upper bound $h_i$ in terms of a so-called M/D/1 queue. In every time step $\Po(1-ε)$ customers arrive and $1$ customer can be serviced. The stated bounds on \textbf{(a)} expectation and \textbf{(b)} tails of $h_i$ stem from the literature on such queues.

	We remark that the term $s'_{i-w+1}-s'_{i-w}$ that we ignored relates to failed insertions. It translates to customers abandoning the queue after waiting for $w$ time steps without being serviced.
\end{proof}

By choosing $w = Ω(\frac{\log n}{ε})$ it follows from \cref{lem:standard-ribbon} \textbf{(b)} that $h_i ≤ w/2$ for all $i ∈ [m]$ whp. \Cref{lem:no-failed-insertion}~\textbf{(b)} then ensures that all keys can be inserted successfully whp. Combining \cref{lem:time-bound-insertions} with \cref{lem:standard-ribbon}~\textbf{(a)} shows that the expected number of row operations during construction is $\O(n/ε)$. This proves \cref{thm:standard-ribbon}.

% \stwa{Should we include smash?}

%!TEX root=./main.tex

\section{Analysis of Homogeneous Ribbon Filters}
\label{sec:homogeneous-ribbon}

In this section we give a precise description and analysis of homogeneous ribbon filters, which are even simpler than filters based on standard ribbon but unsuitable for retrieval.

Recall the approach for constructing a filter by picking hash functions $\svec{h} : \U → \{0,1\}^m$, $f : \U → \{0,1\}^r$ and finding $Z ∈ \{0,1\}^{m×r}$ such that all $x ∈ S$ satisfy $\svec{h}(x)·Z = f(x)$, while most $x ∈ \U \setminus S$ do not. We now dispose of the fingerprint function $f$, effectively setting $f(x) = 0$ for all $x ∈ \U$. A filter is then given by a solution $Z$ to the \emph{homogeneous} system $(\svec{h}(x)·Z = 0^r)_{x ∈ S}$. The FP rate for $Z$ is $ϕ_Z = \Pr_{a \sim H}[a·Z = 0^r]$ where $H$ is the distribution of $\svec{h}(x)$ for $x ∈ \U$. An immediate issue with the idea is that $Z = 0^{m × r}$ is a valid solution but gives $ϕ_Z = 1$. We therefore pick $Z$ \emph{uniformly at random} from all solutions. If $\svec{h}$ has the form $\svec{h}(x) = 0^{s(x)-1}c(x)0^{m-s(x)-w+1}$ from standard ribbon retrieval, we call the resulting filter a \emph{homogeneous ribbon filter}. The full construction is shown in \cref{algo:homogeneous-ribbon}. Two notable simplifications compared to \cref{algo:standard-ribbon} are that no function $f$ is needed and that a restart is never required.
Note, however, that free variables must now be sampled uniformly at random\footnotemark\ during back substitution.
\footnotetext{Our implementation uses trivial pseudo-random assignments instead: a free variable in row $i$ is assigned $pi \bmod 2^{r}$ for some fixed large odd number $p$.}

\begin{algorithm}[h]
    \Input{$S ⊆ \U$ of size $n$.}
    \Parameters{$r ∈ ℕ, w ∈ ℕ, ε > 0$.}
    $m ← n/(1-ε)+w-1$; allocate system $M$ of size $m$\;
    pick hash functions $s : \U → [m-w+1],\ c : \U → \{0,1\}^w$\;
    sort $S$ approximately by $s(x)$ (see \cref{lem:homogeneous-construction})\;
    \For{$x ∈ S$}{
        $\mathrm{insert}(x)$ \tcp{using \cref{algo:ribbon-insertion} with $f ≡ 0$. Cannot fail!}
    }
    $Z ← 0^{m × r}$\;
    \For{$i = m$ \DownTo $1$}{
        \uIf(\tcp*[h]{slot unused?}){$M.c[i] = 0$}{
            sample $Z_i \sim U(\{0,1\}^r)$ \tcp{randomly initialize free variable}
        }\Else{
            $Z_i ← M.c[i] · Z_{i..i+w-1}$  \tcp{back substitution}
        }
    }
    \Return $(s,c,Z)$\;
    \caption[fragile]{The construction algorithm of homogeneous ribbon filters.}
    \label{algo:homogeneous-ribbon}
\end{algorithm}

% To obtain one, all free variables, i.e.\ the variables corresponding to empty rows of $M$, are initialized randomly during back substitution.%
% It has two obvious advantages over Standard \ribbon filters:
% \begin{itemize}
%     • Constructions can never fail, regardless of $n$, $ε$ and $w$. This is simply because a homogeneous linear system always has at least the trivial solution. %with all-zero right hand sides a linear system cannot be inconsistent
%     • The absence of fingerprints slightly improves time and space in construction.
% \end{itemize}

The overall FP rate is $ϕ = \E[ϕ_Z]$ where $Z$ depends on the randomness in $(\svec{h}(x))_{x ∈ S}$ and the free variables.
A complication is that $ϕ = 2^{-r}$ no longer holds, instead there is a gap $ϕ - 2^{-r} > 0$.
% Intuitively, if too many equations constrain some part of $Z$ then that part will be insufficiently random.
We show that this gap is negligible under two conditions. Firstly, the filter must be \emph{underloaded}, with $ε ≈ \frac{m-n}{m} > 0$, which leads to a memory overhead of $𝒪(ε)$. Secondly, the ribbon width $w$ must satisfy $w = Ω(r/ε)$. The good news is that there is no dependence of $w$ on $n$ (such as $w = Ω(\frac{\log n}{ε})$ required in standard ribbon) and that no sharding or bumping is required.
% The bad news is that the requirement of $w$ on $ε$ is asymptotically worse than for bumped ribbon in two regards: $w$ needs to grow with $r$ and must be proportional to $1/ε$ rather than $\tilde{\O}(1/\sqrt{ε})$. This hurts construction and query times when small $ε$ is desired.
More precisely, we prove \cref{thm:homog-ribbon}, restated here for convenience.
\homogRibbonTheorem*
Note that when targeting $w = Θ(\log n)$ we can achieve an overhead of $ε = 𝒪(\frac{\max(r,\log \log n)}{\log n})$.

\subsection{Proof of Theorem \ref{thm:homog-ribbon}}

% We shall later require $εw > C\max(\log w, r)$ for some constant $C$ and choose $w$ as the smallest admissible number. The requirement on $ε$ in \cref{thm:homog-ribbon} is designed to ensure that $w = \O(\log n)$ still fits into $\O(1)$ words. The construction time of a homogeneous ribbon filter is easily dealth with given prior work in \cite{DW:One-Block-per-Row:2019}.

The easier part is to prove the running time bounds. The query time of $𝒪(r)$ is, in fact, obvious. For the construction time, we reuse results for standard ribbon. Though insertions cannot fail, the set of \emph{redundant} keys, i.e. keys for which \cref{algo:ribbon-insertion} returns “\textsc{redundant}” rather than “\textsc{success}” now demands attention.

\begin{lemma}
    \label{lem:homogeneous-construction}
    Consider the setting of \cref{thm:homog-ribbon}.
    \begin{substatement}
        • The fraction of keys that lead to redundant insertions is $\exp(-Ω(εw))$.
        • The expected number of row additions during construction is $\O(n/ε)$.
    \end{substatement}
\end{lemma}
\begin{proof}
    \begin{substatement}
        • Any key $x ∈ S$ with a starting position $i = s(x)$ for which $h_i ≤ w/2$ is, by \cref{lem:no-failed-insertion} \textbf{(b)}, inserted successfully with probability at least $1-2^{-w/2}$. By \cref{lem:standard-ribbon} the expected fraction of positions to which this argument does not apply is $\exp(-Ω(εw))$. From this it is not hard to see that the expected fraction of \emph{keys} to which this argument does not apply is also $\exp(-Ω(εw))$. The fraction of keys not inserted successfully is therefore $\O(2^{-w/2}) + \exp(-Ω(εw)) = \exp(-Ω(εw))$.
        • Combining \cref{lem:time-bound-insertions} with \cref{lem:standard-ribbon} \textbf{(a)} shows that the expected number of row operations during successful insertions is $\O(n/ε)$. Redundant keys are somewhat annoying to deal with. They are the reason we partially sort the key set in \cref{algo:homogeneous-ribbon}. If keys are sorted into buckets of $b$ consecutive starting positions each and buckets handled from left to right, then no attempted insertion can take longer than $b+w$ steps. Thus, $b ≤ \exp(Ω(εw))$ ensures that redundant insertions contribute $𝒪(n)$ to expected total running time.\qedhere
    \end{substatement}
\end{proof}

\noindent To get a grip on the false positive rate, we start with the following simple observation.
\begin{lemma}
    \label{lem:homogeneous-fpr}
    Let $p$ be the probability that for $y ∈ \U \setminus S$ the vector $\svec{h}(y)$ is in the span of $(\svec{h}(x))_{x ∈ S}$. The false positive rate of the homogeneous ribbon filter is
    \[\fpr = p + (1-p)2^{-r}.\]
\end{lemma}

\begin{proof}
    First assume there exists $S' ⊆ S$ with $\svec{h}(y) = \sum_{x ∈ S'} \svec{h}(x)$ which happens with probability $p$. In that case \[\svec{h}(y)·Z = (\sum_{x ∈ S'}\svec{h}(x))·Z = \sum_{x ∈ S'}(\svec{h}(x)·Z) = 0\] and $y$ is a false positive. Otherwise, i.e.\ with probability $1-p$, an attempt to add $\svec{h}(y)·Z = 0$ to $M$ after all equations for $S$ were added would have resulted in a (non-redundant) insertion in some row $i$. During back substitution, only one choice for the $i$-th row of $Z$ satisfies $\svec{h}(y)·Z = 0$. Since the $i$-th row was initialized randomly we have $\Pr[\svec{h}(y)·Z = 0 \mid \svec{h}(y) ∉ \mathrm{span}((\svec{h}(x))_{x ∈ S)}] = 2^{-r}$.
\end{proof}

We now derive an asymptotic bound on $p$ in terms of large $w$ and small $ε$.

\begin{lemma}
    \label{lem:extra-fpr-homogeneous}
    There exists a constant $C$ such that whenever $C\frac{\log w}{w} ≤ ε ≤ \frac 12$ we have $p = \exp(-Ω(εw))$.
\end{lemma}

\begin{proof}
\def\pos{\mathrm{pos}}
\def\BAD#1{\textcolor{darkgray}{\sffamily#1}}
    We may imagine that $S ⊆ \U$ and $y ∈ \U \setminus S$ are obtained from a set $S^+ ⊆ \U$ of size $n+1$ by picking $y ∈ S^+$ at random and setting $S = S^+ \setminus \{y\}$. Then $p$ is simply the expected fraction of keys in $S^+$ that are contained in some \emph{dependent set}, i.e. in some $S' ⊆ S^+$ with $\sum_{x ∈ S'} \svec{h}(x) = 0^m$. Clearly, $x$ is contained in a dependent set if and only if it is contained in a \emph{minimal} dependent set.
    % and we focus on these For the purpose of this argument we may focus on inclusion minimal dependent sets.
    Such a set $S'$ always “touches” a consecutive set of positions, i.e.\ $\pos(S') := \bigcup_{x ∈ S'} [s(x),s(x) + w -1 ]$ is an interval.

    We call an interval $I ⊆ [m]$ \emph{long} if $|I| ≥ w²$ and \emph{short} otherwise. We call it \emph{overloaded} if $S_I := \{x ∈ S^+ \mid s(x) ∈ I\}$ has size $|S_I| ≥ |I|·(1-ε/2)$.
    Finally, we call a position $i ∈ [m]$ \emph{bad} if one of the following is the case:
    \begin{enumerate}[(b1)]
        •\label{it:b1} $i$ is contained in a long overloaded interval.
        •\label{it:b2} $i ∈ \pos(S')$ for a minimal dependent set $S'$ with long non-overloaded interval $\pos(S')$.
        •\label{it:b3} $i ∈ \pos(S')$ for a minimal dependent set $S'$ with short interval $\pos(S')$.
    \end{enumerate}
    We shall now establish the following
    \[\textbf{Claim: } ∀i ∈ [m]: \Pr[i \text{ is bad}] = \exp(-Ω(εw)).\]
    For each $i ∈ [m]$ the contributions from each of the badness conditions (\BAD{b1,b2,b3}) can be bounded separately. In all cases we use our assumption $ε ≥ C\smash{\frac{\log w}{w}}$. It ensures that $\exp(-Ω(εw))$ is at most $\exp(-Ω(\log w)) = \smash{w^{-Ω(1)}}$ and can “absorb” factors of $w$ in the sense that by adapting the constant hidden in $Ω$ we have $w \exp(-Ω(εw)) = \exp(-Ω(εw))$.

    \makeatletter
    \def\smallunderbrace#1{\mathop{\vtop{\m@th\ialign{##\crcr
       $\hfil\displaystyle{#1}\hfil$\crcr
       \noalign{\kern3\p@\nointerlineskip}%
       \tiny\upbracefill\crcr\noalign{\kern3\p@}}}}\limits}
    \makeatother
    \begin{enumerate}[(b1)]
        •
        % ***TODO*** use the chernoff bound from the appendix.
        % A Chernoff bound for sums $X = \sum_{j} X_j$ of i.i.d.\ indicator random variables with $μ = \E[X]$ is
        % \begin{equation}
        %     \Pr[X ≥ (1+δ)μ] ≤ \exp(-δ²μ/3). \label{eq:chernoff}
        % \end{equation}
        Let $I ⊆ [m]$ be any interval and $X₁,…,X_{n+1}$ indicate which of the keys in $S^+$ have a starting position within $I$. For $n \gg w$ and $X := \sum_{j ∈ [n+1]} X_j$ we have
        \[ μ := \E[X] ≤ \frac{(n+1)|I|}{m-w+1} ≈ \frac{n|I|}{m-w+1} = (1-ε)|I|. \]
        Using a Chernoff bound (\cref{lem:chernoff} \textbf{(a)}), the probability for $I$ to be overloaded is (for $n \gg w$)
        \begin{align}
            \Pr[X &≥ (1-ε/2)|I|] ≤ \smash{\Pr[X ≥ (1+\smallunderbrace{ε/2}_{δ})\smallunderbrace{(1-ε)|I|}_{≥ μ}]}
            \stackrel{\text{Lem.\ref{lem:chernoff}}}{≤}
            \exp(\tfrac{-ε²(1-ε)|I|}{12}). % = \exp(-Ω(ε²w²)).
            \label{eq:chernoff-long-overloaded}
            \vphantom{\underbrace{a}_{a}}
        \end{align}
        %where the last step uses $|I| ≥ w²$, i.e.\ $I$ is long.
        The probability for $i ∈ [m]$ to be contained in a long overloaded interval is bounded by the sum of \cref{eq:chernoff-long-overloaded} over all lengths $|I| ≥ w²$ and all $|I|$ offsets that $I$ can have relative to $i$.
        The result is of order $\exp(-Ω(ε²w²))$ and hence small enough.
        %This only introduces factors that are polynomial in $1/ε$ and $w$ and hence negligible.
        • Consider a long interval $I$ that is not overloaded, i.e.\ $|I| ≥ w²$ and $|S_I| ≤ (1-ε/2)|I|$. There are at most $2^{|S_I|}$ sets $S'$ of keys with $\pos(S') = I$ and each is a dependent set with probability $2^{-|I|}$ because each of the $|I|$ positions of $I$ that $S'$ touches imposes one parity condition.

        A union bound on the probability for $I$ to support at least one dependent set is therefore $2^{-|I|}·2^{|S_I|} = 2^{-\frac{ε}{2}|I|} = \exp(-Ω(ε|I|))$.

        Similar as in (\BAD{b1}) for $i ∈ [m]$ we can sum this probability over all admissible lengths $|I| ≥ w²$ and all offsets that $i$ can have in $I$ to bound the probability that $i$ is bad due to (\BAD{b2}).
        •
        \def\Sred{S_{\mathrm{red}}}
        Let $\Sred ⊆ S^+$ be the set of redundant keys. By \cref{lem:homogeneous-construction} we have $\E[|\Sred|] = n·\exp(-Ω(εw))$.

        Now if $i$ is bad due to (\BAD{b3}) then $i ∈ \pos(S')$ for some minimal dependent set $S'$ with short $\pos(S')$. At least one key from $S'$ is redundant (regardless of the insertion order). In particular, $i$ is within short distance ($< w²$) of the starting position of a redundant key $x$. Therefore at most $|\Sred|·2w²$ positions are bad due to (\BAD{b3}), which is an  $\exp(-Ω(εw))$-fraction of all positions as desired.
    \end{enumerate}
    Simple tail bounds on the number of keys with the same starting position suffice to show the following variant of the claim:
    \[\textbf{Claim': } ∀x ∈ S^+: \Pr[s(x) \text{ is bad}] = \exp(-Ω(εw)).\]
    Now assume that the key $y∈ S^+$ we singled out is contained in a minimal dependent set $S'$.
    %Without loss of generality $S'$ is minimal.
    It follows that all of $\pos(S')$ would be bad. Indeed, either $\pos(S')$ is a short interval ($→$ \BAD{b3}) or it is long. If it is long, then it is overloaded ($→$ \BAD{b1}) or not overloaded ($→$ \BAD{b2}). In any case $s(y) ∈ \pos(S')$ would be bad.

    Therefore, the probability $p$ for $y ∈ S^+$ to be contained in a dependent set is at most the probability for $s(y)$ to be bad. This is upper-bounded by $\exp(-Ω(εw))$ using Claim'.
\end{proof}

We are now ready to prove \cref{thm:homog-ribbon}.
\begin{proof}[Proof of \cref{thm:homog-ribbon}]
  \def\SPACE{\textsc{space}}
  \def\OPT{\textsc{opt}}
  We already dealt with running times in \cref{lem:homogeneous-construction}.

  The constraint $\frac{w}{\max(r,\log w)} = 𝒪(1/ε)$ leaves us room to assume $εw > Cr$ and $εw > C\log w$ for a constant $C$ of our choosing.
  Concerning the false positive rate we obtain
  \begin{align*}
    \label{eq:upper-bound-on-p}
    p &\refrel{lem:extra-fpr-homogeneous}{Lem}{≤} \exp(-εw) ≤ \exp(-2\log(w)-r) ≤ \tfrac{1}{w²}e^{-r} ≤ ε²2^{-r}\\
    \text{and hence } \fpr &\refrel{lem:homogeneous-fpr}{Lem}{=} p + (1-p)2^{-r} ≤ p + 2^{-r} ≤ ε²2^{-r}+2^{-r} = (1+ε²)2^{-r}.
  \end{align*}
  which is close to $2^{-r}$ as desired. Concerning the space overhead, recall its definition as $\frac{\SPACE}{\OPT}-1$ where $\SPACE$ is the space usage of the filter and $\OPT = -\log₂(\fpr)n$ is the information-theoretic lower bound for filters that achieve false positive rate $ϕ$. We have:
  \begin{align*}
    \OPT &= -\log₂(\fpr)n ≥ -\log₂(\tfrac{1+ε²}{2^{r}})n = (r-\log₂(1+ε²))n ≥ (r-ε²)n\\
    \text{and }  \SPACE &= rm = r(m-w+1) + 𝒪(rw) = \tfrac{rn}{1-ε} + 𝒪(wr)\\
    \text{ which yields }
    \frac{\SPACE}{\OPT} &= \frac{r}{(1+ε)(r-ε²)} + 𝒪(\tfrac{w}{n}) ≤ \frac{1}{(1+ε)(1-ε²)} + 𝒪(\tfrac{w}{n}) ≤ 1+3ε,
  \end{align*}
  where the last step uses $ε ≤ \frac{1}{2}$.
\end{proof}

\section{Analysis of Bumped Ribbon Retrieval (BuRR)}
\label{sec:bumped-ribbon}

We now single out one variant of BuRR and analyze it fully, thereby proving \cref{thm:bumped-ribbon}, restated here for convenience. The analysis could undoubtably be extended to cover other variants of BuRR (see \cref{s:designspace}), but in the interest of a cleaner presentation we will not do so.
\bumpedRibbonTheorem*
% \stwa{The following paragraphs includes duplications, no time to fix:} % stwa: its fine, I think.
Recall the idea illustrated in \cref{fig:burr-idea} \textbf{(b)}: We use $m < n$, making the data structure \emph{overloaded}. This ensures that the ribbon diagonal $d$ rarely hits the bottom ribbon border and \Obs{1}/\cref{lem:no-position-empty} suggests that almost all slots in $M$ can be utilized. An immediate problem is that $d$ would necessarily hit the right ribbon border in at least $n-m$ places, causing at least $n-m$ insertions to fail.
We deal with this by removing contiguous ranges of keys in strategic positions such that without the corresponding rows, $d$ never hits the right ribbon border.
A small amount of “metadata” indicates the ranges of removed keys. These keys are \emph{bumped} to a fallback retrieval data structure. Many variants of this approach are possible, see \cref{s:designspace}.

\begin{algorithm}[h]
    \def\bumped{S_{\mathrm{bumped}}}
    \def\Dbumped{D_{\mathrm{bumped}}}
    \def\insertAll{\mathrm{insertAll}}
    \def\meta{\mathrm{meta}}
    \def\buckets{{\#}\mathrm{buckets}}
    \def\bumpNothing{\textsc{bumpNothing}}
    \def\bumpHead{\textsc{bumpHead}}
    \def\bumpAll{\textsc{bumpAll}}
    \Input{$f : S → \{0,1\}^r$ for some $S ⊆ \U$ of size $n$.}
    \Parameters{$w ∈ ℕ$.}
    $m ← n/(1+\frac{C\log w}{8w})+w-1$; allocate system $M$ of size $m$\;
    pick hash functions $s : \U → [m-w+1],\ c : \U → \{0,1\}^w$\;
    $b ← \frac{w²}{C \log w}, \ \ \buckets ← \frac{m-w+1}{b}$ \tcp{bucket size \& number of buckets}

    \For(\tcp*[h]{partition}){$j ∈ [\buckets]$}{
        $B_j ← \{x ∈ S \mid \lceil s(x)/b\rceil = j\}$\;
        $H_j ← \{x ∈ B_j \mid s(x) - (j-1)b ≤ \frac{3}{8}w\}$ \tcp{head}
        $T_j = B_j \setminus H_j$ \tcp{tail}
    }
    $\bumped ← ∅$\;
    \For{$j ∈ [\buckets]$}{
        \tcp{insertAll(X): attempt \cref{algo:ribbon-insertion} for all $x ∈ X$, roll back on failure}
        \uIf{$\insertAll(T_j)$}{
            \uIf{$\insertAll(H_j)$}{
                $\meta[j] ← \bumpNothing$\;
            }\Else{
                $\bumped ← \bumped ∪ H_j$\;
                $\meta[j] ← \bumpHead$\;
            }
        }\Else{
            $\bumped ← \bumped ∪ B_j$\;
            $\meta[j] ← \bumpAll$\;
        }
    }
    $Z ← 0^{m × r}$\;
    \For{$i = m$ \DownTo $1$}{
        $Z_i ← M.c[i] · Z_{i..i+w-1}$ \tcp{back substitution}
    }
    $\Dbumped ← \mathrm{construct}(\bumped)$ \tcp{recursive, unless base case reached}
    \Return $D = (s,c,Z,\meta,\Dbumped)$\;
    \caption[fragile]{The construction algorithm of BuRR as analyzed in \cref{sec:bumped-ribbon}.}
    \label{algo:burr}
\end{algorithm}

\subsection{Proof of \texorpdfstring{\cref{thm:bumped-ribbon}}{the Theorem on BuRR}}

% sw: Der Kommentar ist richtig und erhellend. Ich habe nur keinen guten Ort gefunden ihn unterzubringen und will mich da nun nicht mehr hineindenken.
% \madi{Man könnte noch zu erklären versuchen, wieso man
% eine Beschränkung für die Bucketgröße nach oben erntet.

% Für mich hat es etwas damit zu tun,
% dass Po($\lambda$) mit hoher Wahrscheinlichkeit Abweichungen von höchstens
% $\sqrt{\lambda \log(\lambda)}$ vom Erwartungswert $\lambda$ hat, für große $\lambda$.
% Aber kleiner sind die Abweichungen eben auch nur mit nicht so
% überwältigender Wahrscheinlichkeit.

% Mit $\lambda = b + w/8$ und $b=\Theta(w^2/\log w)$ landet man bei Deinen Zahlen.
% }%
Consider \cref{algo:burr}. In what follows, $C$ refers to a constant from \cref{lem:chernoff}. For $n$ keys, $M$ is given $m = n/(1+\smash{\frac{C\log w}{8w}})+w-1$ rows\footnote{We ignore rounding issues for a clearer presentation and assume that $w$ is large. This causes a certain disconnect to practical application where concrete values like $w = 32$ are used.}. The $m-w+1$ possible starting positions are partitioned into \emph{buckets} of size $b = \smash{\frac{w²}{C\log w}}$. The first $\frac{3}{8}w$ slots of a bucket are called its \emph{head}, and the larger rest is called its \emph{tail}. Keys implicitly belong to (the head or tail of) a bucket according to their starting position.
For each bucket the algorithm has three choices:
\begin{enumerate}
	• No keys belonging to the bucket are bumped.
	• The keys belonging to the head of the bucket are bumped.
	• All keys of the bucket are bumped.
\end{enumerate}
These choices are made greedily as follows. Buckets are handled from left to right. For each bucket, we first try to insert all keys belonging to the bucket's tail. If at least one insertion fails, then the successful insertions are undone and the entire bucket is bumped, i.e.\ Option 3 is used. Otherwise, we also try to insert the keys belonging to the bucket's head. If at least one insertion fails, all insertions of head keys are undone and we choose Option 2, otherwise, we choose Option 1.
The main ingredient in the analysis of this algorithm is the following lemma, proved later in this section.
\begin{lemma}
	\label{lem:bumped-ribbon}
	The expected fraction of empty slots in $M$ is $\O(w^{-3})$.
\end{lemma}
If the fraction of empty slots is significantly higher than expected, we simply restart the construction with new hash functions until satisfactory (this is not reflected in \cref{algo:burr}). After back substitution, we obtain a solution vector of $mr$ bits. Additionally, we need to store the choices we made, which takes $\lceil\log₂3\rceil = 2$ bits of metadata per bucket.
Given that $|P_M| = m·(1-\O(w^{-3}))$ keys are taken care of, this suggests a space overhead of
\begin{align*}
	ε &= \frac{\textsc{space}}{\textsc{opt}} - 1 ≤ \frac{mr+2\frac{m}{b}}{|P_M|·r}-1 ≤ \frac{1+2\frac{1}{rb}}{1-\O(w^{-3})}-1\\
	& = (1+\tfrac{2}{rb})(1+\O(w^{-3}))-1 = \O(\tfrac{1}{rb})+\O(w^{-3}) = \O(\tfrac{\log w}{rw²})+\O(w^{-3}) = \O(\tfrac{\log w}{rw²}).
\end{align*}
The last step uses the assumption $r = 𝒪(w)$. The trivial bound in \cref{lem:time-bound-insertions} implies that $\O(bw)$ row operations are performed during the \emph{successful} insertions in a bucket. There can be at most one failed insertion for each bucket which takes $\O(b)$ row operations since insertions cannot extend past the next (still empty) bucket. Since $w = \O(\log n)$ bits fit into a word of a word RAM, these row operations take $\O(\frac{m}{b}·(bw+b)) = \O(nw)$ time in total.

A query of a non-bumped key involves computing the product of the $w$-bit vector $c(x)$ and a block $Z(x)$ of $w × r$ bits from the solution matrix $Z ∈ \{0,1\}^{m × r}$. The $wr$ bit operations can be carried out in $\O(1+\frac{wr}{\log n})$ steps on a word RAM with word size $Ω(\log n)$. A complication is that if $w,r ∈ ω(1) ∩ o(\log n)$ then we are forced us to handle several rows of $Z(x)$ in parallel (xor-ing a $c(x)$-controlled selection) or several columns of $Z(x)$ in parallel (bitwise \textsc{and} with $c(x)$ and popcount). Numbers ``much bigger than $1$ and much smaller than $\log n$'' are a somewhat academic concern, so we believe an academic resolution (not reflected in our implementation\footnote{Though AVX512 instructions such as VPOPCNTDQ may benefit a corresponding niche.}) is sufficient: We resort to the standard techniques of tabulating the results of a suitable set of vector matrix products.
% [[Old query complexity discussion]]
% A query of a non-bumped key involves computing the product of a $w$-bit vector and a block of $w × r$ bits from the solution matrix $Z ∈ \{0,1\}^{m × r}$. For appropriate memory layouts\footnote{To exploit both cache efficiency and bit-parallelism we recommend storing $w × r$ blocks of $Z$ contiguously, with each block stored column-wise.} of $Z$, we obtain a query time of $\O(1+\frac{rw}{\log n})$
%\pesa{We have to be more careful with the query time. As I see it, interleaved storage gives query time $\Oh{r\lceil w/\log n rceil}$, storing and processing individual entries of $r$ bits gives query time $\Oh{w\lceil r/\log n rceil}$. But, e.g. for $r=w=\sqrt{\log n}$ none of these approaches yields the claimed $\Oh{1}$ query time. The easiest way out would be to postulate a horizontal xor-operation that can handle $\Om{\log(n)/r}$ operands in parallel. I am not aware of such an instruction. One can do it in $\log w$ steps though. A better alternative might be popcount vector instructions such as VPOPCNTDQ in AVX512}.
Back substitution has the same complexity as $n$ queries and therefore takes $\O(n(1+\frac{wr}{\log n})) = \O(nw)$ time.

To complete the construction, we still have to deal with the $n-|P_M| = n-m(1-\O(w^{-3})) = \O(\frac{n\log w}{w})$ bumped keys. A query can easily identify from the metadata whether a key is bumped, so all we need is another retrieval data structure that is consulted in this case. We can recursively use bumped ribbon retrieval again. However, to avoid compromising worst-case query time we only do this for four levels. Let $S^{(4)}$ be the set of keys bumped four times. We have $|S^{(4)}| = \O(n\frac{\log⁴w}{w⁴}) = \O(n\frac{\log w}{rw²})$ and we can afford to
% stwa: ok.
% \pesa{slightly reformulated ``construct a minimum perfect hash function (MPHF) for $S^{(4)}$, e.g.\ using the approach from \cite{BPZ:Practical:2013}, which takes $\O(|S^{(4)}|)$ bits, linear construction time and $\O(1)$ worst-case query time. Given an MPHF for $S^{(4)}$, it is trivial to construct a retrieval data structure for $S^{(4)}$ in space $|S^{(4)}|r$, i.e.\ without additional overhead. With this in mind, it should be clear that query time, expected construction time and overhead, as claimed for the first level of the BuRR data structure, carry over to the complete construction.
% We remark that using standard ribbon as a base case is more
% natural. The technical reason we did not do so here are cases where
% the standard ribbon query time of $𝒪(r)$ raises the worst case query
% time of the complete construction (though not the expected query
% time).''}
store $S^{(4)}$ using a retrieval data structure with constant
overhead, linear construction time and $\O(1)$ worst-case query time,
e.g., using minimum perfect hash functions
\cite{BPZ:Practical:2013}.%
\footnote{Our implementation is optimized for $w=\Om{\log n}$ and can simply use ribbon retrieval with an appropriate $\eps>0$.}

%%% Local Variables:
%%% mode: latex
%%% TeX-master: "main"
%%% End:

% LocalWords:  undoubtably

%!TEX root=./main.tex

% \subsection
\subsection{Proof of Lemma \ref{lem:bumped-ribbon}}
\label{sec:bumped-lemma-proof}
We give an induction-like argument showing that “most” buckets satisfy two properties:
\def\P#1{\textbf{\upshape(P#1)}}
\begin{description}
	•[\P1] All slots of the bucket are filled.
	•[\P2] The height $h_i := s_i - d_i$ of the ribbon diagonal over the lower ribbon border at the last position $i$ of the bucket satisfies $h_i ≥ \frac{w}{4}$.\footnote{The key set underlying the definitions of $s_i$ and $d_i$ excludes the bumped keys.}
\end{description}
\begin{claim}
	\label{claim:bucket-induction}
	If \P2 holds for a bucket $B₀$ then \P1 and \P2 hold for the following bucket $B₁$ with probability $1-w^{-3}$.
\end{claim}
\begin{proof}
	Let $i₀$ and $i₁$ be the last positions of buckets $B₀$ and $B₁$, respectively. By \P2 for $B₀$ we have $h_{i₀} ≥ \frac{w}{4}$.
	\begin{description}
		•[Case 1: $\bm{h_{i₀} < \frac 58w}$.] We claim that with probability $1-\O(w^{-3})$ all keys belonging to $B₁$ (head and tail) can be inserted and \P1 and \P2 are fulfilled afterwards. The situation is illustrated in \cref{fig:bumping-ribbon-proof} on the left.

		The dashed black lines show the expected trajectories of the ribbon borders for bucket $B₁$. The lower expected border travels in a straight line from $(s_{i₀},i₀)$ to the point that is $b = i₁ - i₀$ positions to the right and $\E[|\{x ∈ S \mid s(x) ∈ B₁\}|] = \frac{n}{m-w+1}b = (1+\frac{C\log w}{8w})b = b+\frac{w}{8}$ positions below. The actual position of the border randomly fluctuates around the expectation. At each point the (vertical) deviation exceeds $\frac{w}{8}$ with probability at most $\O(w^{-5})$ by \cref{lem:chernoff} \textbf{(b)}. A union bound shows that it is at most $\frac{w}{8}$ \emph{everywhere in the bucket} and thus within the region shaded red with probability at least $1-\O(w^{-3})$. The region shaded yellow represents a “safety distance” of another $\frac w8$ that we wish to keep from the ribbon border. Finally, the blue line is a perfect diagonal starting from $(d_{i₀},i₀)$, which we claim the ribbon diagonal also follows. Due to $s_{i₀} - d_{i₀} = h_{i₀} ≥ \frac{w}{4}$ the diagonal does not intersect the lower yellow region. To see that it does not intersect the right yellow region, note that is passes through position $(s_{i₀},i₀+h_{i₀})$ which, by this case's assumption is at least $\frac{3}{8}w$ positions left of the right ribbon border. This is sufficient to compensate for the width of the red region ($\frac{1}{8}w$), the width of the yellow region ($\frac{1}{8}w$) as well as the difference in slope due to overloading which accounts for a relative vertical shift of another $\frac{1}{8}w$.

		Now our arguments nicely interlock to show that the ribbon diagonal $\{(d_i, i) \mid i ∈ B₁\}$ follows this designated path: As long as $d$ stays away from the yellow region, it remains $\frac{w}{8}$ positions away from the lower and right ribbon borders so each slot remains empty with probability at most $2^{-w/8}$ by \cref{lem:no-position-empty} and each insertion fails with probability at most $2^{-w/8}$ by \cref{lem:no-failed-insertion}. Conversely, as long as no insertion fails and no slot remains empty, $d$ proceeds along a diagonal path. % (see \cref{sec:ribbon-diagonal}).

		Two small caveats concern the area outside of the rectangle. We do not know the right ribbon border above row $s_{i₀}$; however, those rows correspond to keys from the previous bucket and would have been bumped if their insertions failed. We also do not know the lower border to the right of $i₁$; however here \cref{lem:no-failed-insertion} \textbf{(b)} helps: We may use that the vertical distance of the ribbon diagonal to the top ribbon border is at most $\frac{w}{8}$ to conclude that the keys with the last $w-1$ starting positions are also inserted successfully with probability $2^{-w/8}$.

		This establishes \P1 with probability $1-\O(w^{-3})$. Then \P2 follows easily: The extreme case is when both $h_{i₀} = \frac{w}{4}$ and $|\{x ∈ S \mid s(x) ∈ B₁\}| = b$ take the minimum permitted values. In that case we have $h_{i₁} = \frac{w}{4}$, so in general $h_{i₁} ≥ \frac{w}{4}$ follows.
		\begin{figure}
                  %\ \hspace{-1em}
			\includegraphics[scale=0.8,page=1]{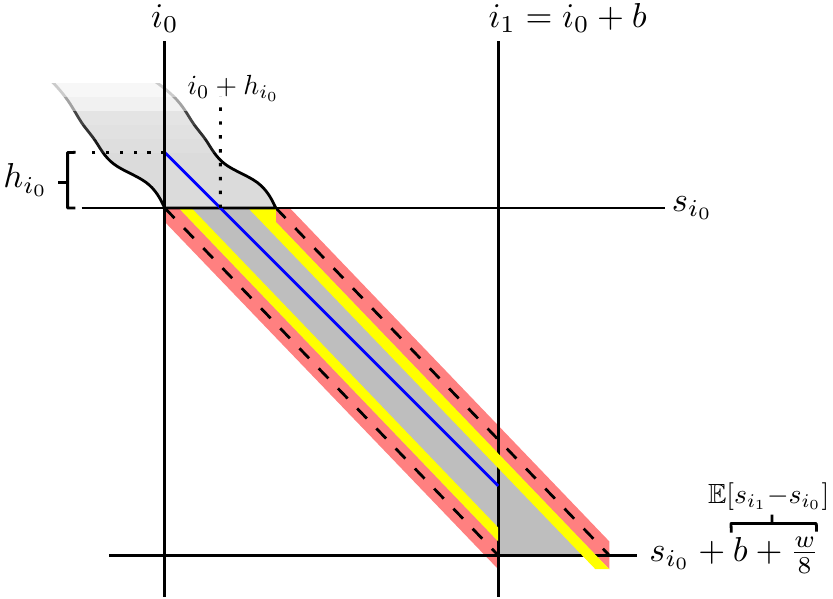}
			\hspace{-2em}
			\includegraphics[scale=0.8,page=2]{figures/bumped-ribbon-proof.pdf}
			\caption{Situation in Cases 1 (left) of 2 (right) of the proof of \cref{claim:bucket-induction}.}
			\label{fig:bumping-ribbon-proof}
		\end{figure}
		•[Case 2: $\bm{h_{i₀} ≥ \frac 58w}$.] We claim that all keys belonging to the tail of $B₁$ can be inserted and that afterwards \P1 and \P2 are fulfilled with probability $1-\O(w^{-3})$. In case the keys in the head of $B₁$ can also be successfully inserted this cannot hurt \P1 or \P2 because the number of filled slots and the height could only increase due to the additional keys.\footnote{Note that our analysis suggests that $B₁$ is already full after the tail-keys are inserted, which means that the head keys can only be inserted if they all “overflow” into the next bucket.}
		% [this argument and the underlying notions need polishing: Need to explicitely say for which set of keys we are looking at the ribbon diagonal model.]

		The situation is illustrated in \cref{fig:bumping-ribbon-proof} on the right. We only consider the keys in the tail of $B₁$ which starts at position $i₀'+1$ where $i₀' = i₀+\frac{3}{8}w$.
		Note that for $i ∈ [i₀,i₀']$ the ribbon diagonal $(d_i,i)$ follows an ideal diagonal trajectory with probability $1-\O(2^{-Ω(w)})$ since keys from $B₀$ are successfully inserted and the distance to the bottom border is at least $\frac{1}{4}w$. This implies that all slots in the head of $B₁$ are filled by keys from $B₀$ and $h_{i₀'} = h_{i₀} - \frac{3}{8}w$.
		Since $h_{i₀} ∈ [\frac{5}{8}w,w]$ we have $h_{i₀'} ∈ [\frac{1}{4}w,\frac{5}{8}w]$, which allows us to reason as in Case $1$ to show that all slots in the tail of bucket $B₁$ are filled and $h_{i₁} ≥ \frac{w}{4}$ with probability $1-\O(w^{-3})$.\qedhere
	\end{description}
\end{proof}

\subparagraph{Handling failures and the first bucket.}

The following claim is needed to deal with the rare cases where \cref{claim:bucket-induction} does not apply.

\begin{claim}
	\label{claim:bucket-recovery}
	Assume $B₁$ is either the first bucket, or preceeded by a bucket $B₀$ for which \P2 does not hold. Then with probability $1-\O(w^{-3})$ all keys of $B₁$ (head and tail) are successfully inserted.
\end{claim}
\begin{proof}
    The ribbon diagonal $d$ starts at a height $h_{i₀} < \frac{w}{4}$, which is lower than desired, and might hit the lower ribbon border within $B₁$. However, $d$ avoids the right ribbon border, because (recycling ideas from Case 1 of \cref{claim:bucket-induction}) $d$ would have to pierce the diagonal starting at the desired height $\frac{w}{4}$ first and would afterwards stay on that diagonal with probability $1-\O(w^{-3})$.
    This allows for some slots in $B₁$ to remain empty but implies all keys are successfully inserted with probability $1-\O(w^{-3})$.
\end{proof}
We now classify the buckets.
Consider a bucket $B₁$. Note that either \cref{claim:bucket-induction} or \cref{claim:bucket-recovery} is applicable. If the corresponding event with probability $1-\O(w^{-3})$ fails to occur or if $B₁$ contains less than $b$ keys (this happens with probability $\O(w^{-5})$), then $B₁$ is a \emph{bad bucket}. Otherwise, if \cref{claim:bucket-induction} applies, then $B₁$ is a \emph{good bucket} and if \cref{claim:bucket-recovery} applies then $B₁$ is a \emph{recovery bucket}.
% On the one hand, there are \emph{bad buckets} for which one of the considered events with probability $1-\O(w^{-3})$ fails to occur.\madi{too vague. Which event?} On the other hand there are \emph{good buckets} and \emph{recovery buckets} to which \cref{claim:bucket-induction} and \cref{claim:bucket-recovery} apply, respectively, and for which the corresponding high probability events occur.
A \emph{recovery sequence} is a maximal contiguous sequence of recovery buckets. Since such a sequence cannot be preceded by a good bucket, the number of recovery sequences is at most the number of bad buckets plus $1$ (the first bucket is always a recovery bucket or bad). Only bad buckets contain fewer than $b$ keys, so a recovery sequence of $k$ buckets contains at least $kb$ keys, all of which are inserted successfully by \cref{claim:bucket-recovery}. At most $w-1$ of these insertions fill slots after the sequence so there are at most $w-1$ empty slots within a recovery sequence. With $x$ denoting the number of bad buckets, the number $m-|P_M|$ of empty slots in total is
\[m-|P_M| ≤ xb + (x+1)(w-1) + w-1\]
where the last $w-1$ accounts for slots $[m-w+1,m]$ that do not belong to any bucket. The dominating term is $xb$ so using $\E[x] = \O(\frac{m}{b}w^{-3})$ we obtain $\E[\frac{m-|P_M|}{m}] = \O(w^{-3})$, which completes the proof of \cref{lem:bumped-ribbon}.

%%% Local Variables:
%%% mode: latex
%%% TeX-master: "main"
%%% End:

%----------------------------------------------------------------------
\section{The Design Space of BuRR}\label{s:designspace}
There is a large
design space for implementing BuRR.  We outline it in some breadth
because there was a fruitful back-and-forth between different
approaches and their analysis, i.e., different approaches gradually
improved the final theoretical results while insights gained by the
analysis helped to navigate to simple and efficient design points.
The description of the design space helps to explain some of the
gained insights and might also show directions for future improvements
of BuRR.  To also accommodate more ``hasty'' readers, we nevertheless
put emphasis on the the main variant of BuRR analyzed
\cref{sec:bumped-ribbon} and also move some details to appendices.
We first introduce a simple generic approach and discuss
concrete instantiations and refinements in separate subsections.  In
\cref{s:moredesign}, we describe further details.

As all layers of BuRR function in the same way, we need only explain the
construction process of a single layer.  The BuRR construction process makes the
ribbon retrieval approach of \cref{s:ribbonRetrieval} more dynamic by
\emph{bumping} ranges of keys when insertion of a row into the linear system
fails by causing a linear dependence.  Bumping is effected by subdividing the
table for the current layer into \emph{buckets} of size $b$.
More concretely, bucket $B$ contains metadata for
keys $x$ with $s(x)\in Bb+1..(B+1)b$.  We also say that $x$ is \emph{allocated}
to bucket $B$ even though retrieving~$x$ can also involve subsequent buckets.
In \cref{sec:bumped-ribbon}, we showed that it basically suffices to adaptively
remove a fraction of the keys from buckets with high \emph{load} to make the
equation system solvable, i.e., to make all remaining keys retrievable from the
current layer.  The structure of the linear system easily absorbs most variance
within and between buckets but bigger fluctuations are more efficiently handled
with bumping.

\paragraph*{Construction.}
We first sort the keys by the buckets they are allocated
to.  For simplicity, we set each key's leading coefficient $\svec{h}(x)[s(x)]$ to 1.
As the starting positions are distributed randomly, this is not an issue.
The ribbon solving approach % of \cref{alg:ribbonSolve}
is adapted to build the row echelon form bucket by bucket from left to right (see
\cref{ss:globalOrder} for a discussion of variants).  Consider now the solution
process for a bucket $B$.  Within $B$, we place the keys in some order that
depends on the concrete strategy; see \cref{ss:insertionOrder}.  One useful
order is from right to left based on~$s(x)$.  We store metadata that indicates
one or several \emph{groups} of keys allocated to the bucket that are
bumped. These groups correspond to ranges in the placement order, not
necessarily from~$s(x)$. Which groups can be represented depends on the metadata
compression strategy, which we discuss in Sections \ref{ss:thresholdCompression} and
\ref{ss:groupCompression}.  For example, in the right-to-left order mentioned
above, it makes sense to bump keys $x$ with $s(x)\in Bb+1..Bb+t$ for some
threshold $t$, i.e., a leftmost range of slots in a bucket. This makes sense
because that part of the bucket is already occupied by keys placed during
construction of the previous bucket (see also \cref{fig:fillrtlnew}).

If placement fails for one key in a group, all keys in that group are bumped
together.  Such transactional grouping of insertions is possible by recording
offsets of rows inserted for the current group, and clearing those rows if
reverting is needed. This implies that we need to record which rows were used by
keys of the current group so that we can revert their insertion if needed.
%% This implies that the \Id{placed}-values for these keys
%% should be buffered and only committed to
%% the global \Id{placed}-array when all placements for column $j$ have
%% succeeded.\pesa{check with PeterD and LHS whether this description needs to be adapted to the actual implementations.}
%% \lhs{the index of the row in which the key was placed is buffered and reset
%%   to 0 if placement of another key $y$ with $s(y) = j$ fails}
%% \pedi{Yes, I call this the ``incremental'' capability in the VLDB paper:
%% ``The insertion phase is incremental because we can easily undo a set of
%% most recent successful insertions: simply remove the rows from M that were
%% added.'' Implication: keep a list of insertion positions of sufficient
%% length for your undoing needs.}\pesa{Anyone want to propose a concrete reformulation?}
%% \lhs{suggestion: This implies that we need to store which rows were used by keys
%%   of the current group so that we can revert their insertion if needed.}

Keys bumped during the construction of a layer are recorded and passed into the
construction process of the next layer. Note that additional or independent
hashing erases any relationships among keys that led to bumping.  In the last
layer, we allocate enough space to ensure that no keys are bumped, as in the
standard ribbon approach.

\myparagraph{Querying.}
At query time, if we try to retrieve a key $x$ from a bucket $B$, we check
whether $x$'s position in the insertion order indicates that $x$ is in a bumped
group. If not, we can retrieve $x$ from the current layer, otherwise we have to
go on to query the next layer.

\myparagraph{Overloading.}
% \pesa{reformulated to mitigate repetitiousness}
The tuning parameter $\eps=1-m/n$ is very important
for space efficiency.  While other linear algebra based retrieval data
structures need $\eps>0$ to work, a decisive feature of BuRR is that
negative $\eps$  almost eliminates empty table slots by avoiding
underloaded ranges of columns.
%
%% Recall
%% that for basic ribbon retrieval, $\eps$ must be positive in order to make the
%% equation system solvable.  This is different for BuRR. Since we can bump keys
%% on-the-fly, we are free to choose $\eps$.  In particular, it turns out that a
%% \emph{negative} value allows us to radically reduce space overhead. The reason
%% is that in ribbon retrieval, underloaded regions of the table are the main
%% reason for empty slots.  We can almost eliminate underloaded regions by
%% systematically overloading the table.  It turns out that small overloads are
%% already sufficient to achieve this effect.

We discuss further aspects of the design space of BuRR in additional
subsections.  By default, bucket construction is \emph{greedy}, i.e., proceeds
as far as possible. \Cref{ss:cautious} presents a \emph{cautious} variant that
might have advantages in some cases.  \Cref{ss:globalOrder} justifies our choice
to construct buckets from left to right. \Cref{ss:sparse} discusses how more
sparse bit patterns can improve performance.  Construction can be parallelized
by bumping ranges of $w$ consecutive table slots.  This separates the equation
system into independent blocks; see \cref{ss:parallel}.  In that section, we
also explain external memory construction that with high probability needs only
a single scan of the input and otherwise scanning and integer sorting of simple
objects. The computed table and metadata can be represented in various forms
that we discuss in \cref{representation}. In particular, \emph{interleaved
  storage} allows to efficiently retrieve $f(x)$ one bit at a time, which is
advantageous when using BuRR as an AMQ. We can also reduce cache faults by
storing metadata together with the actual table entries.

A very interesting variant of BuRR is \emph{bump-once ribbon retrieval
  (Bu$^1$RR)} described in \cref{ss:twoLayer} that guarantees that each key can
be found in one of two layers.

%---------------------------------------------------------------------
\subsection{Threshold-Based Bumping}\label{ss:thresholdCompression}

Recall that BuRR places the keys one bucket $B$ at a time and within $B$
according to some ordering -- say from right to left defined by a position in
$1..b$.  A very simple way to represent metadata is to store a single threshold
$j_B$ that remembers (or approximates) the first time this insertion process
fails for $B$ ($j_B=0$ if insertion does not fail). During a query, when
retrieving a key $x$ that has position $j$ in the insertion process, $x$ is
bumped if $j\geq j_B$.  We need $\log b$ bits if we conflate $b-1$ and $b$ onto
$b-1$ and bump the entire bucket.
It turns out that for small $r$ (few retrieved bits), the space overhead for
this threshold dominates the overhead for empty slots. Thus, we now discuss
several ways to further reduce this space.

\myparagraph{2-bit Thresholds and $c$-bit Thresholds.} % Two bits of
metadata implies that we have to choose between four possible
thresholds values.  It makes sense to have support for threshold
values for $j_B=0$ (no bumping needed in this bucket) and $j_B=b$
(bump entire bucket). The latter is a convenient way to ensure
correctness in a nonprobabilistic way, thus obviating restarts with a
fresh hash function.  The former threshold makes sense because the
effect of obliviously bumping some keys from \emph{every} bucket could
also be achieved by choosing a larger value of $\eps$. This leaves two
threshold values $\ell$, $u$, as tuning parameters.
% \pesa{reformulated rest of paragraph. Lorenz, please read}
The experiments in \cref{s:exp} use linear formulas of the form $\ceil{(c_1-c_2\eps)b}$ for $\ell$
and $u$ with the values of $c_1$ and $c_2$ dependent on $w$,
but it also turns out that $u=2\ell$ is a good
choice so that the cases $0$, $\ell$, and $2\ell$ can be decoded with
a simple linear formula and only the case $j_B=b$ needs special case
treatment.  Moreover, this approach can be generalized to $c$-bit
thresholds, where we use $2^c-1$ equally spaced thresholds starting at
$0$ plus the threshold $b$.
% \pesa{reinstated $c$-bit because otherwise
%   the paragraph title makes no sense. c-bit is also potentially useful
%   for sparse coefficients and large $r$?}

\myparagraph{$1^+$-Bit Thresholds.} %
The experiments performed for 2-bit thresholds indicated that actually choosing
$\ell=u$ performs quite well.  Indeed, the analysis in \cref{sec:bumped-ribbon} takes up this
scheme.  Moreover, both experiments and theory show that the threshold $b$ (bump
entire bucket) occurs only rarely.  Hence, we considered compression schemes
that store only a single bit most of the time, using some additional storage to
store larger bumping thresholds.  We slightly deviate from the theoretical
setting by allowing arbitrary larger thresholds in order to reduce the space
incurred by empty buckets.  Thus, we considered a variant where the threshold
values $0$ (bump nothing), $t$ (bump something), and also values $t+1..b$ (bump
a lot) are possible but where the latter (rare) cases are stored in a separate
small hash table $H^+$ whose keys are bucket indices.

%% Our experiments indicate that a good threshold value is obtained when the
%% expected number of bumped keys is $-2\eps b$. Roughly, this means that 50\,\%
%% of the buckets will bump something. In other words, about half of the bumping
%% bits are one giving this bit maximum information content.

Compared to 2-bit thresholds, we get a space-time trade-off, however, because
accessing the exception table $H^+$ costs additional time (even if it is very
small and will often fit into the cache).  Thus, a further refinement of \onebit
thresholds is to partition the buckets into ranges of size $b^+$ and to store
one bit for each such range to indicate whether any bucket in that range has an
entry in $H^+$.

%---------------------------------------------------------------------
\subsection{Sparse Bit Patterns}\label{ss:sparse}
A query to a BuRR data structure as described so far needs to process about
$rw/2$ bits of data to produce $r$ bits of output. Despite the opportunity for
word parallelism, this can be a considerable cost.  It turns out that BuRR also
works with significantly more sparse bit patterns.  We can split $w$ bits into
$k$ groups and set just one randomly chosen bit per group.
% \pesa{dropped the construction speed argument since, as PD noted, we should
% probably compare with the speed of dense coefficients at smaller $w$. Since
% the effects are small either way, I view it as not very interesting to further
% investigate this question.}
The downside of sparse patterns is that they incur more linear
dependencies. This will induce more bumped keys and possibly more empty slots.
Our experiments indicate that the compromise is quite interesting.  For example
for $w=64$ and $k=8$ we can reduce the expected number of 1-bits by a factor of
four and observe faster queries for large $r$ at the cost of a slight increase
of space overhead. \Cref{fig:heatmap} indicates that our implementation
of sparse coefficient BuRR is indeed a good choice for $r\in\{8,16\}$.

%---------------------------------------------------------------------
% \subsection{Parallelization and Memory Hierarchies}\label{ss:parallelShort}

% In \cref{ss:parallel}, we explain how to construct BuRRs in parallel by bumping
% appropriate keys to subdivide the equation system into independent parts. We
% also explain how to do construction in external memory.  This is mostly
% straight-forward as the construction steps correspond to sorting and scanning
% operations. However, an interesting aspect is that hashing keys multiple times
% can be avoided by injectively mapping them to \emph{master hash codes}.
% \pedi{This seems like the most obvious (least interesting) aspect of a well
% engineered implementation, though retrieval can only make such assurance
% w.h.p. I can't speak to what ESA thinks is interesting.}\pesa{I do not find it
% obvious. Also it needs a bit of a theoretical justification both in terms of
% machine model and analysis of failure probability.}

%---------------------------------------------------------------------
\subsection{Parallelization and Memory Hierarchies}\label{ss:parallel}
As a static data structure, queries are trivial to parallelize on a
shared-memory  system.
%% \footnote{Of course one should avoid the pitfall of allocating the
%%   entire data structure to a single NUMA node.\lhs{seems out of place
%%     here?}}
Parallelizing construction can use (hard) sharding, i.e.,
subdividing the data structure into pieces that can be built
independently. An interesting property of BuRR is that sharding can be
done transparently in a way that does not affect query implementation
and thus avoids performance overheads for queries.  To subdivide the
data structure, we set the bumping information in such a way that each
piece is delimited by empty ranges of at least $w$ columns in the
constraint matrix.  This has the effect that the equation system can
be solved independently in parallel for ranges of columns separated by
such gaps. With the parametrization we choose for BuRR, the fraction
of additionally bumped keys will be negligible as long as $n\gg pw^2$
where $p$ is the number of threads.

For large BuRRs that fill a significant part of main memory, space
consumption during construction can be a major limitation.  Once more,
sharding is a simple way out -- by constructing the data structure one
shard at a time, the temporary space overhead for construction is only
proportional to the size of a shard rather than to the size of the
entire data structure.

An alternative is to consider a ``proper'' external memory algorithm.
The input is a file $F$ of $n$ key--value pairs. The output is a file containing
the layers of the BuRR data structure. A difficulty here is that
keys could be much larger than $\log n$ bits and that random
accesses to keys is much more expensive than just scanning or
sorting. We therefore outline an algorithm that, with high probability,
reads the key file only once. The trick is to first
substitute keys by \emph{master hash codes} (MHCs) with $c\log n$ bits for a constant $c>2$.
Using the well-known calculations for the birthday problem, this size
suffices that the MHCs are unique with high probability.
If a collision should occur, the
entire construction is restarted.%
\footnote{For use in AMQs, restarts are not needed since duplicate
  MHCs lead to identical equations that will be ignored as redundant
  by the ribbon solver.% \pesa{Only this nice feature led me to add the
    % seemingly uninteresting case of redundant equations to the ribbon
    % solver.}
}  Otherwise, the keys are never accessed again -- all
further hash bits are obtained by applying secondary hash functions to
the MHC.%
\footnote{We use
  a (very fast) linear congruential mapping
  \cite{Knuth:Vol2:81,Ribbon:Arxiv:2021} that, with some care, (multiplier
  congruent 1 modulo 4 and an odd summand) even defines a bijective
  mapping \cite{Knuth:Vol2:81}.  We also tried linear combinations of two
  parts of the MHC \cite{kirsch2008less} which did not work well
  however for a 64 bit MHC.}

Construction then begins by scanning the input file $F$ and producing
a stream $F_1$ of pairs (MHC, function value)%
\footnote{When used as an AMQ,
the function value need not be stored since it can be derived from the
MHC.}.  Construction at layer $i$ amounts to reading a stream $F_i$ of
MHC--value pairs.  This stream $F_i$ is then sorted by a bucket ID that
is derived from the MHCs.  A collision is detected when two pairs with
the same MHC show up.  These are easy to detect since they will
be sorted to the same bucket and the same column within that bucket.
Constructing the row echelon form (REM) then amounts to simultaneously
scanning $F_i$, the REM and the right-hand side. At no time during this
process do we need to keep more than two buckets worth of data in
main memory.  Bumped MHC--value pairs are output as a stream $F_{i+1}$ that is
the input for construction of the next layer.  Back-substitution
amounts to a backward scan of the REM and the right-hand
side -- producing the table for layer $i$ as an output.

Overall, the I/O complexity is the I/Os for scanning $n$ keys plus
sorting $\Oh{n}$ items consisting of a constant number of machine
words ($\Oh{\log n}$ bits).  The fact that there are multiple layers
contributes only a small constant factor to the sorting term since
these layers are shrinking geometrically with a rather large
shrinking factor.

%---------------------------------------------------------------------
\subsection{\texorpdfstring{Bu$^1$RR}{Bu1RR}: Accessing Only Two Layers in the Worst Case}\label{ss:twoLayer}
BuRR as described so far has worst-case constant access time when the
the number of layers is constant (our analysis and experiments use four layers).  However,
for real-time applications, the resulting worst case might be a limitation.
Here, we describe how the worst-case number of accessed layers can be limited to
two.
The idea is quite simple: Rather than mapping all keys to the first layer, we
map the keys to all layers using a distribution still to be determined.  We now
guarantee that a key originally mapped to layer $i$ is either retrieved there or
from layer $i+1$. A query for key $x$ now proceeds as follows. First the primary
layer $i(x)$ for that key is determined. Then the bumping information of layer
$i$ is queried to find out whether $i$ is bumped.  If not, $x$ is retrieved from
layer $i$, otherwise it is retrieved from layer $i+1$ \emph{without} querying
the bumping information for layer $i+1$.

For constructing layer $i$, the input consists of keys $E'_i$ bumped from layer
$i-1$ ($E'_1=\emptyset$) and keys $E_i$ having layer $i$ as their primary layer
($E_i=\emptyset$ for the last layer).  First, the bumped keys
$E'_i$ are inserted bucket by bucket.  In this phase, a failure to insert a key
causes construction to fail.  Then the keys $E_i$ are processed bucket by bucket
as before, recording bumped keys in $E'_{i+1}$.
% \pesa{check: really? Perhaps \emph{all}
%   bumped keys are placed first.}\lhs{I checked the implementation, they
%   are. In fact, they are inserted immediately (i.e., when processing a bucket
%   and insertion fails, the items are immediately inserted into the secondary
%   layer. But conceptually that's no different than buffering them, and in fact,
%   PD's implementation supports that, too (but it's not used).)}

The size of layer $i$ can be chosen as $(1+\eps)(|E'_i|+|E_i|)$ to achieve the
same overloading effect as in basic BuRR.  An exception is the last layer $i^*$
where we choose the size large enough so that construction succeeds with high
probability.  Note that if $|E_i|$ shrinks geometrically, we can choose $i^*$
such that $|E_{i^*}|$ is negligible.

Except in the last layer, construction can only fail due to keys already bumped,
which is a small fraction even for practical $w$.  In the unlikely\footnotemark\
case of construction failing, construction of that layer can be retried with
more space. Tests suggest that increasing space by a factor of $\frac{w+1}{w}$
has similar or better construction success probability than a fresh hash
function.  \footnotetext{We do not have a complete analysis of this case yet but
  believe that our analysis in \cref{sec:bumped-ribbon} can be adapted to show
  that the construction process will succeed with high probability for
  $w=\Om{\log n}$.}

A simpler Bu$^1$RR construction has layers of predetermined sizes, all in one
linear system. Construction reliability and/or space efficiency are reduced
slightly because of reduced adaptability. For moderate $n \approx w^{3}$, layers of
uniform size can work well, especially if the last layer is of variable size.

Our implementation of Bu$^1$RR, tested to scale to billions of keys, uses layers
with sizes shrinking by a factor of two, each with a power of two number of
buckets.  To a first approximation, the primary layer $i(x)$ of a key is simply
the number of leading zeros in an appropriate hash value, up to the maximum
layer.  To consistently saturate construction with expected overload of
$\alpha = -\eps$, this is modified with a bias for the first layer. A portion
($\alpha$) of values with a leading 0 (not first layer) are changed to a leading
1 (re-assigned to first layer), so $|E_0| \approx (1 + \alpha)2^{-1}m$ and other
$|E_i| \approx (1 - \alpha)2^{-(i+1)}m$.  With bumped entries,
$|E'_i| \approx \alpha 2^{-i}m$, expected overload is consistent through the
layers: $|E_i| + |E'_i| \approx (1 + \alpha)2^{-(i+1)}m$.

%%% Local Variables:
%%% mode: latex
%%% TeX-master: "main"
%%% End:

% LocalWords:  BuRR AMQ BuRRs MHCs AMQs MHC congruential bijective Bu1RR

\section{Experiments}\label{s:exp}

\vspace{3mm} % undo negative space in myparagraph
\myparagraph{Implementation Details.}  We implemented BuRR in C++
using template parameters that allow us to navigate a large part of
the design space mapped in the full paper.  We use sequential
construction using 64-bit master-hash-codes (MHCs) so that the input
keys themselves are hashed only once.  Linear congruential mapping is
used to derive more pseudo-random bits from the MHC.  When not
otherwise mentioned, our default configuration is BuRR with
left-to-right processing of buckets, aggressive right-to-left
insertion within a bucket, threshold-based bumping, interleaved
storage of the solution $Z$, and separately stored metadata. The data
structure has four layers, the last of which uses
$w'\Def\min(w,64)$ and $\eps\geq0$, where $\eps$ is increased in increments of 0.05 until no keys are bumped.  For \onebit, we
choose $t \Def \lceil -2\eps b + \sqrt{b/(1+\eps)}/2 \rceil$ and
$\eps \Def -2/3 \cdot w/(4b+w)$.  For 2-bit, parameter tuning showed that
 $\ell\Def\ceil{(0.13-\eps/2)b}, u\Def\ceil{(0.3-\eps/2)b}$, and
$\eps \Def -3/w$~work well for $w=32$; for $w \geq 64$, we use
$\ell=\ceil{(0.09-3\eps/4)b}$, $u=\ceil{(0.22-1.3\eps)b}$, and $\eps \Def -4/w$.

In addition, there is a prototypical implementation of Bu$^1$RR from
\cite{Ribbon:Arxiv:2021}; see \cref{ss:twoLayer}. Both BuRR and Bu$^1$RR build
on the same software for ribbon solving from
\cite{Ribbon:Arxiv:2021}. %\pesa{any more details needed?}
For validation we extend the experimental setup used for Cuckoo and Xor
filters~\cite{Lemire:fastfilter:2020}, with our code available at
\url{https://github.com/lorenzhs/fastfilter_cpp} and \url{https://github.com/lorenzhs/BuRR}.\lhs{todo publish}

\myparagraph{Experimental Setup.}\label{exp:hw} %
All experiments were run on a machine with an AMD EPYC 7702 processor with 64
cores, a base clock speed of 2.0\,GHz, and a maximum boost clock speed of
3.35GHz.  The machine is equipped with 1\,TiB of DDR4-3200 main memory and runs
Ubuntu 20.04.  We use the \texttt{clang++} compiler in version 11.0 with
optimization flags \texttt{-O3 -march=native}.  During sequential experiments,
only a single core was active at any time to minimize interference.

We looked at different input sizes $n\in\set{10^6,10^7,10^8}$.  Like most
studies in this area, we first look at a {\bf sequential} workload on a powerful
processor with a considerable number of cores.  However, this seems unrealistic
since in most applications, one would not let most cores lay bare but use
them. Unless these cores have a workload with very high locality this would have
a considerable effect on the performance of the AMQs. We therefore also look at
a scenario that might be the opposite extreme to a sequential unloaded
setting. We run the benchmarks on all available hardware threads in {\bf
  parallel}.  Construction builds many AMQs of size $n$ in parallel. Queries
select AMQs randomly.
% \pesa{or was it round-robin?}\lhs{no, random is correct}
This emulates a large AMQ that is parallelized using sharding and
puts very high
load on the memory system.

\myparagraph{Space Overhead of BuRR}

\Cref{fig:epsilon64}
plots the fraction $e$ of empty slots of BuRR for
$w=64$ and several combinations of bucket size $b$ and different
threshold compression schemes.  Similar plots are given in the appendix in
Figure~\ref{fig:furtherEpsilon} for $w=32$, $w=128$, and for $w=64$
with sparse coefficients.  Note that (for an infinite number of
layers), the overhead is about
$o=e+\mu/(rb(1-e))$ where $r$ is the
number of retrieved bits and $\mu$ is the number of metadata bits per
bucket. Hence, at least when $\mu$ is constant, the overhead is a
monotonic function in $e$ and thus minimizing $e$ also minimizes overhead.

We see that for small $|\eps|$, $e$ decreases exponentially. For
sufficiently small $b$, $e$ can get almost arbitrarily small. For fixed
$b>w$, $e$ eventually reaches a local minimum because with
threshold-based compression, a large overload enforces large
thresholds ($>w$) and thus empty regions of buckets.  Which actual
configuration to choose depends primarily on $r$.  Roughly, for larger
$r$, more and more metadata bits (i.e., small $b$, higher resolution
of threshold values) can be invested to reduce~$e$.  For fixed~$b$
and threshold compression scheme, one can choose $\eps$ to
minimize $e$. One can often choose a larger $\eps$ to get slightly
better performance due to less bumping with little impact on $o$.
Perhaps the most delicate tuning parameters are the thresholds to use
for 2-bit and \onebit compression (see \cref{ss:thresholdCompression}).
Indeed, in \cref{fig:epsilon64} \onebit compression has lower $e$ than
2-bit compression for $b=64$ but higher $e$ for $b=128$. We expect
that 2-bit compression could always achieve smaller $e$ than \onebit
compression, but we have not found choices for the threshold values
that always ensure this. \Cref{tab:configs} in \cref{app:exp} summarizes key parameters
of some selected BuRR configurations.
% \pesa{todo for Lorenz:add some lines to the table that encompass sparse
% coefficients and Bu$^1$RR}\lhs{sparse done, don't have data for Bu$^1$RR}

In all following experiments, we use $b=2^{\floor{w^2/(2\log_2w)}}$ for
uncompressed and 2-bit compressed thresholds, and $b=2^{\floor{w^2/(4\log_2w)}}$
when using \onebit threshold compression.

\subsubsection*{Performance of BuRR Variants}
We have performed experiments with numerous configurations of
BuRR.
\iftr% TR: include giant pages with full tables at the end
See \cref{tab:performance} in \cref{app:exp} for a small sample
and \cref{app:fullexp} for the complete list.
\else% conference version: link to TR / online table
See \cref{tab:performance} in \cref{app:exp} for a small sample;
we will publish a complete list online.\lhs{todo}
\fi%
The scatter plot in
\cref{fig:scatter} summarizes the performance--overhead trade-off for
$r\approx 8$. Plots for different values of $r$ and for construction
and query times separately are in \cref{app:exp}
(\cref{fig:scatter2,fig:scatter16,fig:scatterQuery,fig:scatterQueryNeg,fig:scatterConstruction}).

A small {\bf ribbon width} of {\boldmath $w=16$} is feasible but does not pay
off with respect to performance because its high bumping rates drive up
the expected number of layers accessed.  Choosing {\boldmath $w=32$} yields the
best performance in many configurations but the penalty for going to
{\boldmath $w=64$} is very small while reducing overheads.  In contrast, {\boldmath
  $w=128$} has a large performance penalty for little additional gain --
overheads far below $1\,\%$ are already possible with $w=64$. Thus, on a
64-bit machine, $w=64$ seems the logical choice in most situations.

\begin{figure}[tb]
  \centering
  \includegraphics[page=2,width=\textwidth]{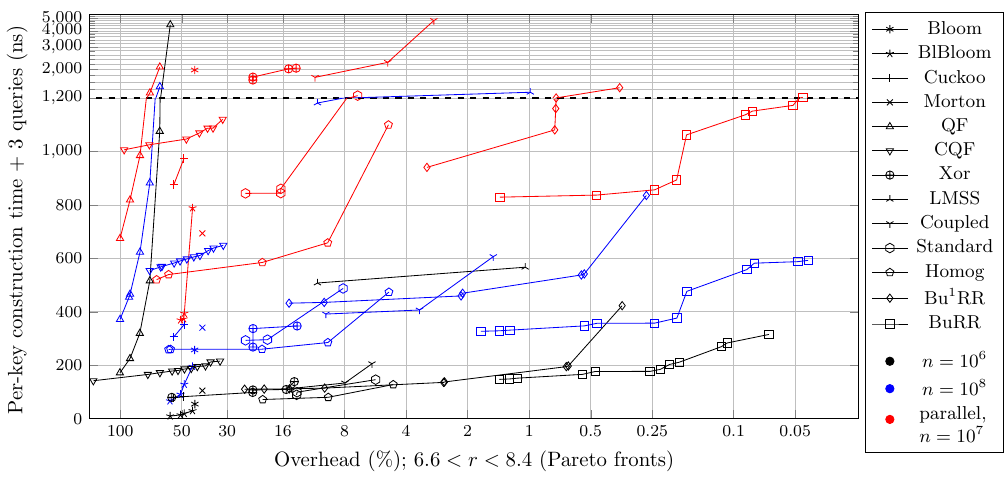}
  \caption{\label{fig:scatter}Performance--overhead trade-off for
    false-positive rate $2^{-r}$, approximately $0.3\,\%$ to $1\,\%$,
    for different AMQs and inputs. For each category of approaches, only
    points are shown that are not dominated by variants in the same
    category.  Sequential benchmarks use a single AMQ of size $n$
    while the parallel benchmark uses $1280$ AMQs and 64 cores.
    Logarithmic vertical axis above 1200\,ns.}
\end{figure}

With respect to performance, {\bf 1$^+$-bit} compression is slightly
slower than {\bf 2-bit} compression or uncompressed thresholds but not
by a large margin.  However, {\bf 1$^+$-bit} achieves the lowest
overheads.
{\bf Interleaved} table representation (see \cref{representation})
is mostly faster than {\bf contiguous} representation. This might change for
sufficiently large $r$ and use of advanced bit-manipulation or SIMD
instructions.
Nevertheless, {\bf sparse coefficients} with 8 out of 64 bits using contiguous
representation achieve significantly better query performance than the best
dense variant with comparable or better overhead when contiguous storage is
efficiently addressable, i.e., $r$ is a multiple of 8.
{\bf Bu$^1$RR} is around 20\,\% slower than BuRR and also somewhat
inferior with respect to the achieved overheads.  This may in part be due to
less extensive parameter tuning.  When worst-case guarantees
matter, Bu$^1$RR thus might be a good choice.\pesa{PD would you like to add anything?}

% \pesa{todo for Peter and for discussion with Lorenz: Something on parallel scaling. A short paragraph here and one or two plots for the appendix:
% Scalability of construction throughput when many threads build (different) tables in parallel.
% (y-axis is number of PEs.)

% Scalability of query throughput when many threads access shared tables in parallel.
% Use very large aggregate table size.
% (y-axis is number of PEs)}
% \lhs{added:}%
In a parallel scaling experiment with $10^{10}$ keys, construction and
queries both scaled well, shown in \cref{fig:scaling} in \cref{app:exp}.  Constructing many AMQs in parallel
achieves speedups of $65-71$ depending on the configuration when using all
64 cores plus hyperthreads of our machine (50 without hyperthreading).
Individual query times are around $15\,\%$ higher than sequentially when using all
cores, and $50\,\%$ higher when using all hyperthreads.  This
approximately matches the speedups for construction.

\subsubsection*{Comparison with Other Retrieval and AMQ Data Structures}
To compare BuRR with other approaches, we performed an extensive case study of
different {\bf AMQs} and retrieval data structures.  To compare space overheads,
we compare $r$-bit retrieval data structures to AMQs with false positive rate
$2^{-r}$.  Our benchmark extends the framework used to evaluate Cuckoo and Xor
filters~\cite{Lemire:fastfilter:2020}, with our modified version available at
\url{https://github.com/lorenzhs/fastfilter_cpp}. %\lhs{todo final}
In addition to adding our implementations of standard ribbon, homogeneous
ribbon, BuRR, and Bu${}^1$RR, we integrated existing implementations of Quotient
Filters~\cite{Maier:lpqfilter:2020,MSW:ConcQuotientFilter:2020} and LMSS,
Coupled, GOV, 2-block, and BPZ
retrieval~\cite{Walzer:RetrievalImplementation:2020,W:Thesis}.  We extended the
implementations of LMSS, Coupled, and BPZ to support the cases of $r=8$ and
$r=16$ in addition to one-bit retrieval.\footnote{This is easy for peeling-based
  approaches, but far more work would be required to do the same for GOV and
  2-block.}

We also added a parallel version of the benchmark where each thread constructs
a number of AMQs independently, but queries access all of them randomly.  In
the Cuckoo filter implementation, we replaced calls to \texttt{rand()} with a
\texttt{std::mt19937} Mersenne Twister to eliminate a parallel bottleneck.  The
implementation of LMSS cannot be run simultaneously in multiple threads of the
same process and was excluded from the parallel benchmark.

Both the sequential and the parallel benchmark use three query workloads:
positive queries, negative queries, and a mixed set containing equal numbers of
positive and negative queries.  We report many results in the form of
construction time per key plus the time for one query from each of the three
sets, measured by dividing the running time for construction plus $n$ queries of
each type by the input size~$n$.  This metric is a reasonable tradeoff between
construction and queries; we provide figures for the individual components in
\cref{app:exp} (\cref{fig:scatterQuery,fig:scatterQueryNeg,fig:scatterConstruction}).

Once more, the scatter plot in \cref{fig:scatter} summarizes the
performance--overhead trade-off for $r\approx 8$; other values of
$r$ are covered in \cref{app:exp} (\cref{fig:scatter2,fig:scatter16}). In addition,
\cref{fig:heatmap} gives an overview of the
fastest approach for different values of $r$ and overhead.
We now discuss different approaches progressing from high to low space overhead.

\newcommand{\smallerskip}{\medskip}
\smallerskip\par
\noindent
    {\bf\sffamily Bloom Filters Variants:} %\pesa{slightly expanded}
    Plain Bloom filters \cite{B:Space:1970} set $k\sim
    \log(1/\varphi)$ random bits in a table for each key in order to
    achieve false-positive rate $\varphi$. They are the most well-known
    and perhaps most widely used AMQ. However, they have an inherent
    space overhead of at least 44\,\% compared to the
    information-theoretic optimum.  Moreover, for large inputs they
    cause almost $k$ cache faults for each insertion and positive
    query.  {\bf Blocked Bloom filters}
    \cite{Putze:Efficient-Bloom-Filters:2009} are faster because they
    set all of the $k$ bits in the same cache block.  The downside is
    that this increases the false-positive rate, in particular for
    large~$k$. This can be slightly reduced using {\bf two blocks}.

\smallerskip\par
\noindent {\bf\sffamily Cuckoo Filters}
\cite{FAK:CuckooFilterBetter:2013} store a random fingerprint for each
key (similar to retrieval-based AMQs). However, to allow good space efficiency, several positions need to be
possible. This introduces an intrinsic space overhead
of a few bits per key
that is further
increased by some empty slots that are required to allow fast insertions. The
latter overhead is reduced in {\bf Morton filters}
\cite{BJ:MortonFilters:2020} which can be viewed as a compressed
representation of cuckoo filters.
In our measurements, cuckoo and Morton filters are the most space efficient
dynamic AMQ for small $\varphi$, but are otherwise outperformed
by other constructions.

\smallerskip\par\noindent \textbf{\textsf{Quotient Filters (QF)}}
\cite{BFJKKMMSS:QuotientFilters:2012} can be viewed as a compressed
representation of a Bloom filter with a single hash function.
QFs support not only insertions but also deletions and
counting.  Similar to cuckoo filters, they incur an overhead of a few
bits per key (2--3 depending on the implementation) plus a
multiplicative overhead due to empty entries. Search time grows fairly
quickly as the multiplicative overhead approaches zero.
\textbf{Counting Quotient Filters (CQF)} \cite{Pandey:CQF:2017}
mitigate this dependence at the cost of less locality for low fill
degrees. Overall, quotient filters are good if their rich set of
operations is needed but cannot compete with alternatives (e.g.,
blocked Bloom filters and cuckoo filters) for dynamic AMQs without
deletions. Compared to static AMQs, they have comparable or slower
speed than BuRR but two orders of magnitude higher overhead.

\smallerskip\par
\noindent {\bf\sffamily Xor filter/retrieval.} {\bf
  Xor Filters} and {\bf Xor+ Filters} \cite{GL:XorFilters:2020} are a
recent implementation of peeling-based retrieval that can reduce
overhead to 23\,\% and 14\,\%, respectively.  An earlier
implementation \cite{BPZ:Practical:2013} is 35--90\,\% slower to
construct and also has slightly slower queries.
\textbf{Fuse}
filters \cite{DW:DensePeelable:2019}, a variant of Xor, achieve higher loads and
are fast to construct if $n$ is not too large, but construction becomes slow for
large $n$ and parallel queries are less efficient than for plain Xor
filters. In the plots we only show the Pareto-optimal variants in each case and
call all of them Xor in the following and in \cref{s:intro}.
Among the tried retrieval data structures, Xor has the fastest queries for sequential settings
but is otherwise dominated by BuRR, which has one to two orders of magnitude smaller overheads.

\smallerskip\par
\noindent {\bf\sffamily Low overhead peeling-based retrieval.}  By
using several hash functions in a nonuniform way, peeling-based
retrieval data structures can in principle achieve arbitrarily small
overheads and thus could be viewed as the primary competitors of
ribbon-based data structures.
% We extended available implementations for $r=1$ from Stefan Walzer's PhD thesis~\cite{W:Thesis}
% to also work for $r=8$ and $r=16$.
{\bf LMSS} \cite{LMSS:Efficient_Erasure:2001} was
originally proposed as an error-correcting code but can also be used
for retrieval.  However,
in the experiments it is clearly dominated by BuRR.  The more recent
{\bf Coupled} peeling approach \cite{W:SpatialCoupling:2021} can achieve
Pareto-optimal query times for large sequential inputs but is
otherwise dominated by ribbon-based data structures. Coupled has
faster construction times than LMSS but in that respect is several times slower than
BuRR for large $n$ and in our parallel benchmark, even when it is allowed an order of magnitude more
overhead.  Nevertheless, when
disregarding ribbon-based data structures, Coupled comes closest to a
practical retrieval data structure with very low overhead.
For small inputs ($n=10^6$), $r\approx 16$ and overhead between 8 and 15\,\%, it
is even the fastest AMQ in our benchmark (see \cref{fig:heatmap}).
Perhaps for large $r$\pesa{and large $n$?}\lhs{no, for small $n$!} and by engineering faster construction algorithms,
Coupled could become more competitive in the future.

%% \lhs{BPZ has approx the same overhead as Xor but is
%%   35--90\,\% slower to construct (depending on input size) and also has slightly
%%   slower queries.}\lhs{added, Stefan is this fair to say?:} \textbf{Fuse}
%% filters \cite{DW:DensePeelable:2019}, a variant of Xor, achieve higher loads and
%% are fast to construct if $n$ is not too large, but construction becomes slow for
%% large $n$ and parallel queries are less efficient than for plain Xor
%% filters. Further compression is possible using {\bf Walzer's coupling approach}
%% \cite{W:SpatialCoupling:2021}. However, since these techniques have worse
%% locality than ribbon approaches, they are only competitive in settings that fit
%% in cache.

\smallerskip\par
\noindent {\bf\sffamily Standard Ribbon} can achieve overhead
around 10\,\%.  However, it often fails for large $n$, requiring larger
space overhead when used without sharding on top of it.
For AMQs this can be elegantly remedied using {\bf homogeneous ribbon filters}.
Thus, in the heatmap, homogeneous ribbon occupies the area between (blocked) Bloom filters and BuRR and its variants.
However, the performance advantage over BuRR in parallel and large settings is not very large (typically 20\,\%).

\smallerskip\par
\noindent {\bf\sffamily BuRR} and its variants take the entire right part of
the heatmap.  Compared to the best competitor -- homogeneous ribbon
filters -- overhead drops from around 10\,\% to well below 1\,\% at a
moderate performance penalty. In particular, due to BuRR's high
locality, performance is even better than for successful competitors
like Xor, Cuckoo, or Bloom filters.
%% Together, blocked
%% Bloom, homogeneous ribbon, and BuRR basically cover the entire heatmap
%% -- considerably pruning the spectrum of methods needed to build good
%% static AMQs.

\smallerskip\par
\noindent {\bf\sffamily 2-block}
\cite{DW:Retrieval-log-extra-bits:2019} can be viewed as a
generalization of ribbon-based approaches that use two rather than one
block of nonzeroes in each row of the constraint
matrix. Unfortunately this makes the equation system much more
difficult to solve. This implies very expensive construction even
when aggressively using the sharding trick. In our experiments, an implementation
by Walzer \cite{W:Thesis} for $r=1$ achieves smaller overhead than BuRR with $w=128$ at the price of an order of magnitude larger construction time. It is however likely that a BuRR implementation able to handle $w=256$ would dominate 2-block.
% 2-block is 15-18x slower than BuRR and achieves 0.234% overhead instead of 0.354% for BuRR

\smallerskip\par
\noindent {\bf\sffamily Techniques not tried.}  There are a few
interesting retrieval data structures for which we had no available
implementation.  {\bf FiRe}
\cite{Sanders:Retrieval-FingerPrinting:2014} is likely to be the
fastest existing retrieval data structure and also supports updates to
function values as well as a limited form of insertions. FiRe maps
keys to buckets fitting into a cache line.  Per-bucket metadata is
used to uniquely map some keys to data slots available in the
bucket while bumping the remaining keys. This requires a constant
number of metadata bits per key (around 4) and thus  implies an
overhead two orders of magnitude larger than BuRR.  Additionally, the only
known implementation of FiRe is closed source and owned by SAP, and was not
available to us.

We are not aware of implementations of the lookup-table based
approaches \cite{P:An_Optimal:2009,BelazzouguiV13} and do not view
them as practically promising for the reasons discussed in
\cref{s:related}.
%% We extended the
%% techniques %\pesa{Stefan please check.}
%% LMSS~\cite{LMSS:Efficient_Erasure:2001} and Coupled~\cite{W:SpatialCoupling:2021}
%% (peeling based), %\lhs{did NOT manage to extend GOV and 2-block yet}
%% %, GOV \cite{GOV:retrieval-Compressed:2020} (table lookups), and
%% %2-block \cite{DieWal19} (systems of linear equations with two random blocks per row)
%% which were studied for $r=1$ in Stefan Walzer's PhD thesis~\cite{W:Thesis}, to
%% also work with $r=8$ and $r=16$.  Both have unfavorable performance--overhead
%% trade-offs.  Still, 2-block is the only approach that can match or exceed the
%% overhead of BuRR.  Note that we could not evaluate LMSS in parallel settings, as
%% the code is not thread-safe.

%%% Local Variables:
%%% mode: latex
%%% TeX-master: "main"
%%% End:

% LocalWords:  BuRR MHCs congruential MHC AMQs AMQ SIMD hyperthreads LMSS BPZ
% LocalWords:  hyperthreading QF QFs CQF heatmap nonzeroes Walzer FiRe

%!TEX root=./main.tex

\section{Conclusion and Future Work}

\pesa{significantly rewritten conclusion}BuRR is a considerable
contribution to close a gap between theory and practice of retrieval
and static approximate membership data structures.  From the
theoretical side, BuRR is succinct while achieving constant access
cost for small number of retrieved bits ($r=\Oh{\log(n)/w}$). In
contrast to previous succinct constructions with better asymptotic
\stwa{performance$→$running times. Moreover: “our overhead is small” and “our overhead is tunable” are separate statements, so I slightly rephrased}running times,
its overhead is \emph{tunable} and already small for realistic values of $n$.
% it is \emph{tunable} to actually achieve low overhead for
% physically achievable values of $n$.
In practice, BuRR is faster than
widely used data structures with much larger overhead and reasonably
simple to implement. Our results further strengthen the success of
linear algebra based solutions to the problem. Our on-the-fly approach
shows that Gauss-like solvers can be superior to peeling based greedy
solvers even with respect to speed.

While the wide design space of BuRR leaves room for further practical
improvements, we see the main open problems for large $r$. In
practice, peeling based solvers (e.g., \cite{W:SpatialCoupling:2021})
might outperform BuRR if faster construction algorithms can be found
-- perhaps using ideas like overloading and bumping. In theory,
existing succinct data structures
(e.g. \cite{P:An_Optimal:2009,BelazzouguiV13}) allow constant query
time but have high space overhead for realistic input sizes.
Combining constant cost per element for large $r$ with small (preferably tunable) space overhead therefore remains a theoretical promise yet to be convincingly redeemed in practice.
\stwa{Old: \pesa{Achieving
succinctness \emph{and} tunability at constant cost per element thus
is a theoretical goal that might close the remaining gap between theory
and practice.} “succinctness \emph{and} tunability at constant cost per element” is theoretically brittle, i.e.\ easily achieved by superficially modifying e.g.\ Porat in a way that is not helpful. I rephrased the last sentence.}

\pesa{reformulated: We present a new static retrieval data structure which achieves excellent trade-offs between high space efficiency and fast running times. Like \cite{DW:One-Block-per-Row:2019}, the construction involves solving a system of linear equations over the two-element field, where each row contains a block of random bits in a random position. We augment this setup with a load balancing idea: The construction algorithm is given less space than required to store all keys but is permitted to \emph{bump} sets of keys in overloaded regions to a fallback data structure. There are two major benefits: Firstly, space utilization is significantly improved, and secondly, the width $w$ of blocks need no longer scale with $\O(\log n)$, which improves running times.
\ifconfVersion\else
We give a self-contained analysis in terms of an easy-to-grasp visual concept (the ribbon diagonal).
\fi

In an extensive experimental evaluation, we achieve overheads around 1\,\% at
running times with which previously, only overheads of more than 10\,\% were
achievable.  Overheads below 0.1\,\% are possible at modest additional
cost.}

% \section{Making Ribbon Practical}
% \input{configurability}

% \input{solution-layout}

% \input{practical-hashing}

% \input{standard-scalability}
% END section Making Ribbon Practical

% \input{balanced-ribbon}

% \input{experiments}

% \input{conc-future}

% \pesa{commented out Acknowledgements since they make anonymity questionable.}
\subsection*{Acknowledgements}
We thank Martin Dietzfelbinger for early contributions to this line of research. We thank others at Facebook for their supporting roles, including Jay Zhuang, Siying Dong, Shrikanth Shankar, Affan Dar, and Ramkumar Vadivelu.

%\clearpage

\bibliographystyle{plainurl}
\bibliography{bibliographie}

\begin{thebibliography}{10}

\bibitem{ADR:Experimental:2009}
Martin Aumüller, Martin Dietzfelbinger, and Michael Rink.
\newblock Experimental variations of a theoretically good retrieval data
  structure.
\newblock In {\em Proc. 17th {ESA}}, pages 742--751, 2009.
\newblock \href {https://doi.org/10.1007/978-3-642-04128-0_66}
  {\path{doi:10.1007/978-3-642-04128-0_66}}.

\bibitem{AWFS20}
Michael Axtmann, Sascha Witt, Daniel Ferizovic, and Peter Sanders.
\newblock Engineering in-place (shared-memory) sorting algorithms.
\newblock {\em CoRR}, 2020.
\newblock \href {http://arxiv.org/abs/2009.13569} {\path{arXiv:2009.13569}}.

\bibitem{BBOVV:Cache-Oblivious-Peeling:14}
Djamal Belazzougui, Paolo Boldi, Giuseppe Ottaviano, Rossano Venturini, and
  Sebastiano Vigna.
\newblock Cache-oblivious peeling of random hypergraphs.
\newblock In {\em Proc. DCC}, pages 352--361, 2014.
\newblock \href {https://doi.org/10.1109/DCC.2014.48}
  {\path{doi:10.1109/DCC.2014.48}}.

\bibitem{BelazzouguiV13}
Djamal Belazzougui and Rossano Venturini.
\newblock Compressed static functions with applications.
\newblock In Sanjeev Khanna, editor, {\em Proc. 24th {SODA}}, pages 229--240.
  {SIAM}, 2013.
\newblock \href {https://doi.org/10.1137/1.9781611973105.17}
  {\path{doi:10.1137/1.9781611973105.17}}.

\bibitem{BFJKKMMSS:QuotientFilters:2012}
Michael~A. Bender, Martin Farach{-}Colton, Rob Johnson, Russell Kraner,
  Bradley~C. Kuszmaul, Dzejla Medjedovic, Pablo Montes, Pradeep Shetty,
  Richard~P. Spillane, and Erez Zadok.
\newblock Don't thrash: How to cache your hash on flash.
\newblock {\em Proc. {VLDB} Endow.}, 5(11):1627--1637, 2012.
\newblock \href {https://doi.org/10.14778/2350229.2350275}
  {\path{doi:10.14778/2350229.2350275}}.

\bibitem{BGV:OnTheFlyGauss:2010}
Valerio Bioglio, Marco Grangetto, Rossano Gaeta, and Matteo Sereno.
\newblock On the fly gaussian elimination for {LT} codes.
\newblock {\em Communications Letters, IEEE}, 13(12):953 -- 955, 12 2009.
\newblock \href {https://doi.org/10.1109/LCOMM.2009.12.091824}
  {\path{doi:10.1109/LCOMM.2009.12.091824}}.

\bibitem{B:Space:1970}
Burton~H. Bloom.
\newblock Space/time trade-offs in hash coding with allowable errors.
\newblock {\em Commun. ACM}, 13(7):422–426, 1970.
\newblock \href {https://doi.org/10.1145/362686.362692}
  {\path{doi:10.1145/362686.362692}}.

\bibitem{B:Near-Optimal}
Fabiano~Cupertino Botelho.
\newblock {\em Near-Optimal Space Perfect Hashing Algorithms}.
\newblock PhD thesis, Federal University of Minas Gerais, 2008.
\newblock URL: \url{http://cmph.sourceforge.net/papers/thesis.pdf}.

\bibitem{BPZ:Simple:2007}
Fabiano~Cupertino Botelho, Rasmus Pagh, and Nivio Ziviani.
\newblock Simple and space-efficient minimal perfect hash functions.
\newblock In {\em Proc. 10th WADS}, pages 139--150, 2007.
\newblock \href {https://doi.org/10.1007/978-3-540-73951-7_13}
  {\path{doi:10.1007/978-3-540-73951-7_13}}.

\bibitem{BPZ:Practical:2013}
Fabiano~Cupertino Botelho, Rasmus Pagh, and Nivio Ziviani.
\newblock Practical perfect hashing in nearly optimal space.
\newblock {\em Inf. Syst.}, 38(1):108--131, 2013.
\newblock \href {https://doi.org/10.1016/j.is.2012.06.002}
  {\path{doi:10.1016/j.is.2012.06.002}}.

\bibitem{BJ:MortonFilters:2020}
Alex~D. Breslow and Nuwan Jayasena.
\newblock Morton filters: fast, compressed sparse cuckoo filters.
\newblock {\em {VLDB} J.}, 29(2-3):731--754, 2020.
\newblock \href {https://doi.org/10.1007/s00778-019-00561-0}
  {\path{doi:10.1007/s00778-019-00561-0}}.

\bibitem{BK:Multilevel:1990}
Andrei~Z. Broder and Anna~R. Karlin.
\newblock Multilevel adaptive hashing.
\newblock In David~S. Johnson, editor, {\em Proc. 1st {SODA}}, pages 43--53.
  {SIAM}, 1990.
\newblock URL: \url{http://dl.acm.org/citation.cfm?id=320176.320181}.

\bibitem{CKRT:The_Bloomier:2004}
Bernard Chazelle, Joe Kilian, Ronitt Rubinfeld, and Ayellet Tal.
\newblock The {Bloomier} filter: An efficient data structure for static support
  lookup tables.
\newblock In {\em Proc. 15th SODA}, pages 30--39. {SIAM}, 2004.
\newblock URL: \url{http://dl.acm.org/citation.cfm?id=982792.982797}.

\bibitem{Cooper:Rank-Of-Random-Matrices:2000}
Colin Cooper.
\newblock On the rank of random matrices.
\newblock {\em Random Structures \& Algorithms}, 16(2):209--232, 2000.
\newblock \href
  {https://doi.org/10.1002/(SICI)1098-2418(200003)16:2<209::AID-RSA6>3.0.CO;2-1}
  {\path{doi:10.1002/(SICI)1098-2418(200003)16:2<209::AID-RSA6>3.0.CO;2-1}}.

\bibitem{DAI:Monkey:2017}
Niv Dayan, Manos Athanassoulis, and Stratos Idreos.
\newblock Monkey: Optimal navigable key-value store.
\newblock In {\em Proceedings of the 2017 ACM International Conference on
  Management of Data}, SIGMOD '17, page 79–94, New York, NY, USA, 2017.
  Association for Computing Machinery.
\newblock \href {https://doi.org/10.1145/3035918.3064054}
  {\path{doi:10.1145/3035918.3064054}}.

\bibitem{DGMMPR:Tight:2010}
Martin Dietzfelbinger, Andreas Goerdt, Michael Mitzenmacher, Andrea Montanari,
  Rasmus Pagh, and Michael Rink.
\newblock Tight thresholds for cuckoo hashing via {XORSAT}.
\newblock In {\em Proc. 37th ICALP (1)}, pages 213--225, 2010.
\newblock \href {https://doi.org/10.1007/978-3-642-14165-2_19}
  {\path{doi:10.1007/978-3-642-14165-2_19}}.

\bibitem{DP:Succinct:2008}
Martin Dietzfelbinger and Rasmus Pagh.
\newblock Succinct data structures for retrieval and approximate membership
  (extended abstract).
\newblock In {\em Proc. 35th ICALP (1)}, pages 385--396, 2008.
\newblock \href {https://doi.org/10.1007/978-3-540-70575-8_32}
  {\path{doi:10.1007/978-3-540-70575-8_32}}.

\bibitem{DR:Applications:2009}
Martin Dietzfelbinger and Michael Rink.
\newblock Applications of a splitting trick.
\newblock In {\em Proc. 36th ICALP (1)}, pages 354--365, 2009.
\newblock \href {https://doi.org/10.1007/978-3-642-02927-1_30}
  {\path{doi:10.1007/978-3-642-02927-1_30}}.

\bibitem{DW:Retrieval-log-extra-bits:2019}
Martin Dietzfelbinger and Stefan Walzer.
\newblock Constant-time retrieval with {$O(\log m)$} extra bits.
\newblock In {\em Proc. 36th {STACS}}, pages 24:1--24:16, 2019.
\newblock \href {https://doi.org/10.4230/LIPIcs.STACS.2019.24}
  {\path{doi:10.4230/LIPIcs.STACS.2019.24}}.

\bibitem{DW:DensePeelable:2019}
Martin Dietzfelbinger and Stefan Walzer.
\newblock Dense peelable random uniform hypergraphs.
\newblock In {\em Proc. 27th {ESA}}, pages 38:1--38:16, 2019.
\newblock \href {https://doi.org/10.4230/LIPIcs.ESA.2019.38}
  {\path{doi:10.4230/LIPIcs.ESA.2019.38}}.

\bibitem{DW:One-Block-per-Row:2019}
Martin Dietzfelbinger and Stefan Walzer.
\newblock Efficient {Gauss} elimination for near-quadratic matrices with one
  short random block per row, with applications.
\newblock In {\em Proc. 27th {ESA}}, pages 39:1--39:18, 2019.
\newblock \href {https://doi.org/10.4230/LIPIcs.ESA.2019.39}
  {\path{doi:10.4230/LIPIcs.ESA.2019.39}}.

\bibitem{Dillinger:RocksDBBloom}
Peter~C. Dillinger.
\newblock {RocksDB} {FastLocalBloom}, 2019.
\newblock URL:
  \url{https://github.com/facebook/rocksdb/blob/master/util/bloom_impl.h}.

\bibitem{Ribbon:Arxiv:2021}
Peter~C. Dillinger and Stefan Walzer.
\newblock Ribbon filter: practically smaller than {Bloom} and {Xor}.
\newblock {\em CoRR}, 2021.
\newblock \href {http://arxiv.org/abs/2103.02515} {\path{arXiv:2103.02515}}.

\bibitem{DM:The_3-XORSAT:2002}
Olivier Dubois and Jacques Mandler.
\newblock The 3-{XORSAT} threshold.
\newblock In {\em Proc. 43rd FOCS}, pages 769--778, 2002.
\newblock \href {https://doi.org/10.1109/SFCS.2002.1182002}
  {\path{doi:10.1109/SFCS.2002.1182002}}.

\bibitem{FAK:CuckooFilterBetter:2013}
Bin Fan, David~G. Andersen, and Michael Kaminsky.
\newblock Cuckoo filter: Better than {B}loom.
\newblock {\em ;login:}, 38(4), 2013.
\newblock URL:
  \url{https://www.usenix.org/publications/login/august-2013-volume-38-number-4/cuckoo-filter-better-bloom}.

\bibitem{FPSS:Space_Efficient:2005}
Dimitris Fotakis, Rasmus Pagh, Peter Sanders, and Paul~G. Spirakis.
\newblock Space efficient hash tables with worst case constant access time.
\newblock {\em Theory Comput. Syst.}, 38(2):229--248, 2005.
\newblock \href {https://doi.org/10.1007/s00224-004-1195-x}
  {\path{doi:10.1007/s00224-004-1195-x}}.

\bibitem{FP:Sharp:2012}
Nikolaos Fountoulakis and Konstantinos Panagiotou.
\newblock Sharp load thresholds for cuckoo hashing.
\newblock {\em Random Struct. Algorithms}, 41(3):306--333, 2012.
\newblock \href {https://doi.org/10.1002/rsa.20426}
  {\path{doi:10.1002/rsa.20426}}.

\bibitem{Vigna:Fast-Scalable-Construction-of-Functions:2016}
Marco Genuzio, Giuseppe Ottaviano, and Sebastiano Vigna.
\newblock Fast scalable construction of (minimal perfect hash) functions.
\newblock In {\em Proc. 15th SEA}, pages 339--352, 2016.
\newblock \href {https://doi.org/10.1007/978-3-319-38851-9_23}
  {\path{doi:10.1007/978-3-319-38851-9_23}}.

\bibitem{GOV:retrieval-Compressed:2020}
Marco Genuzio, Giuseppe Ottaviano, and Sebastiano Vigna.
\newblock Fast scalable construction of ([compressed] static | minimal perfect
  hash) functions.
\newblock {\em Information and Computation}, 2020.
\newblock \href {https://doi.org/10.1016/j.ic.2020.104517}
  {\path{doi:10.1016/j.ic.2020.104517}}.

\bibitem{Lemire:fastfilter:2020}
Thomas~Mueller Graf and Daniel Lemire.
\newblock fastfilter\_cpp, 2019.
\newblock URL: \url{https://github.com/FastFilter/fastfilter_cpp}.

\bibitem{GL:XorFilters:2020}
Thomas~Mueller Graf and Daniel Lemire.
\newblock Xor filters: Faster and smaller than {Bloom} and cuckoo filters.
\newblock {\em {ACM} J. Exp. Algorithmics}, 25:1--16, 2020.
\newblock \href {https://doi.org/10.1145/3376122} {\path{doi:10.1145/3376122}}.

\bibitem{HKP:CompressedFunction:2009}
J{\'{o}}hannes~B. Hreinsson, Morten Kr{\o}yer, and Rasmus Pagh.
\newblock Storing a compressed function with constant time access.
\newblock In {\em Proc. 17th {ESA}}, pages 730--741, 2009.
\newblock \href {https://doi.org/10.1007/978-3-642-04128-0_65}
  {\path{doi:10.1007/978-3-642-04128-0_65}}.

\bibitem{Luczak:A-simple-solution}
Svante Janson and Malwina~J. Luczak.
\newblock A simple solution to the {$k$}-core problem.
\newblock {\em Random Struct. Algorithms}, 30(1-2):50--62, 2007.
\newblock \href {https://doi.org/10.1002/rsa.20147}
  {\path{doi:10.1002/rsa.20147}}.

\bibitem{kirsch2008less}
Adam Kirsch and Michael Mitzenmacher.
\newblock Less hashing, same performance: Building a better {B}loom filter.
\newblock {\em Random Struct. Algorithms}, 33(2):187--218, 2008.
\newblock \href {https://doi.org/10.1002/rsa.20208}
  {\path{doi:10.1002/rsa.20208}}.

\bibitem{Knuth:Vol2:81}
D.~E. Knuth.
\newblock {\em The Art of Computer Programming --- Seminumerical Algorithms},
  volume~2.
\newblock Addison Wesley, 2nd edition, 1981.

\bibitem{LNKB:Bloom:2019}
Harald Lang, Thomas Neumann, Alfons Kemper, and Peter~A. Boncz.
\newblock Performance-optimal filtering: Bloom overtakes cuckoo at
  high-throughput.
\newblock {\em Proc. {VLDB} Endow.}, 12(5):502--515, 2019.
\newblock \href {https://doi.org/10.14778/3303753.3303757}
  {\path{doi:10.14778/3303753.3303757}}.

\bibitem{L:A_New_Approach:2012}
Marc Lelarge.
\newblock A new approach to the orientation of random hypergraphs.
\newblock In {\em Proc. 23rd SODA}, pages 251--264. {SIAM}, 2012.
\newblock \href {https://doi.org/10.1137/1.9781611973099.23}
  {\path{doi:10.1137/1.9781611973099.23}}.

\bibitem{LMSS:Efficient_Erasure:2001}
Michael Luby, Michael Mitzenmacher, Mohammad~Amin Shokrollahi, and Daniel~A.
  Spielman.
\newblock Efficient erasure correcting codes.
\newblock {\em IEEE Trans. Inf. Theory}, 47(2):569--584, 2001.
\newblock \href {https://doi.org/10.1109/18.910575}
  {\path{doi:10.1109/18.910575}}.

\bibitem{Maier:lpqfilter:2020}
Tobais Maier.
\newblock lpqfilter, 2020.
\newblock URL: \url{https://github.com/TooBiased/lpqfilter}.

\bibitem{MSW:ConcQuotientFilter:2020}
Tobias Maier, Peter Sanders, and Robert Williger.
\newblock Concurrent expandable {AMQs} on the basis of quotient filters.
\newblock In {\em Proc. 18th {SEA}}, pages 15:1--15:13, 2020.
\newblock \href {https://doi.org/10.4230/LIPIcs.SEA.2020.15}
  {\path{doi:10.4230/LIPIcs.SEA.2020.15}}.

\bibitem{MU:Probability:2017}
Michael Mitzenmacher and Eli Upfal.
\newblock {\em Probability and Computing: Randomization and Probabilistic
  Techniques in Algorithms and Data Analysis}.
\newblock Cambridge University Press, New York, NY, USA, 2nd edition, 2017.

\bibitem{Molloy05:Cores-in-random-hypergraphs}
Michael Molloy.
\newblock Cores in random hypergraphs and {Boolean} formulas.
\newblock {\em Random Struct. Algorithms}, 27(1):124--135, 2005.
\newblock \href {https://doi.org/10.1002/rsa.20061}
  {\path{doi:10.1002/rsa.20061}}.

\bibitem{MRF14}
Ingo M{\"{u}}ller, Cornelius Ratsch, and Franz F{\"{a}}rber.
\newblock Adaptive string dictionary compression in in-memory column-store
  database systems.
\newblock In {\em Proc. 17th {EDBT}}, pages 283--294, 2014.
\newblock \href {https://doi.org/10.5441/002/edbt.2014.27}
  {\path{doi:10.5441/002/edbt.2014.27}}.

\bibitem{Sanders:Retrieval-FingerPrinting:2014}
Ingo M{\"u}ller, Peter Sanders, Robert Schulze, and Wei Zhou.
\newblock Retrieval and perfect hashing using fingerprinting.
\newblock In {\em Proc. 14th {SEA}}, pages 138--149, 2014.
\newblock \href {https://doi.org/10.1007/978-3-319-07959-2_12}
  {\path{doi:10.1007/978-3-319-07959-2_12}}.

\bibitem{Pandey:CQF:2017}
Prashant Pandey, Michael~A. Bender, Rob Johnson, and Rob Patro.
\newblock A general-purpose counting filter: Making every bit count.
\newblock In {\em Proc. 2017 {SIGMOD}}, pages 775--787. {ACM}, 2017.
\newblock \href {https://doi.org/10.1145/3035918.3035963}
  {\path{doi:10.1145/3035918.3035963}}.

\bibitem{P:An_Optimal:2009}
Ely Porat.
\newblock An optimal {Bloom} filter replacement based on matrix solving.
\newblock In {\em Proc. 4th CSR}, pages 263--273, 2009.
\newblock \href {https://doi.org/10.1007/978-3-642-03351-3_25}
  {\path{doi:10.1007/978-3-642-03351-3_25}}.

\bibitem{Putze:Efficient-Bloom-Filters:2009}
Felix Putze, Peter Sanders, and Johannes Singler.
\newblock Cache-, hash-, and space-efficient {Bloom} filters.
\newblock {\em {ACM} Journal of Experimental Algorithmics}, 14, 2009.
\newblock \href {https://doi.org/10.1145/1498698.1594230}
  {\path{doi:10.1145/1498698.1594230}}.

\bibitem{VKKM:Learned:2020}
Kapil Vaidya, Eric Knorr, Tim Kraska, and Michael Mitzenmacher.
\newblock Partitioned learned bloom filter.
\newblock {\em CoRR}, abs/2006.03176, 2020.
\newblock URL: \url{https://arxiv.org/abs/2006.03176}, \href
  {http://arxiv.org/abs/2006.03176} {\path{arXiv:2006.03176}}.

\bibitem{Walzer:RetrievalImplementation:2020}
Stefan Walzer.
\newblock Experimental comparison of retrieval data structures, 2020.
\newblock URL: \url{https://github.com/sekti/retrieval-test}.

\bibitem{W:Thesis}
Stefan Walzer.
\newblock {\em Random Hypergraphs for Hashing-Based Data Structures}.
\newblock PhD thesis, Technische Universität Ilmenau, 2020.
\newblock URL: \url{https://www.db-thueringen.de/receive/dbt_mods_00047127}.

\bibitem{W:SpatialCoupling:2021}
Stefan Walzer.
\newblock Peeling close to the orientability threshold: Spatial coupling in
  hashing-based data structures.
\newblock In {\em Proc. 32nd {SODA}}, pages 2194--2211. {SIAM}, 2021.
\newblock \href {https://doi.org/10.1137/1.9781611976465.131}
  {\path{doi:10.1137/1.9781611976465.131}}.

\end{thebibliography}

\appendix

\iffalse
\section{Ribbon filters: Practically Smaller than Bloom}
\label{app:benchmark}
\begin{figure}[h]
  \centering
  \includegraphics[page=1]{plot_scatter}
  \caption{Trade-offs between overhead and running times achieved by various filter data structures as measured in \cref{s:exp}. A mix of query and construction times is considered. All filters have a false positive rate $ϕ$ with $r = \log(1/ϕ) ∈ [6.6,8.4]$ and constructed sequentially for $n = 10⁷$ keys. All implemented competitors fail to achieve the space overheads of bumped ribbon retrieval or do so only at higher running times.}
  \label{fig:filterplot}
\end{figure}
\fi

%%%%%%%%%%%%%%%%%%%%%%%%%%%%%%%%%%%%%%%%%%%%%%%%%%%%%%%%%%%%%%%%%%%%%%
\section{More on the Design Space of BuRR}
\label{s:moredesign}

%---------------------------------------------------------------------
\subsection{Different Insertion Orders Within a Bucket}\label{ss:insertionOrder}

\vspace{3mm}% undo negative space of \myparagraph
\myparagraph{Right-to-left.} Assuming that we process buckets from left to right (see also \cref{ss:globalOrder}), a
simple and useful insertion order within a bucket is from right to
left, i.e., by decreasing values of $s(x)$.
%% \pesa{new sentence:}In particular, it harmonizes well with the threshold-based approach to bumping
%% discussed in \cref{ss:thresholdCompression}.
This ordering takes into account that the leftmost part of the bucket
is already mostly filled with keys spilled over from the previous
bucket whereas the next bucket is still empty. Thus, proceeding from
the right to the left, placement is initially easy.  With overloaded
buckets (negative $\eps$), placement gets more and more difficult
and gradually needs to place more and more keys in the next bucket.
The analysis in \cref{sec:bumped-ribbon}
suggests that linear dependencies mostly occur in two ways. Either, when
the overload of this bucket (or some subrange of it) is too large, or
when it runs into the keys placed into the left part of the bucket
when the previous bucket was constructed.  We make the former event
unlikely by choosing appropriate values of $b$ and $\eps$.  When
the latter event happens, we can bump a range of keys allocated to the
leftmost range of the bucket. This bumping scheme is discussed in
\cref{ss:thresholdCompression}.  The right-to-left ordering then helps
us to find the right threshold value.  An illustration is shown in
\cref{fig:fillrtlnew}.
%% \pedi{Are you claiming that you are performing fewer row reductions, or
%% that high variance in number of row reductions per insertion is faster?
%% Fewer row reductions only with cautious bumping?}\pesa{Fewer row reductions than what? I think cautious bumping is mostly orthogonal to the discussion above -- it only affects how late the construction process runs into a wall of already placed rows.}
%% \pedi{Why does ``running into a wall'' matter? ``Easy'' and ``difficult'' are
%% meaningless unless they have some relationship to efficiency or even
%% implementation difficulty. You haven't claimed to have improved anything
%% substantive by adding right-to-left rather than random or left-to-right.
%% Why should we care about insertion order when the theory paper says it doesn't
%% matter?}

% \begin{figure}
%   \centering
%   \includegraphics{fig_bucketfill}
%   \caption{Right-to-left insertion order. Placed keys' coefficients and occupied
%     rows in the result table are marked in gray.  Placement of the golden item is shown.  The positions that will eventually be bumped are marked in blue. \label{fig:fillrtl}}
% \end{figure}

\myparagraph{(Quasi)random order.}
The above simple insertion order is limited to situations where the
overload per bucket is less than $w$ most of the time.  Otherwise,
placement will often fail early, bumping many keys that could
actually be placed because their $s$ falls into a range of slots
that remain empty. We can achieve more flexibility in choosing
$\eps$ by spreading keys more uniformly over the bucket during
the insertion process.  We tried several such approaches.  Bu$^1$RR
uses additional hash bits to make the ordering for bumping independent
of the position within the bucket.
%% \pesa{adapted previous sentence to proposal by PD}.
%% \pedi{I would say that the grouping or ordering for
%% bumping is independent of position within the bucket. (Starting
%% position always comes from a hash function in some way.)}
Another interesting variant is most easy to explain
when $b$ is a power of two.  We use $c=\log b$ bits of hash function
value to define the position within a bucket. However, rather than
directly using the value as a column number, we use its \emph{mirror
  image}, i.e., a value $h_{c-1}h_{c-2}\ldots h_1h_0$ addresses column
$h_0h_1\ldots h_{c-2}h_{c-1}$ of a bucket.  We also tried a tabulated
random permutation, which according to early experiments, works slightly
worse than the mirror permutation.
%% \pesa{this implies that probably Bu$^1$RRs
%%   would also work better with sth like a mirror permutation?}

%---------------------------------------------------------------------
\subsection{Metadata for Bumping Multiple Groups}\label{ss:groupCompression}
%\pesa{move to appendix in submission?}

A disadvantage of threshold-based bumping is that a single failed
placement of a key implies that all subsequently placed keys
allocated to that bucket must be bumped. This penalty can be mitigated
by subdividing the positions in the insertion order into multiple groups
that can be bumped separately, e.g., by storing a single bit that
indicates whether a group is bumped.
Choosing \emph{groups of uniform (expected) size} is simple and fast.
It works well when
highly compressed metadata is not of primary importance, e.g., when
$r$ is large.

Better compression can be achieved by choosing \emph{groups
of variable size}. Bu$^1$RRs use groups whose sizes are a geometric
progression with a factor about $\sqrt{2}$ between subsequent group
sizes.  For example, to cover a bucket of size $b=1024$ using 8 bits of metadata, one could use
groups of expected sizes 28, 40, 57, 80, 113, 160, 226, and 320.
% \pesa{PD
%   please check whether I understood you correctly. Do you want to add anything on implementing this?}%
% \pedi{\emph{Expected} sizes, yes, this is how it works, except
% (implementation detail) my actual ratios of expected sizes are
% 5:3, 6:5, 5:3, 6:5, etc. IIRC I also tried 4:3, 3:2, 4:3, 3:2,
% etc. and I surely would have kept it if it were actually better.
% An interesting thing that you might re-emphasize here is that with
% hard sharding, this does not provide sufficient granularity for < 1%
% space overhead, at least under the naive analysis of the remainder
% being uniformly from 0 to 27 empty slots, but with soft sharding
% it empirically is sufficient.}\pesa{not sure how to describe the actual implementation. Perhaps we could give the concrete expected group sizes used in PDs implementation?}
% \pesa{Two new sentence to ``reemphasize'' not sure whether we should also explicitly refer to hard sharding here (perhaps in related work?):}
Note that smaller groups cannot and need not be bumped since
the main effect of bumping ``too much'' is that fewer keys spill over to the next bucket which can be rectified there. This is also the underlying reason why highly compressed representations of the threshold metadata from \cref{ss:thresholdCompression} are sufficient.

%---------------------------------------------------------------------
\subsection{Aggressive versus Cautious Bumping}\label{ss:cautious}

By default, the generic BuRR solving approach described in \cref{s:designspace} is
greedy, i.e., it tries to place as many keys as possible into the current bucket.
This will usually spill over close to $w$ keys into the next
bucket.  This increases the likelihood that construction in that
bucket fails early.  It might be a better approach to be more cautious
and try to avoid this situation.  For example, when after processing a
column $j$, more than $\alpha w$ keys are already placed in the
next bucket then all further keys are bumped from the current
bucket. More sophisticated balancing approaches are conceivable.

\pesa{NEW, pd please check.}For example, we use a form of cautious bumping for our implementation
of Bu$^1$RR (see \cref{ss:twoLayer}).  After placing keys bumped from
the previous layer, when processing bucket $i$, we first try to place
keys in its largest group. Then we try to place the keys in the second
largest group in bucket $i-1$, the third largest group in bucket $i-2$,
etc.  This can be viewed as driving a ``wedge'' through the buckets.

%---------------------------------------------------------------------
\subsection{Different Global Insertion Orders}\label{ss:globalOrder}
%\pesa{move to appendix in submission}
Many of our implemented variants of BuRR process buckets from left to
right and, within a bucket, place keys $x$ from right to left with
respect to $s(x)$.
We also tried the dual approach -- traversing the buckets from right to left and then
inserting from left to right within a bucket.  This behaves
identically with respect to space efficiency but leads to far more row operations and much
higher construction times. The straightforward ordering from left to
right both between and inside buckets does not work well with
aggressive bumping -- placement frequently fails early. We expect that
it may turn out to be a natural order for a cautious bumping
strategy.  Many other insertion orders can be considered. However, the
global left-to-right order has the advantage that spilling keys to
the next bucket is cheap since it is still empty. Thus, other orders
might have higher insertion time.

%---------------------------------------------------------------------
\subsection{Table Representation}\label{representation}

\begin{SCfigure}
  \centering
  \begin{subfigure}[t]{0.3\columnwidth}\centering
    \includegraphics{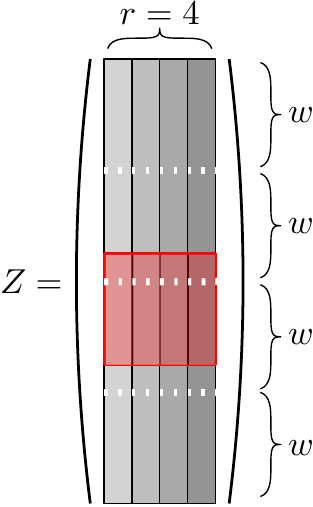}
  \end{subfigure}
  \begin{subfigure}[t]{0.3\columnwidth}\centering
    \includegraphics{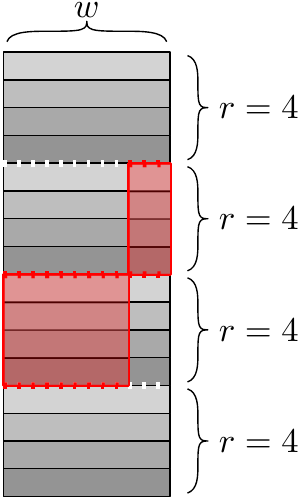}
  \end{subfigure}
  \caption{A solution matrix (left) and $w$-bit interleaved column-major layout
    (right) of that matrix.  The red shaded region shows the bits used in a
    query crossing a boundary.}
\end{SCfigure}

%\pesa{move to appendix in submission}
\myparagraph{Interleaved Versus Contiguous Storage.}
\emph{Contiguous} storage is the ``obvious'' representation of the
table by $m$ slots of $r$ contiguous bits each.  \emph{Interleaved}
means that the table is organized as $rm/w$ words of $w$ bits each
where word $i$ represents bit $i\bmod r$ of $w$ subsequent table
entries~\cite{Ribbon:Arxiv:2021}.  This organization allows the extraction
of one retrieved bit from two adjoining machine words using
population-count instructions.  Interleaved representation is
advantageous for uses of BuRR as an AMQ data structure since a negative
query only has to extract two bits in expectation. Moreover, the
implementation directly works for any value of $r$.

The contiguous representation, despite its conceptual simplicity, is
more difficult to implement, in particular when $r$ is not a power of
two. On the other hand, we expect it to be more efficient when all $r$
bits need to be retrieved, in particular when $r$ is large and when
the implementation makes careful use of available bit-manipulation\footnote{For example, the BMI2 bit manipulation operations
  PDEP and PEXT in newer x86 processors look useful.}  and SIMD
instructions.  This is particularly true when the sparse bit patterns
from \cref{ss:sparse} are used.

\myparagraph{Embedded Versus Separate Metadata.}
The ``obvious'' way to represent metadata is as a separate array with
one entry for each bucket (and a separate hash table for the \onebit
representation). On the other hand, if the metadata cannot be assumed
to be resident in cache, it is more cache-efficient to
store one bucket completely in one cache line, holding both its table
entries and metainformation. Then, assuming $b\geq w$, querying
the data structure accesses only one cache line plus possibly the next
cache line when the accessed part of the table extends there. In preliminary
experiments with variants of this approach, we observed
performance improvements of up to 10\,\% in some cases. We believe
that, depending on the implementation and the use case, the difference
could also be bigger but have not investigated this further since there
are too many disadvantages to this approach:
In particular, in the most
space-efficient configurations of BuRR, buckets can be bigger than a
cache line and the metadata will often fit into cache
anyway. Furthermore, when the memory system is not too contended,
metadata and primary table can be accessed in parallel, thus
eliminating the involved overhead.
Finally, embedded metadata is more complicated to implement.

%% \pesa{A subsection on retrieving non-binary data? I guess this can be
%%   generalized to numbers in the range $0..p^k-1$ for any prime $p$? Then,
%%   via prime factorization to any integer range? Practical encoding can
%%   be made space efficient buy using ``sweet spots'' where a prime
%%   power is just slightly below a power of two. For example, $3^5=243$
%%   this corresponds to a space loss of only 1\,\% if we encode blocks of
%%   5 numbers $\bmod 3$ as one byte. Similarly, $5^3=127$ so that 3
%%   numbers $\bmod 5$ can be encoded with 7 bits -- this is even more
%%   space efficient. Further decoding of these small packets can be done
%%   using table lookup. To achieve similar efficiency $\bmod 7$, we need
%%   larger packets -- $7^6=0.9908311488\cdot 2^{17}$,
%%   $7^{21}=558545864083284007=0.9992280230\cdot 2^{59}$ -- where
%%   explicit modulo computations may be needed. Either way this is
%%   reasonable efficient, since we need about half the decoded numbers
%%   in the packets.}

%%% Local Variables:
%%% mode: latex
%%% TeX-master: "main"
%%% End:

% LocalWords:  BuRR subrange AMQ BMI2 PDEP PEXT x86 SIMD

% IMPORT-DATA epstuning ../../ribbon/logs/epstuning_v5
% SQL CREATE TEMPORARY TABLE modename (mode INT, name VARCHAR);
% SQL INSERT INTO modename VALUES (0, "plain"), (1, "\onebits"), (2, "2-bit")
%%%%%%%%%%%%%%%%%%%%%%%%%%%%%%%%%%%%%%%%%%%%%%%%%%%%%%%%%%%%%%%%%%%%%%
\section{Further Experimental Data}\label{app:exp}

The following figures and tables contain
\begin{description}
  •[\cref{fig:furtherEpsilon}.] Shows the fraction of empty slots for many BuRR variants, supplementary to \cref{fig:epsilon64}.
  •[\cref{fig:scatter2,fig:scatter16}.] Performance-overhead trade-off of AMQs for very high false positive rate ($≈ 50\%$) and very low false positive rate ($≈ 0.01\%$) roughly corresponding to performance of $1$-bit retrieval and $16$-bit retrieval for  the retrieval-based AMQs.
  •[\cref{fig:scatterQuery,fig:scatterQueryNeg,fig:scatterConstruction}.] Performance-overhead trade-off of AMQs as in \cref{fig:scatter}, but seperately for positive queries, negative queries and construction.
  •[\cref{fig:scaling}.] Experiments concerning parallel scaling.
  •[\cref{tab:performance,tab:performance2}.] Numerical display of selected data.
  •[\cref{tab:configs}.] Overloading factors used for various BuRR configurations and corresponding overhead.
\end{description}

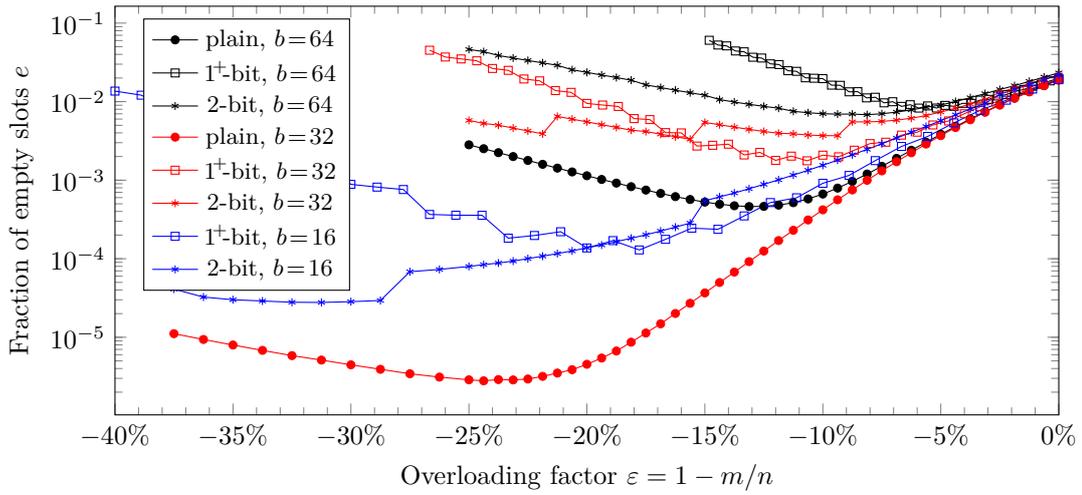
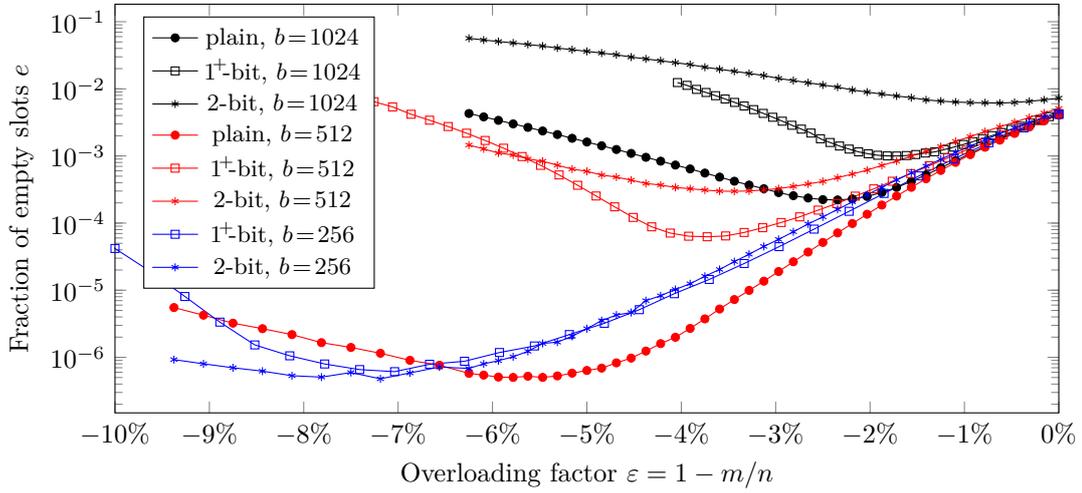
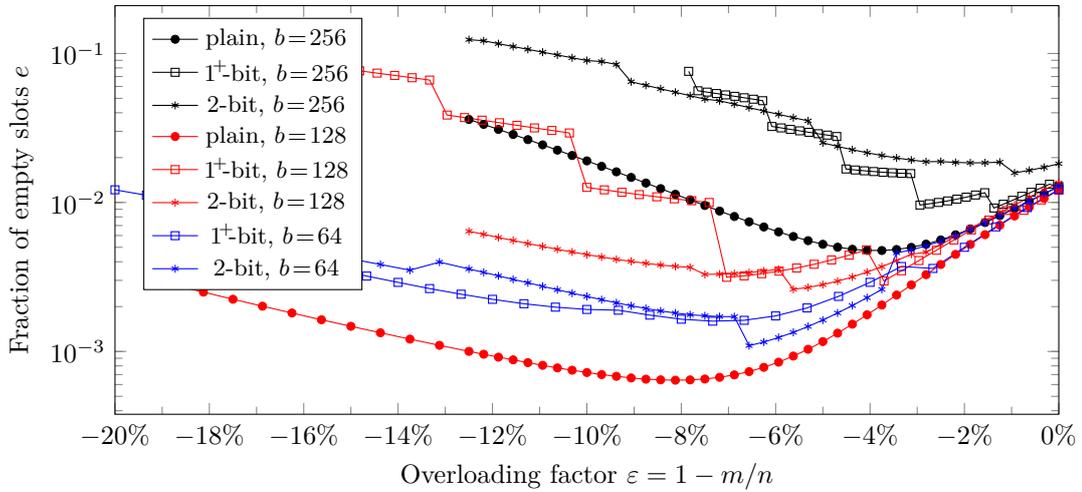
\begin{figure}
  \pgfplotsset{
    xlabel={Overloading factor $\eps = 1-m/n$},
    ylabel={Fraction of empty slots $e$},
    width=\textwidth,%16.5cm,
    height=7cm,
    xmax=0,
    legend style={font=\small},
    cycle list name=eps2,
    every axis/.style={ymode = log, mark size=1.5pt},,
    xticklabel={\pgfmathprintnumber\tick\%}
  }
  \begin{subfigure}[t]{\columnwidth}
  \centering
  \begin{tikzpicture}
    \begin{axis}[xmin=-40, legend pos=north west,xtick distance=5,minor tick num=4]
%% MULTIPLOT(mode, B|ptitle) SELECT 100*eps AS x, avg(tlemptyfrac) AS y,
%% name || ', $b\!=\!' || B || '$' as ptitle, MULTIPLOT
%% FROM epstuning NATURAL JOIN modename WHERE sparse=0 AND L=32 AND B<128
%% GROUP BY MULTIPLOT,x ORDER BY B DESC, mode, x
\addplot coordinates { (-25.0,0.00281985) (-24.375,0.00251523) (-23.75,0.00224156) (-23.125,0.00199934) (-22.5,0.00178396) (-21.875,0.00159252) (-21.25,0.00142432) (-20.625,0.00127346) (-20.0,0.00113899) (-19.375,0.00102135) (-18.75,0.000918189) (-18.125,0.000826938) (-17.5,0.000748364) (-16.875,0.000678421) (-16.25,0.000617974) (-15.625,0.000567619) (-15.0,0.000526109) (-14.375,0.000494696) (-13.75,0.000472418) (-13.125,0.000461536) (-12.5,0.000463788) (-11.875,0.000481065) (-11.25,0.00051721) (-10.625,0.000577045) (-10.0,0.000666301) (-9.375,0.000792733) (-8.75,0.000966067) (-8.125,0.00119825) (-7.5,0.00150193) (-6.875,0.00189584) (-6.25,0.00240321) (-5.625,0.00304146) (-5.0,0.00383581) (-4.375,0.00481971) (-3.75,0.00601273) (-3.125,0.0074447) (-2.5,0.00914354) (-1.875,0.011126) (-1.25,0.0134129) (-0.625,0.0160094) (0.0,0.0189151) };
\addlegendentry{plain, $b\!=\!64$};
\addplot coordinates { (-14.8148,0.0604196) (-14.4444,0.0525656) (-14.0741,0.0512619) (-13.7037,0.0441724) (-13.3333,0.0431096) (-12.963,0.0367608) (-12.5926,0.0359168) (-12.2222,0.0302971) (-11.8519,0.0296466) (-11.4815,0.0247607) (-11.1111,0.0243214) (-10.7407,0.0201198) (-10.3704,0.0198607) (-10.0,0.0196279) (-9.62963,0.0162194) (-9.25926,0.0161384) (-8.88889,0.0133389) (-8.51852,0.0133988) (-8.14815,0.0111703) (-7.77778,0.0113698) (-7.40741,0.009654) (-7.03704,0.00997595) (-6.66667,0.0087028) (-6.2963,0.00913976) (-5.92593,0.00827434) (-5.55556,0.00882198) (-5.18518,0.00828868) (-4.81481,0.00895836) (-4.44444,0.00871386) (-4.07407,0.0095093) (-3.7037,0.00951505) (-3.33333,0.0104543) (-2.96296,0.0106995) (-2.59259,0.0117985) (-2.22222,0.0122777) (-1.85185,0.0135451) (-1.48148,0.014277) (-1.11111,0.0157352) (-0.740741,0.0167265) (-0.37037,0.0183889) (0.0,0.019677) };
\addlegendentry{\onebits, $b\!=\!64$};
\addplot coordinates { (-25.0,0.0462501) (-24.375,0.0433203) (-23.75,0.0386847) (-23.125,0.0359251) (-22.5,0.033425) (-21.875,0.0311664) (-21.25,0.0290017) (-20.625,0.025471) (-20.0,0.0235715) (-19.375,0.0218268) (-18.75,0.0202709) (-18.125,0.0187572) (-17.5,0.0163564) (-16.875,0.0151153) (-16.25,0.013982) (-15.625,0.0129939) (-15.0,0.0120632) (-14.375,0.0106249) (-13.75,0.00987803) (-13.125,0.00925084) (-12.5,0.00870861) (-11.875,0.00824855) (-11.25,0.00753084) (-10.625,0.00722003) (-10.0,0.00702144) (-9.375,0.00693155) (-8.75,0.00694659) (-8.125,0.00681598) (-7.5,0.00701364) (-6.875,0.00735555) (-6.25,0.0078484) (-5.625,0.00848145) (-5.0,0.00905993) (-4.375,0.0100333) (-3.75,0.0112283) (-3.125,0.0126522) (-2.5,0.0142769) (-1.875,0.0159698) (-1.25,0.0181269) (-0.625,0.0205525) (0.0,0.0232584) };
\addlegendentry{2-bit, $b\!=\!64$};
\addplot coordinates { (-37.5,1.11081e-05) (-36.25,9.37563e-06) (-35.0,7.96875e-06) (-33.75,6.81062e-06) (-32.5,5.84e-06) (-31.25,5.125e-06) (-30.0,4.45688e-06) (-28.75,3.91187e-06) (-27.5,3.43125e-06) (-26.25,3.11188e-06) (-25.0,2.87187e-06) (-24.375,2.78875e-06) (-23.75,2.89312e-06) (-23.125,2.86312e-06) (-22.5,2.94188e-06) (-21.875,3.17813e-06) (-21.25,3.50063e-06) (-20.625,3.8625e-06) (-20.0,4.52938e-06) (-19.375,5.43937e-06) (-18.75,6.68375e-06) (-18.125,8.65375e-06) (-17.5,1.13256e-05) (-16.875,1.48469e-05) (-16.25,2.00931e-05) (-15.625,2.70644e-05) (-15.0,3.66987e-05) (-14.375,4.98331e-05) (-13.75,6.76469e-05) (-13.125,9.18306e-05) (-12.5,0.000125168) (-11.875,0.000170349) (-11.25,0.000230206) (-10.625,0.000311814) (-10.0,0.000419496) (-9.375,0.000562524) (-8.75,0.000751539) (-8.125,0.000997064) (-7.5,0.00131382) (-6.875,0.00171862) (-6.25,0.00223454) (-5.625,0.00288355) (-5.0,0.00368815) (-4.375,0.00467869) (-3.75,0.00587932) (-3.125,0.00731922) (-2.5,0.00902541) (-1.875,0.0110184) (-1.25,0.0133117) (-0.625,0.0159144) (0.0,0.0188268) };
\addlegendentry{plain, $b\!=\!32$};
\addplot coordinates { (-26.6667,0.0451441) (-26.0,0.0370059) (-25.3333,0.0349849) (-24.6667,0.033125) (-24.0,0.0264459) (-23.3333,0.0250276) (-22.6667,0.0194392) (-22.0,0.018398) (-21.3333,0.0138227) (-20.6667,0.0131238) (-20.0,0.00949072) (-19.3333,0.00905382) (-18.6667,0.0086791) (-18.0,0.00608605) (-17.3333,0.00590583) (-16.6667,0.00403633) (-16.0,0.00400066) (-15.3333,0.00271859) (-14.6667,0.00278075) (-14.0,0.00287182) (-13.3333,0.00209442) (-12.6667,0.00225914) (-12.0,0.00178687) (-11.3333,0.00202378) (-10.6667,0.00176561) (-10.0,0.00208084) (-9.33333,0.00198425) (-8.66667,0.00240735) (-8.0,0.00292369) (-7.33333,0.00304325) (-6.66667,0.00374371) (-6.0,0.0040814) (-5.33333,0.00504708) (-4.66667,0.00567698) (-4.0,0.00701796) (-3.33333,0.00860974) (-2.66667,0.00989144) (-2.0,0.0120256) (-1.33333,0.0139236) (-0.666667,0.0166905) (0.0,0.0193325) };
\addlegendentry{\onebits, $b\!=\!32$};
\addplot coordinates { (-25.0,0.00577003) (-24.375,0.00528416) (-23.75,0.00501518) (-23.125,0.00460493) (-22.5,0.00423449) (-21.875,0.00388932) (-21.25,0.00648413) (-20.625,0.0059883) (-20.0,0.00552539) (-19.375,0.00509236) (-18.75,0.00470129) (-18.125,0.00433298) (-17.5,0.00415975) (-16.875,0.00385628) (-16.25,0.00357131) (-15.625,0.00331777) (-15.0,0.00543868) (-14.375,0.00507431) (-13.75,0.00473551) (-13.125,0.00443709) (-12.5,0.00417045) (-11.875,0.00395806) (-11.25,0.00389311) (-10.625,0.00375986) (-10.0,0.00369022) (-9.375,0.00368281) (-8.75,0.00553159) (-8.125,0.00552104) (-7.5,0.00559661) (-6.875,0.0058124) (-6.25,0.00615072) (-5.625,0.00665013) (-5.0,0.00737401) (-4.375,0.00823606) (-3.75,0.00931172) (-3.125,0.0106273) (-2.5,0.0131815) (-1.875,0.0148951) (-1.25,0.0169044) (-0.625,0.0192242) (0.0,0.0218684) };
\addlegendentry{2-bit, $b\!=\!32$};
\addplot coordinates { (-44.4444,0.0225065) (-43.3333,0.0198044) (-42.2222,0.0174487) (-41.1111,0.0154018) (-40.0,0.0136246) (-38.8889,0.0120889) (-37.7778,0.00821061) (-36.6667,0.00726234) (-35.5556,0.00445611) (-34.4444,0.00395665) (-33.3333,0.00352877) (-32.2222,0.00194977) (-31.1111,0.00175859) (-30.0,0.000885302) (-28.8889,0.000816258) (-27.7778,0.000760293) (-26.6667,0.000366244) (-25.5556,0.00035818) (-24.4444,0.000357531) (-23.3333,0.000182732) (-22.2222,0.000197558) (-21.1111,0.000220695) (-20.0,0.000137114) (-18.8889,0.000170116) (-17.7778,0.000129374) (-16.6667,0.000176599) (-15.5556,0.000245044) (-14.4444,0.000236926) (-13.3333,0.000350877) (-12.2222,0.000518728) (-11.1111,0.000596664) (-10.0,0.000912674) (-8.88889,0.00114955) (-7.77778,0.00177438) (-6.66667,0.00268749) (-5.55556,0.00360516) (-4.44444,0.00535865) (-3.33333,0.00776109) (-2.22222,0.0104364) (-1.11111,0.0144698) (0.0,0.0190303) };
\addlegendentry{\onebits, $b\!=\!16$};
\addplot coordinates { (-37.5,4.1035e-05) (-36.25,3.24931e-05) (-35.0,3.00581e-05) (-33.75,2.89594e-05) (-32.5,2.79144e-05) (-31.25,2.77469e-05) (-30.0,2.83969e-05) (-28.75,2.93913e-05) (-27.5,6.83519e-05) (-26.25,7.31181e-05) (-25.0,7.97856e-05) (-24.375,8.40081e-05) (-23.75,8.82e-05) (-23.125,9.33725e-05) (-22.5,0.000100269) (-21.875,0.000107761) (-21.25,0.000116295) (-20.625,0.000125629) (-20.0,0.000137941) (-19.375,0.000149408) (-18.75,0.000165076) (-18.125,0.000182319) (-17.5,0.000201579) (-16.875,0.00022491) (-16.25,0.000252497) (-15.625,0.000284001) (-15.0,0.000546439) (-14.375,0.000612853) (-13.75,0.000686725) (-13.125,0.000778369) (-12.5,0.000885126) (-11.875,0.00101131) (-11.25,0.00116089) (-10.625,0.00133802) (-10.0,0.00154758) (-9.375,0.00179992) (-8.75,0.00210319) (-8.125,0.00247025) (-7.5,0.00290858) (-6.875,0.00343521) (-6.25,0.0040625) (-5.625,0.00482532) (-5.0,0.00572406) (-4.375,0.00679704) (-3.75,0.00805673) (-3.125,0.00953381) (-2.5,0.0123804) (-1.875,0.0143338) (-1.25,0.0165366) (-0.625,0.0190109) (0.0,0.0217528) };
\addlegendentry{2-bit, $b\!=\!16$};

    \end{axis}
  \end{tikzpicture}
  \vspace*{-0.2cm}
  \caption{Ribbon width $w=32$, regular (dense) coefficient vectors.}
  \end{subfigure}
  \begin{subfigure}[b]{\columnwidth}
  \centering
  \vspace*{0.3cm}
  \begin{tikzpicture}
    \begin{axis}[xmin=-10,legend pos=north west]
%% MULTIPLOT(mode, B|ptitle) SELECT 100*eps AS x, avg(tlemptyfrac) AS y,
%% name || ', $b\!=\!' || B || '$' as ptitle, MULTIPLOT
%% FROM epstuning NATURAL JOIN modename WHERE sparse=0 AND L=128 AND B<2048
%% GROUP BY MULTIPLOT,x ORDER BY B desc, mode, x
\addplot coordinates { (-6.25,0.00429232) (-6.09375,0.00381947) (-5.9375,0.00339302) (-5.78125,0.00300424) (-5.625,0.00266436) (-5.46875,0.00235733) (-5.3125,0.00208554) (-5.15625,0.00183385) (-5.0,0.00161411) (-4.84375,0.00141806) (-4.6875,0.00124807) (-4.53125,0.0010922) (-4.375,0.000955052) (-4.21875,0.000836561) (-4.0625,0.000732209) (-3.90625,0.000640694) (-3.75,0.000555828) (-3.59375,0.000484521) (-3.4375,0.000420458) (-3.28125,0.000368566) (-3.125,0.000324881) (-2.96875,0.000284609) (-2.8125,0.00025592) (-2.65625,0.00023622) (-2.5,0.000223648) (-2.34375,0.000221175) (-2.1875,0.000229477) (-2.03125,0.000249369) (-1.875,0.000285812) (-1.71875,0.000340032) (-1.5625,0.000416757) (-1.40625,0.000522986) (-1.25,0.000666645) (-1.09375,0.000850683) (-0.9375,0.00108763) (-0.78125,0.00138618) (-0.625,0.00175433) (-0.46875,0.00220116) (-0.3125,0.00273636) (-0.15625,0.00336643) (0.0,0.00409127) };
\addlegendentry{plain, $b\!=\!1024$};
\addplot coordinates { (-4.0404,0.0124013) (-3.93939,0.0112651) (-3.83838,0.00981269) (-3.73737,0.00882491) (-3.63636,0.00790681) (-3.53535,0.00705888) (-3.43434,0.00627766) (-3.33333,0.00555485) (-3.23232,0.00490196) (-3.13131,0.00430834) (-3.0303,0.00377295) (-2.92929,0.00329542) (-2.82828,0.00286634) (-2.72727,0.00251073) (-2.62626,0.00219318) (-2.52525,0.0018039) (-2.42424,0.00158761) (-2.32323,0.00140649) (-2.22222,0.00126579) (-2.12121,0.00115038) (-2.0202,0.00107759) (-1.91919,0.00102424) (-1.81818,0.00100167) (-1.71717,0.000998783) (-1.61616,0.00101413) (-1.51515,0.00104999) (-1.41414,0.00110547) (-1.31313,0.00115102) (-1.21212,0.00124487) (-1.11111,0.00135528) (-1.0101,0.00148342) (-0.909091,0.00163466) (-0.808081,0.00180911) (-0.707071,0.00200963) (-0.606061,0.00223447) (-0.50505,0.00249055) (-0.40404,0.00277472) (-0.30303,0.00309821) (-0.20202,0.0034575) (-0.10101,0.00385502) (0.0,0.00427077) };
\addlegendentry{\onebits, $b\!=\!1024$};
\addplot coordinates { (-6.25,0.0566915) (-6.09375,0.0533568) (-5.9375,0.0507542) (-5.78125,0.048237) (-5.625,0.0458036) (-5.46875,0.0434054) (-5.3125,0.0405421) (-5.15625,0.0383077) (-5.0,0.0361828) (-4.84375,0.0341446) (-4.6875,0.0321788) (-4.53125,0.0297507) (-4.375,0.0279327) (-4.21875,0.0262088) (-4.0625,0.0245426) (-3.90625,0.0229695) (-3.75,0.0210032) (-3.59375,0.0195945) (-3.4375,0.0182466) (-3.28125,0.0169671) (-3.125,0.0157747) (-2.96875,0.0142889) (-2.8125,0.0132834) (-2.65625,0.0123205) (-2.5,0.0114216) (-2.34375,0.0106137) (-2.1875,0.00961201) (-2.03125,0.00895035) (-1.875,0.00835244) (-1.71875,0.00782652) (-1.5625,0.007388) (-1.40625,0.00684765) (-1.25,0.00656039) (-1.09375,0.00636441) (-0.9375,0.00623395) (-0.78125,0.0061988) (-0.625,0.00612915) (-0.46875,0.00627439) (-0.3125,0.00651205) (-0.15625,0.0068456) (0.0,0.00726905) };
\addlegendentry{2-bit, $b\!=\!1024$};
\addplot coordinates { (-9.375,5.51687e-06) (-9.0625,4.24375e-06) (-8.75,3.2275e-06) (-8.4375,2.67375e-06) (-8.125,2.18188e-06) (-7.8125,1.65562e-06) (-7.5,1.40438e-06) (-7.1875,1.15e-06) (-6.875,9.025e-07) (-6.5625,7.63125e-07) (-6.25,5.775e-07) (-6.09375,5.4375e-07) (-5.9375,5.08125e-07) (-5.78125,5.01875e-07) (-5.625,5.23125e-07) (-5.46875,5.025e-07) (-5.3125,5.2875e-07) (-5.15625,5.80625e-07) (-5.0,6.375e-07) (-4.84375,6.9125e-07) (-4.6875,8.275e-07) (-4.53125,9.75e-07) (-4.375,1.23938e-06) (-4.21875,1.60063e-06) (-4.0625,1.99063e-06) (-3.90625,2.69687e-06) (-3.75,3.73875e-06) (-3.59375,5.24438e-06) (-3.4375,7.26062e-06) (-3.28125,9.93625e-06) (-3.125,1.36256e-05) (-2.96875,1.8985e-05) (-2.8125,2.65244e-05) (-2.65625,3.70406e-05) (-2.5,5.15406e-05) (-2.34375,7.151e-05) (-2.1875,9.84619e-05) (-2.03125,0.000135438) (-1.875,0.000185694) (-1.71875,0.000252326) (-1.5625,0.000340968) (-1.40625,0.000457403) (-1.25,0.000608494) (-1.09375,0.000802327) (-0.9375,0.00104632) (-0.78125,0.00134913) (-0.625,0.00172183) (-0.46875,0.00217315) (-0.3125,0.00271256) (-0.15625,0.00334535) (0.0,0.00407416) };
\addlegendentry{plain, $b\!=\!512$};
\addplot coordinates { (-7.84314,0.0103954) (-7.64706,0.00893144) (-7.45098,0.00761546) (-7.2549,0.00643461) (-7.05882,0.00537885) (-6.86274,0.00417985) (-6.66667,0.00340992) (-6.47059,0.00274492) (-6.27451,0.00218223) (-6.07843,0.00170193) (-5.88235,0.00130428) (-5.68627,0.000979891) (-5.4902,0.000724812) (-5.29412,0.000521229) (-5.09804,0.000365991) (-4.90196,0.000253331) (-4.70588,0.000174774) (-4.5098,0.000121413) (-4.31373,8.797e-05) (-4.11765,7.10213e-05) (-3.92157,6.37675e-05) (-3.72549,6.25563e-05) (-3.52941,6.43806e-05) (-3.33333,7.25181e-05) (-3.13725,8.52813e-05) (-2.94118,0.000102008) (-2.7451,0.00012386) (-2.54902,0.000155999) (-2.35294,0.000195652) (-2.15686,0.000252056) (-1.96078,0.000323567) (-1.76471,0.000420184) (-1.56863,0.000546867) (-1.37255,0.000714973) (-1.17647,0.000936604) (-0.980392,0.00122451) (-0.784314,0.00158858) (-0.588235,0.00205393) (-0.392157,0.00263619) (-0.196078,0.00335) (0.0,0.00420935) };
\addlegendentry{\onebits, $b\!=\!512$};
\addplot coordinates { (-6.25,0.00145383) (-6.09375,0.00127478) (-5.9375,0.00110588) (-5.78125,0.00103278) (-5.625,0.000902668) (-5.46875,0.00083588) (-5.3125,0.000728214) (-5.15625,0.000637672) (-5.0,0.000590294) (-4.84375,0.000520367) (-4.6875,0.000484713) (-4.53125,0.000436282) (-4.375,0.000389316) (-4.21875,0.000368744) (-4.0625,0.000339763) (-3.90625,0.000325217) (-3.75,0.000310409) (-3.59375,0.000299263) (-3.4375,0.000297697) (-3.28125,0.000300758) (-3.125,0.000311971) (-2.96875,0.000326493) (-2.8125,0.000351401) (-2.65625,0.000381766) (-2.5,0.000418046) (-2.34375,0.000474294) (-2.1875,0.000538369) (-2.03125,0.000615594) (-1.875,0.000722112) (-1.71875,0.000836681) (-1.5625,0.000998472) (-1.40625,0.00117396) (-1.25,0.00137954) (-1.09375,0.00164847) (-0.9375,0.0019451) (-0.78125,0.00230934) (-0.625,0.00271123) (-0.46875,0.0031817) (-0.3125,0.00374817) (-0.15625,0.00436195) (0.0,0.00507799) };
\addlegendentry{2-bit, $b\!=\!512$};
\addplot coordinates { (-14.8148,0.00682469) (-14.4444,0.00543505) (-14.0741,0.00425002) (-13.7037,0.00325591) (-13.3333,0.00244532) (-12.963,0.00178359) (-12.5926,0.00141394) (-12.2222,0.000994121) (-11.8519,0.000667589) (-11.4815,0.000431073) (-11.1111,0.000266231) (-10.7407,0.000156031) (-10.3704,8.33438e-05) (-10.0,4.18525e-05) (-9.62963,1.91563e-05) (-9.25926,8.05813e-06) (-8.88889,3.34562e-06) (-8.51852,1.53063e-06) (-8.14815,1.05375e-06) (-7.77778,7.9625e-07) (-7.40741,6.56875e-07) (-7.03704,6.11875e-07) (-6.66667,7.8875e-07) (-6.2963,8.74375e-07) (-5.92593,1.18e-06) (-5.55556,1.47625e-06) (-5.18518,2.18438e-06) (-4.81481,3.21438e-06) (-4.44444,5.12875e-06) (-4.07407,8.87188e-06) (-3.7037,1.44525e-05) (-3.33333,2.49962e-05) (-2.96296,4.46062e-05) (-2.59259,8.1315e-05) (-2.22222,0.000149766) (-1.85185,0.000276474) (-1.48148,0.000504379) (-1.11111,0.000904461) (-0.740741,0.00156981) (-0.37037,0.00261817) (0.0,0.00415691) };
\addlegendentry{\onebits, $b\!=\!256$};
\addplot coordinates { (-9.375,9.25e-07) (-9.0625,8.00625e-07) (-8.75,6.99375e-07) (-8.4375,6.225e-07) (-8.125,5.325e-07) (-7.8125,5.075e-07) (-7.5,5.9375e-07) (-7.1875,4.7875e-07) (-6.875,5.85625e-07) (-6.5625,7.175e-07) (-6.25,6.83125e-07) (-6.09375,7.99375e-07) (-5.9375,8.9125e-07) (-5.78125,1.02e-06) (-5.625,1.23e-06) (-5.46875,1.61125e-06) (-5.3125,1.68313e-06) (-5.15625,2.06e-06) (-5.0,2.675e-06) (-4.84375,3.52375e-06) (-4.6875,4.26938e-06) (-4.53125,4.65625e-06) (-4.375,7.00063e-06) (-4.21875,8.26625e-06) (-4.0625,1.02619e-05) (-3.90625,1.24838e-05) (-3.75,1.6145e-05) (-3.59375,2.037e-05) (-3.4375,2.66781e-05) (-3.28125,3.43556e-05) (-3.125,4.48619e-05) (-2.96875,5.71231e-05) (-2.8125,7.41906e-05) (-2.65625,9.71144e-05) (-2.5,0.000123464) (-2.34375,0.000160066) (-2.1875,0.000207111) (-2.03125,0.000264326) (-1.875,0.000340231) (-1.71875,0.000438931) (-1.5625,0.000559761) (-1.40625,0.000705899) (-1.25,0.000894432) (-1.09375,0.00112737) (-0.9375,0.0014022) (-0.78125,0.0017376) (-0.625,0.002146) (-0.46875,0.0026084) (-0.3125,0.00316448) (-0.15625,0.00380739) (0.0,0.00453439) };
\addlegendentry{2-bit, $b\!=\!256$};
    \end{axis}
  \end{tikzpicture}
  \vspace*{-0.2cm}
  \caption{Ribbon width $w=128$, regular (dense) coefficient vectors.}
  \end{subfigure}
  \begin{subfigure}[b]{\columnwidth}
  \centering
  \vspace*{0.3cm}
  \begin{tikzpicture}
    \begin{axis}[xmin=-20,legend pos=north west,minor tick num=1]
%% MULTIPLOT(mode, B|ptitle) SELECT 100*eps AS x, avg(tlemptyfrac) AS y,
%% name || ', $b\!=\!' || B || '$' as ptitle, MULTIPLOT
%% FROM epstuning NATURAL JOIN modename WHERE sparse=1 AND L=64  AND B<512
%% GROUP BY MULTIPLOT,x ORDER BY B desc, mode, x
\addplot coordinates { (-12.5,0.0360624) (-12.1875,0.0334285) (-11.875,0.0309274) (-11.5625,0.0286042) (-11.25,0.0264118) (-10.9375,0.0243555) (-10.625,0.0224514) (-10.3125,0.0206706) (-10.0,0.0190135) (-9.6875,0.0174812) (-9.375,0.0160566) (-9.0625,0.0147263) (-8.75,0.0134989) (-8.4375,0.0123841) (-8.125,0.0113521) (-7.8125,0.0104081) (-7.5,0.00954375) (-7.1875,0.00875755) (-6.875,0.00805252) (-6.5625,0.00741493) (-6.25,0.00684671) (-5.9375,0.00634476) (-5.625,0.00590949) (-5.3125,0.00554159) (-5.0,0.00524079) (-4.6875,0.00500926) (-4.375,0.00485287) (-4.0625,0.00476782) (-3.75,0.00475924) (-3.4375,0.00483434) (-3.125,0.00499607) (-2.8125,0.00525153) (-2.5,0.00560306) (-2.1875,0.00606621) (-1.875,0.00664329) (-1.5625,0.00733856) (-1.25,0.00816003) (-0.9375,0.0091165) (-0.625,0.0102072) (-0.3125,0.0114314) (0.0,0.0127925) };
\addlegendentry{plain, $b\!=\!256$};
\addplot coordinates { (-7.84314,0.0758595) (-7.64706,0.0561849) (-7.45098,0.0550221) (-7.2549,0.0538545) (-7.05882,0.052708) (-6.86274,0.0515598) (-6.66667,0.050438) (-6.47059,0.0493184) (-6.27451,0.0481869) (-6.07843,0.0323558) (-5.88235,0.0316566) (-5.68627,0.0309376) (-5.4902,0.0302542) (-5.29412,0.0295729) (-5.09804,0.0289296) (-4.90196,0.0283079) (-4.70588,0.0277123) (-4.5098,0.0167738) (-4.31373,0.0165225) (-4.11765,0.0163115) (-3.92157,0.0161231) (-3.72549,0.0159535) (-3.52941,0.0158161) (-3.33333,0.015698) (-3.13725,0.0156308) (-2.94118,0.00960843) (-2.7451,0.00979105) (-2.54902,0.0100009) (-2.35294,0.0102436) (-2.15686,0.0105324) (-1.96078,0.0108576) (-1.76471,0.0112186) (-1.56863,0.0116249) (-1.37255,0.00917441) (-1.17647,0.00972912) (-0.980392,0.0103429) (-0.784314,0.0109982) (-0.588235,0.0117034) (-0.392157,0.012459) (-0.196078,0.0132585) (0.0,0.0127924) };
\addlegendentry{\onebits, $b\!=\!256$};
\addplot coordinates { (-12.5,0.123777) (-12.1875,0.121639) (-11.875,0.116322) (-11.5625,0.111305) (-11.25,0.106545) (-10.9375,0.101992) (-10.625,0.0977083) (-10.3125,0.0935827) (-10.0,0.0896568) (-9.6875,0.088211) (-9.375,0.0843772) (-9.0625,0.0645807) (-8.75,0.0611436) (-8.4375,0.0579156) (-8.125,0.0549001) (-7.8125,0.052063) (-7.5,0.0493858) (-7.1875,0.0481318) (-6.875,0.0456637) (-6.5625,0.0433146) (-6.25,0.0410989) (-5.9375,0.0390488) (-5.625,0.0371011) (-5.3125,0.0352994) (-5.0,0.0251338) (-4.6875,0.0237857) (-4.375,0.0225927) (-4.0625,0.0215656) (-3.75,0.0206575) (-3.4375,0.0198946) (-3.125,0.0192679) (-2.8125,0.0187718) (-2.5,0.0188228) (-2.1875,0.0185672) (-1.875,0.0184484) (-1.5625,0.01847) (-1.25,0.0186586) (-0.9375,0.0157962) (-0.625,0.0164268) (-0.3125,0.0172052) (0.0,0.0182022) };
\addlegendentry{2-bit, $b\!=\!256$};
\addplot coordinates { (-18.75,0.00278973) (-18.125,0.00250117) (-17.5,0.00224098) (-16.875,0.0020128) (-16.25,0.00181191) (-15.625,0.001632) (-15.0,0.00147525) (-14.375,0.00133466) (-13.75,0.00120985) (-13.125,0.00109932) (-12.5,0.00100147) (-12.1875,0.000958046) (-11.875,0.000915745) (-11.5625,0.000878519) (-11.25,0.000841176) (-10.9375,0.000807127) (-10.625,0.000775578) (-10.3125,0.00074843) (-10.0,0.000722181) (-9.6875,0.000699269) (-9.375,0.000680975) (-9.0625,0.000664101) (-8.75,0.000651436) (-8.4375,0.000644704) (-8.125,0.000641213) (-7.8125,0.000644188) (-7.5,0.000653041) (-7.1875,0.000670534) (-6.875,0.000696124) (-6.5625,0.000732291) (-6.25,0.000782681) (-5.9375,0.000847519) (-5.625,0.000931423) (-5.3125,0.00103437) (-5.0,0.00116475) (-4.6875,0.00132677) (-4.375,0.00152464) (-4.0625,0.00176217) (-3.75,0.00204919) (-3.4375,0.00239109) (-3.125,0.0027972) (-2.8125,0.00327694) (-2.5,0.00383262) (-2.1875,0.00448323) (-1.875,0.00522795) (-1.5625,0.00607758) (-1.25,0.0070375) (-0.9375,0.00811631) (-0.625,0.00932256) (-0.3125,0.0106515) (0.0,0.012105) };
\addlegendentry{plain, $b\!=\!128$};
\addplot coordinates { (-14.8148,0.0764426) (-14.4444,0.0738007) (-14.0741,0.0712267) (-13.7037,0.0687178) (-13.3333,0.0662908) (-12.963,0.038623) (-12.5926,0.0371323) (-12.2222,0.0356804) (-11.8519,0.0342898) (-11.4815,0.0329481) (-11.1111,0.0316583) (-10.7407,0.030438) (-10.3704,0.0292665) (-10.0,0.0126243) (-9.62963,0.012152) (-9.25926,0.0117085) (-8.88889,0.0112906) (-8.51852,0.0109159) (-8.14815,0.0105767) (-7.77778,0.0102743) (-7.40741,0.00999701) (-7.03704,0.00314588) (-6.66667,0.00321417) (-6.2963,0.00332473) (-5.92593,0.00345909) (-5.55556,0.00363115) (-5.18518,0.00384803) (-4.81481,0.00411389) (-4.44444,0.00443247) (-4.07407,0.00480772) (-3.7037,0.00297119) (-3.33333,0.00347906) (-2.96296,0.00407544) (-2.59259,0.00478023) (-2.22222,0.00559954) (-1.85185,0.00654057) (-1.48148,0.00761352) (-1.11111,0.00883904) (-0.740741,0.00886103) (-0.37037,0.0103968) (0.0,0.0121048) };
\addlegendentry{\onebits, $b\!=\!128$};
\addplot coordinates { (-12.5,0.00641042) (-12.1875,0.00611223) (-11.875,0.00581966) (-11.5625,0.00555844) (-11.25,0.00529898) (-10.9375,0.00506945) (-10.625,0.00485055) (-10.3125,0.00465007) (-10.0,0.0044648) (-9.6875,0.0043063) (-9.375,0.00415055) (-9.0625,0.00402345) (-8.75,0.00390371) (-8.4375,0.00380242) (-8.125,0.00372375) (-7.8125,0.00366831) (-7.5,0.0032872) (-7.1875,0.00330515) (-6.875,0.00332761) (-6.5625,0.00338019) (-6.25,0.00346225) (-5.9375,0.00356333) (-5.625,0.00261241) (-5.3125,0.00268839) (-5.0,0.00280848) (-4.6875,0.00296293) (-4.375,0.00316017) (-4.0625,0.00340902) (-3.75,0.00370591) (-3.4375,0.00406854) (-3.125,0.004491) (-2.8125,0.00464008) (-2.5,0.00523221) (-2.1875,0.00591506) (-1.875,0.00669336) (-1.5625,0.00756671) (-1.25,0.00854027) (-0.9375,0.00962701) (-0.625,0.0108179) (-0.3125,0.0121256) (0.0,0.0135478) };
\addlegendentry{2-bit, $b\!=\!128$};
\addplot coordinates { (-26.6667,0.0528637) (-26.0,0.0497528) (-25.3333,0.0468506) (-24.6667,0.0441817) (-24.0,0.041705) (-23.3333,0.0393966) (-22.6667,0.0372115) (-22.0,0.0351927) (-21.3333,0.0144142) (-20.6667,0.0132182) (-20.0,0.0121589) (-19.3333,0.0111824) (-18.6667,0.0103034) (-18.0,0.00952124) (-17.3333,0.00882506) (-16.6667,0.00820777) (-16.0,0.00765325) (-15.3333,0.00355073) (-14.6667,0.00320712) (-14.0,0.00290784) (-13.3333,0.00264188) (-12.6667,0.00242257) (-12.0,0.0022381) (-11.3333,0.00209049) (-10.6667,0.00198332) (-10.0,0.00191388) (-9.33333,0.00189165) (-8.66667,0.00175312) (-8.0,0.00164598) (-7.33333,0.00159684) (-6.66667,0.00161844) (-6.0,0.00173224) (-5.33333,0.00195977) (-4.66667,0.00234224) (-4.0,0.00291038) (-3.33333,0.0037175) (-2.66667,0.00360863) (-2.0,0.0050003) (-1.33333,0.00685694) (-0.666667,0.00923899) (0.0,0.0121862) };
\addlegendentry{\onebits, $b\!=\!64$};
\addplot coordinates { (-18.75,0.00655782) (-18.125,0.00612713) (-17.5,0.00571341) (-16.875,0.00530629) (-16.25,0.00491653) (-15.625,0.00453825) (-15.0,0.00417994) (-14.375,0.00383646) (-13.75,0.00351061) (-13.125,0.00397983) (-12.5,0.00357791) (-12.1875,0.00339439) (-11.875,0.00321697) (-11.5625,0.00304659) (-11.25,0.00289064) (-10.9375,0.00273669) (-10.625,0.00259549) (-10.3125,0.00246236) (-10.0,0.00234113) (-9.6875,0.00222523) (-9.375,0.00212041) (-9.0625,0.00202522) (-8.75,0.00194374) (-8.4375,0.00186916) (-8.125,0.0018113) (-7.8125,0.00176301) (-7.5,0.00173551) (-7.1875,0.0017131) (-6.875,0.0017128) (-6.5625,0.00109362) (-6.25,0.00115631) (-5.9375,0.00124049) (-5.625,0.00134297) (-5.3125,0.00147114) (-5.0,0.00162334) (-4.6875,0.00181209) (-4.375,0.00203321) (-4.0625,0.00229518) (-3.75,0.00260856) (-3.4375,0.00458465) (-3.125,0.0047988) (-2.8125,0.00511393) (-2.5,0.00550306) (-2.1875,0.00599515) (-1.875,0.00659439) (-1.5625,0.00730775) (-1.25,0.00814516) (-0.9375,0.00910817) (-0.625,0.010194) (-0.3125,0.0114155) (0.0,0.0127833) };
\addlegendentry{2-bit, $b\!=\!64$};

    \end{axis}
  \end{tikzpicture}
  \vspace*{-0.2cm}
  \caption{Ribbon width $w=64$, sparse coefficient vectors with 8 out of 64 positions occupied.}
  \end{subfigure}
  \caption{\label{fig:furtherEpsilon}Fraction of empty slots for various configurations of BuRR, depending on the overloading factor $\eps$.}
\end{figure}

\begin{table}[b]
\def\num#1{\ifnum#1<100\phantom{0}#1\else#1\fi}

\caption{Experimental performance comparisons.  Overhead, construction and query
  times (positive and negative queries) for various AMQs. Tested configurations:
  $n=10^6$ keys, $n=10^8$ keys, both sequential, as well as 1280 AMQs with
  $n=10^7$ keys each (total: $1.28\times10^{10}$ keys), constructed and
  queried in parallel using 64 threads, with each query operating on a randomly
  chosen AMQ.}
\label{tab:performance}
\centering
 \begin{tabular}{rc|rrr|rrr|rrr}
   \toprule
                & Space  & \multicolumn{3}{c|}{ns/key, $n=10^6$} & \multicolumn{3}{c|}{ns/key, $n=10^8$} & \multicolumn{3}{c}{parallel, $n\!=\!10^7$} \\
  Configuration & ovr \% & con & pos & neg & con & pos & neg & con & pos & neg \\
  \midrule
    \multicolumn{11}{c}{$\downarrow\,\,$ False positive rate around \textbf{1\,\%}, ribbons using $r=7$ $\,\,\downarrow$} \\
   \midrule
%                       name & overh & con & qpo & qne & con & qpo & qne & con & qpo & qne \\
Blocked Bloom \cite{LNKB:Bloom:2019}                  &    52.0 &    3 &   3 &   3 &   17 &  24 &  24 &   50 &  116 & 101 \\
Blocked Bloom \cite{Dillinger:RocksDBBloom}           &    49.8 &    7 &   4 &   4 &   50 &  26 &  26 &   78 &  109 &  85 \\
Blocked Bloom \cite{Dillinger:RocksDBBloom} $k=2$     &    45.0 &    9 &   7 &   7 &   57 &  43 &  43 &  110 &  149 & 194 \\

Cuckoo12 \cite{FAK:CuckooFilterBetter:2013} $\dagger$ &    46.3 &   29 &  10 &   7 &  118 &  56 &  52 &  239 &  162 & 282 \\
Cuckoo12 \cite{FAK:CuckooFilterBetter:2013}           &    40.4 &   35 &  12 &   7 &  166 &  58 &  51 &  288 &  180 & 271 \\
Morton \cite{BJ:MortonFilters:2020}                   &    40.6 &   32 &  25 &  22 &   64 &  96 &  87 &  130 &  182 & 203 \\ % Morton3_8
Xor \cite{GL:XorFilters:2020} $r=7$ $\ddagger$        &    23.0 &   91 &   8 &   8 &  169 &  56 &  56 &  644 &  333 & 348 \\
Xor \cite{GL:XorFilters:2020} $r=8$                   &    23.0 &   91 &   5 &   5 &  159 &  41 &  41 &  586 &  386 & 392 \\
% BPZ $r=8$                                           &    23.4 &  124 &   7 &   7 &  305 &  56 &  56 &  962 &  404 & 446 \\
Xor+ \cite{GL:XorFilters:2020} $r=8$                  &    14.4 &   94 &  14 &  15 &  209 &  86 &  85 &  853 &  372 & 475 \\
XorFuse \cite{DW:DensePeelable:2019} $r=8$            &   16;14 &   89 &   6 &   6 &  215 &  43 &  44 &  453 &  532 & 534 \\
 LMSS \cite{LMSS:Efficient_Erasure:2001} $D\!=\!12, c\!=\!0.91, r\!=\!8$ &  11.1 &  421 &  28 &  28 &  779 & 134 & 134 & \multicolumn{3}{c}{not tested} \\
LMSS \cite{LMSS:Efficient_Erasure:2001} $D\!=\!150, c\!=\!0.99, r\!=\!8$ &   1.0 &  464 &  34 &  34 &  877 & 152 & 152 & \multicolumn{3}{c}{not tested} \\
% Coupled \cite{W:SpatialCoupling:2021} {\footnotesize $k\!=\!4, z\!=\!120, c\!=\!0.96, r\!=\!8$} & 5.0 & \multicolumn{3}{c|}{not applicable} & 241 &  58 & 58 & 543 & 580 & 579 \\
%   Coupled \cite{W:SpatialCoupling:2021} $k=3, r=8$  &   15;10 &   86 &   8 &   8 &  222 &  56 &  56 &  505 &  410 & 417 \\
   Coupled \cite{W:SpatialCoupling:2021} $k=4, r=8$   &     8;4 &  104 &   9 &   9 &  229 &  59 &  59 &  546 &  579 & 578 \\
   Coupled \cite{W:SpatialCoupling:2021} $k=7, r=8$   &     6;2 &  169 &  12 &  12  &  331 &  91 &  91&  813 & 1337 & 1342 \\
Quotient Filter \cite{BFJKKMMSS:QuotientFilters:2012} &    81.9 &   69 & 432 & 272 &  114 & 225 & 169 &  133 &  385 & 308 \\
Counting Quotient Filter \cite{Pandey:CQF:2017}       &   67;55 &   60 &  45 &  31 &  172 & 153 & 113 &  183 &  307 & 252 \\

       Standard Ribbon $w=64$ &   14;20 &   32 &  16 &  20 &   70 &  78 &  66 &  324 &  234 & 194 \\
      Standard Ribbon $w=128$ &     6;8 &   69 &  24 &  25 &  121 & 140 &  87 &  464 &  296 & 206 \\

         Homog. Ribbon $w=16$ &    52.2 &   19 &  14 &  20 &   42 &  73 &  61 &   69 &  148 & 128 \\
         Homog. Ribbon $w=32$ &    20.7 &   20 &  13 &  19 &   50 &  73 &  60 &  105 &  147 & 168 \\
         Homog. Ribbon $w=64$ &     9.9 &   28 &  14 &  19 &   67 &  75 &  63 &  155 &  164 & 170 \\
        Homog. Ribbon $w=128$ &     4.9 &   58 &  21 &  23 &  118 & 135 &  85 &  306 &  292 & 208 \\

\textbf{Bu$^1$RR} $w=32\dagger$ &    10.3 &   40 &  21 &  26 &   94 & 125 &  88 &  163 &  275 & 286 \\
       \textbf{Bu$^1$RR} $w=32$ &     2.4 &   62 &  21 &  26 &  121 & 123 &  88 &  174 &  261 & 239 \\
     \textbf{BuRR} plain $w=32$ &     1.4 &   76 &  20 &  26 &   81 &  82 &  79 &  151 &  247 & 210 \\
     \textbf{BuRR} 2-bit $w=32$ &     1.3 &   77 &  19 &  26 &   82 &  82 &  80 &  152 &  233 & 245 \\
  \textbf{BuRR} \onebits $w=32$ &    0.82 &   80 &  29 &  34 &   84 &  88 &  88 &  158 &  240 & 231 \\

       \textbf{Bu$^1$RR} $w=64$ &    0.62 &  121 &  21 &  26 &  188 & 128 &  90 &  197 &  292 & 300 \\
     \textbf{BuRR} plain $w=64$ &    0.48 &  109 &  18 &  24 &  115 &  82 &  74 &  182 &  229 & 213 \\
     \textbf{BuRR} 2-bit $w=64$ &    0.25 &  110 &  19 &  24 &  115 &  82 &  74 &  190 &  215 & 215 \\
% BuRR 2-bit sparse $w\!=\!64,r\!=\!8$ &    0.56 &  100 &  21 &  21 &  105 &  80 &  80 &  161 &  226 & 254 \\
% BuRR 2-bit contig $w\!=\!64,r\!=\!8$ &    0.23 &  111 &  47 &  47 &  116 & 145 & 144 &  176 &  307 & 306 \\
% BuRR 2-bit contig $w\!=\!32,r\!=\!8$ &     1.2 &   80 &  33 &  32 &   84 & 101 & 101 &  147 &  262 & 269 \\
  \textbf{BuRR} \onebits $w=64$ &    0.21 &  110 &  27 &  33 &  115 &  86 &  84 &  189 &  238 & 228 \\

      \textbf{Bu$^1$RR} $w=128$ &    0.31 &  332 &  29 &  30 &  442 & 147 & 100 &  427 &  369 & 285 \\
    \textbf{BuRR} plain $w=128$ &    0.18 &  209 &  27 &  27 &  214 & 141 &  88 &  319 &  317 & 226 \\
    \textbf{BuRR} 2-bit $w=128$ &    0.10 &  186 &  27 &  27 &  191 & 140 &  88 &  294 &  313 & 211 \\
 \textbf{BuRR} \onebits $w=128$ &    0.06 &  208 &  34 &  37 &  214 & 142 &  92 &  304 &  319 & 235 \\
 \midrule
    \multicolumn{11}{l}{$\dagger$ Larger space allocated to improve construction time.} \\
    \multicolumn{11}{l}{$\ddagger$ Potentially unfavorable bit alignment.} \\
    \multicolumn{11}{l}{$;$ Standard Ribbon, XorFuse, Coupled, and CQF space overhead depend on $n$.} \\
   \bottomrule
 \end{tabular}
\end{table}
\begin{table}
 \caption{Experimental performance comparisons (continued from \cref{tab:performance}).}
 \label{tab:performance2}
 \centering
  \begin{tabular}{rc|rll|rll|rll}
  \toprule
                & Space  & \multicolumn{3}{c|}{ns/key, $n=10^6$} & \multicolumn{3}{c|}{ns/key, $n=10^8$} & \multicolumn{3}{c}{parallel, $n=10^7$} \\
  Configuration & ovr \% & con & pos & neg & con & pos & neg & con & pos & neg \\
    \midrule
    \multicolumn{11}{c}{$\downarrow\,\,$ False positive rate around \textbf{10\,\%}, ribbons using $r=3$ $\,\,\downarrow$} \\
    \midrule
Xor \cite{GL:XorFilters:2020} $r=3\ddagger$ &  23.0 &  91 &   6 &   6 & 169 &  51 &  51 & 633 & 292 & 270 \\
   Standard Ribbon $w=64$ &   14;20 &   27 &   9 &   9 &   65 &  54 &  51 &  326 &  118 & 138 \\
  Standard Ribbon $w=128$ &     6;8 &   62 &  14 &  14 &  114 &  67 &  67 &  467 &  139 & 191 \\
     Homog. Ribbon $w=16$ &    34.6 &   16 &   9 &   9 &   40 &  58 &  59 &   67 &  141 & 140 \\
     Homog. Ribbon $w=32$ &    16.1 &   17 &   8 &   8 &   49 &  45 &  45 &  101 &  153 & 151 \\
     Homog. Ribbon $w=64$ &     8.0 &   24 &   8 &   8 &   65 &  45 &  45 &  158 &  132 & 145 \\
    Homog. Ribbon $w=128$ &     4.0 &   57 &  13 &  13 &  119 &  65 &  65 &  324 &  162 & 193 \\
       \textbf{Bu$^1$RR} $w=32$ &     2.8 &   60 &  14 &  14 &  119 &  67 &  67 &  175 &  193 & 212 \\
     \textbf{BuRR} plain $w=32$ &     3.2 &   73 &  14 &  21 &   78 &  64 &  64 &  145 &  180 & 170 \\
     \textbf{BuRR} 2-bit $w=32$ &     2.5 &   74 &  14 &  21 &   79 &  64 &  65 &  147 &  227 & 214 \\
  \textbf{BuRR} \onebits $w=32$ &     1.6 &   77 &  23 &  30 &   82 &  67 &  74 &  150 &  202 & 215 \\
       \textbf{Bu$^1$RR} $w=64$ &     0.8 &  118 &  13 &  13 &  187 &  66 &  67 &  195 &  192 & 220 \\
     \textbf{BuRR} plain $w=64$ &     1.1 &  104 &  13 &  20 &  110 &  63 &  62 &  177 &  160 & 172 \\
     \textbf{BuRR} 2-bit $w=64$ &     0.6 &  106 &  13 &  21 &  111 &  64 &  64 &  188 &  174 & 181 \\ % 0.55%
  \textbf{BuRR} \onebits $w=64$ &     0.4 &  106 &  22 &  29 &  111 &  66 &  71 &  189 &  196 & 221 \\ % 0.42%
%   \textbf{BuRR} plain $w=128$ &    0.38 &  206 &  17 &  23 &  211 &  70 &  74 &  321 &  199 & 222 \\
%   \textbf{BuRR} 2-bit $w=128$ &    0.18 &  191 &  17 &  23 &  196 &  70 &  74 &  303 &  171 & 180 \\
%\textbf{BuRR} \onebits $w=128$ &    0.12 &  191 &  28 &  34 &  197 &  72 &  77 &  304 &  182 & 217 \\

\midrule
    \multicolumn{11}{c}{$\downarrow\,\,$ False positive rate around $\mathbf{2^{-11}\approx 0.05\,\%}$, ribbons using $r=11$ $\,\,\downarrow$} \\
   \midrule
%                       name & overh & con & qpo & qne & con & qpo & qne & con & qpo & qne \\
Cuckoo16 \cite{FAK:CuckooFilterBetter:2013}           &    30.1 &   31 &  11 &   7 &  156 &  56 &  50 &  309 &  180 & 294 \\
Cuckoo16 \cite{FAK:CuckooFilterBetter:2013} $\dagger$ &    35.7 &   28 &  10 &   7 &  119 &  56 &  44 &  243 &  188 & 308 \\
CuckooSemiSort                                        &    26.6 &   64 &  15 &  14 &  259 &  79 &  79 &  376 &  264 & 326 \\
Morton \cite{BJ:MortonFilters:2020}                   &    36.8 &   38 &  40 &  35 &   69 & 167 & 156 &  127 &  314 & 305 \\ % Morton7_12
Xor \cite{GL:XorFilters:2020} $r=12$                  &    23.0 &   89 &   8 &   7 &  163 &  57 &  57 &  632 &  415 & 435 \\
%Xor \cite{GL:XorFilters:2020} $r=10.666\ddagger$      &    23.0 &   91 &   7 &   7 &  206 &  51 &  53 &  656 &  414 & 422 \\
Xor+ \cite{GL:XorFilters:2020} $r=11\ddagger$         &    12.8 &   98 &  16 &  16 &  215 & 101 &  99 &  759 &  470 & 494 \\
Quotient Filter \cite{BFJKKMMSS:QuotientFilters:2012} &    93.6 &   71 & 485 & 304 &  109 & 235 & 175 &  138 &  402 & 324 \\

     Standard Ribbon, $w=64$ &   14;20 &   38 &  23 &  21 &   76 & 143 &  69 &  342 &  294 & 205 \\
    Standard Ribbon, $w=128$ &     6;8 &   71 &  33 &  26 &  124 & 158 &  91 &  442 &  336 & 243 \\
%      Homog. Ribbon, $w=16$ &    88.3 &   33 &  23 &  20 &   55 & 135 &  64 &   83 &  263 & 183 \\
       Homog. Ribbon, $w=32$ &    28.5 &   24 &  18 &  20 &   53 &  82 &  63 &  109 &  230 & 179 \\
       Homog. Ribbon, $w=64$ &    12.1 &   32 &  18 &  20 &   70 &  86 &  65 &  165 &  218 & 195 \\
      Homog. Ribbon, $w=128$ &     6.5 &   58 &  30 &  25 &  115 & 155 &  89 &  281 &  333 & 305 \\

     \textbf{Bu$^1$RR}, $w=32$ &     2.4 &   66 &  28 &  27 &  125 & 152 &  92 &  178 &  405 & 277 \\
    \textbf{BuRR} plain $w=32$ &    0.94 &   81 &  28 &  26 &   86 & 145 &  83 &  149 &  327 & 229 \\
    \textbf{BuRR} 2-bit $w=32$ &    0.95 &   82 &  27 &  26 &   87 & 145 &  81 &  150 &  328 & 217 \\
 \textbf{BuRR} \onebits $w=32$ &    0.61 &   85 &  35 &  35 &   90 & 148 &  92 &  158 &  352 & 220 \\
     \textbf{Bu$^1$RR}, $w=64$ &    0.57 &  128 &  28 &  27 &  196 & 152 &  95 &  203 &  400 & 296 \\
    \textbf{BuRR} plain $w=64$ &    0.32 &  116 &  27 &  25 &  121 & 146 &  79 &  175 &  334 & 250 \\
    \textbf{BuRR} 2-bit $w=64$ &    0.17 &  118 &  27 &  25 &  123 & 146 &  77 &  164 &  328 & 209 \\
 \textbf{BuRR} \onebits $w=64$ &    0.14 &  115 &  33 &  34 &  120 & 147 &  90 &  189 &  353 & 239 \\
%   \textbf{Bu$^1$RR}, $w=128$ &    0.29 &  334 &  38 &  32 &  444 & 168 & 106 &  429 &  417 & 365 \\
   \textbf{BuRR} plain $w=128$ &    0.13 &  214 &  35 &  28 &  220 & 159 &  92 &  305 &  364 & 244 \\
   \textbf{BuRR} 2-bit $w=128$ &    0.08 &  202 &  35 &  28 &  208 & 159 &  92 &  297 &  358 & 226 \\
\textbf{BuRR} \onebits $w=128$ &    0.05 &  214 &  42 &  38 &  219 & 160 &  96 &  317 &  389 & 294 \\
 \midrule
    \multicolumn{6}{l}{$\dagger$ Larger space allocated to improve construction time.} \\
    \multicolumn{6}{l}{$\ddagger$ Potentially unfavorable bit alignment.} \\
   \multicolumn{6}{l}{$;$ Standard Ribbon space overhead depends on $n$.} \\
   \bottomrule
 \end{tabular}
\end{table}

\begin{table}
  \caption{%\centering
    Selected BuRR configurations for various $r$.  Sparse coefficient %\protect\linebreak
    vectors used for rows with threshold compression mode marked $^s$.}\label{tab:configs}
  \centering
  \begin{tabular}{rrrllrrr}
    \toprule
        &     &     & thresh & overloading & empty      & metabits & estimated \\
    $r$ & $w$ & $b$ & mode & factor $\eps$ & slots (\%) & /bucket  & overhead (\%) \\
    \midrule
    % copied from stats.tex

 1 & 128 & 512 & \onebits & -0.0470588 & 0.017477 & 1.3194 & 0.275219 \\
 1 &  64 & 256 &    2-bit &  -0.034375 & 0.442091 &      2 & 1.226810 \\
 1 &  64 & 128 & 2-bit$^s$&   -0.05625 & 0.261241 &      2 & 1.827833 \\
 1 &  32 &  64 &    2-bit &   -0.08125 & 0.681598 &      2 & 3.828044 \\\midrule
 2 & 128 & 512 & \onebits & -0.0470588 & 0.017477 & 1.3194 & 0.146348 \\
 2 &  64 & 128 & \onebits & -0.0888889 & 0.062375 & 1.2511 & 0.551393 \\
 2 &  64 & 128 &    2-bit &   -0.08125 & 0.013288 &      2 & 0.794642 \\
 2 &  64 & 128 & 2-bit$^s$&   -0.05625 & 0.261241 &      2 & 1.044537 \\
 2 &  32 &  64 &    2-bit &   -0.08125 & 0.681598 &      2 & 2.254821 \\\midrule
 4 & 128 & 512 & \onebits & -0.0431373 & 0.008797 & 1.4402 & 0.079125 \\
 4 &  64 & 128 &    2-bit &   -0.08125 & 0.013288 &      2 & 0.403965 \\
 4 &  64 & 128 & 2-bit$^s$&   -0.05625 & 0.261241 &      2 & 0.652889 \\
 4 &  32 &  64 &    2-bit &   -0.08125 & 0.681598 &      2 & 1.468210 \\\midrule
 8 & 128 & 512 & \onebits & -0.0431373 & 0.008797 & 1.4402 & 0.043961 \\
 8 &  64 & 128 &    2-bit &   -0.08125 & 0.013288 &      2 & 0.208626 \\
 8 &  64 & 128 & 2-bit$^s$&   -0.05625 & 0.261241 &      2 & 0.457065 \\
 8 &  32 &  64 &    2-bit &   -0.08125 & 0.681598 &      2 & 1.074904 \\\midrule
16 & 128 & 512 & \onebits & -0.0411765 & 0.007102 & 1.5117 & 0.025557 \\
16 &  64 & 128 &    2-bit &   -0.08125 & 0.013288 &      2 & 0.110957 \\
16 &  64 &  64 & 2-bit$^s$&  -0.065625 & 0.109362 &      2 & 0.304888 \\
16 &  32 &  64 &    plain &   -0.13125 & 0.046154 &      6 & 0.632362 \\
    \bottomrule
  \end{tabular}
\end{table}

\begin{figure}\centering
  \includegraphics[page=3,width=\textwidth]{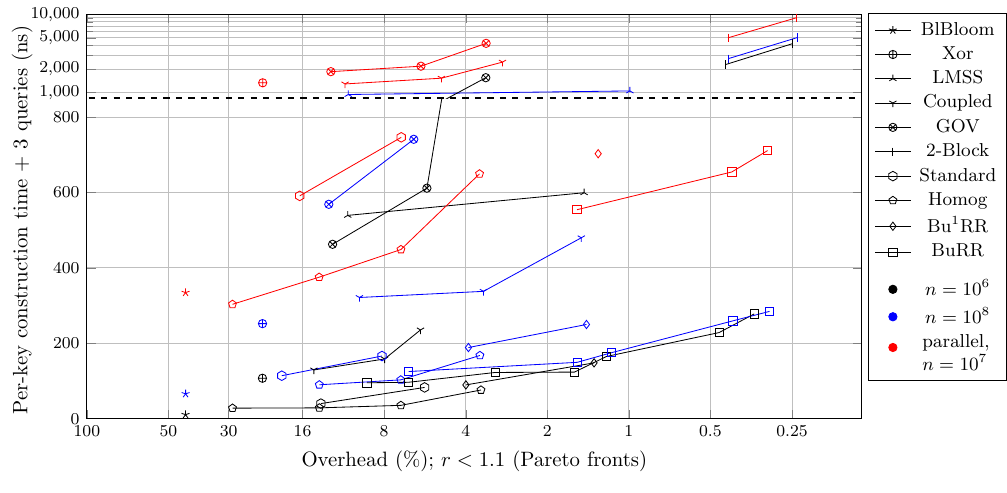}
  \caption{\label{fig:scatter2}Performance--overhead trade-off for false-positive
    rate $>46\,\%$ for different AMQs and different inputs. This large
    false-positive rate is the only one for which we have implementations for
    GOV \cite{GOV:retrieval-Compressed:2020} and 2-block \cite{DW:Retrieval-log-extra-bits:2019}.
    Note that the vertical axis switches to a logarithmic scale above 900\,ns.}
\end{figure}
\begin{figure}\centering
  \includegraphics[page=4,width=\textwidth]{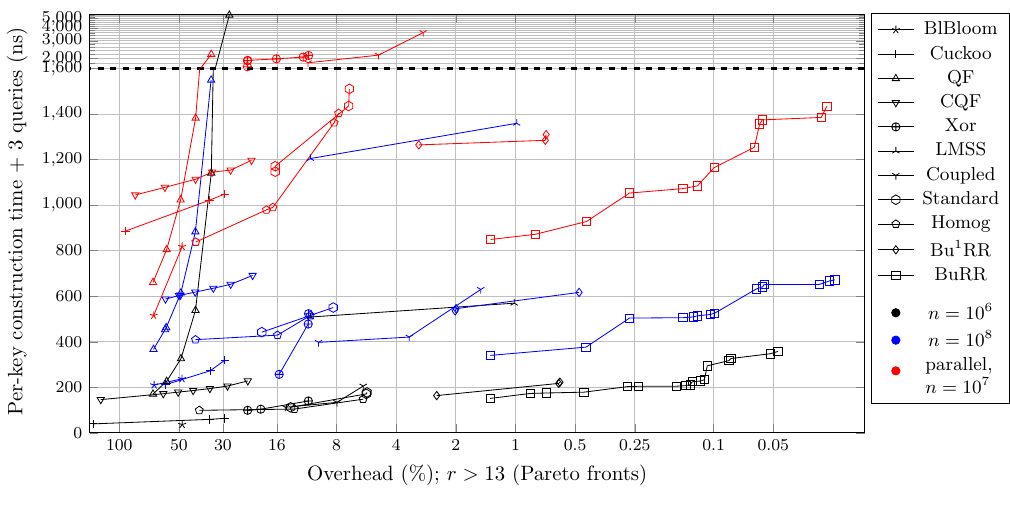}
  \caption{\label{fig:scatter16}Performance--overhead trade-off for
    false-positive rate $<2^{-13} \approx 0.01\,\%$ for
    different AMQs and different inputs.  Logarithmics vertical axis above 1600\,ns.}
\end{figure}
\begin{figure}\centering
  \includegraphics[page=5,width=\textwidth]{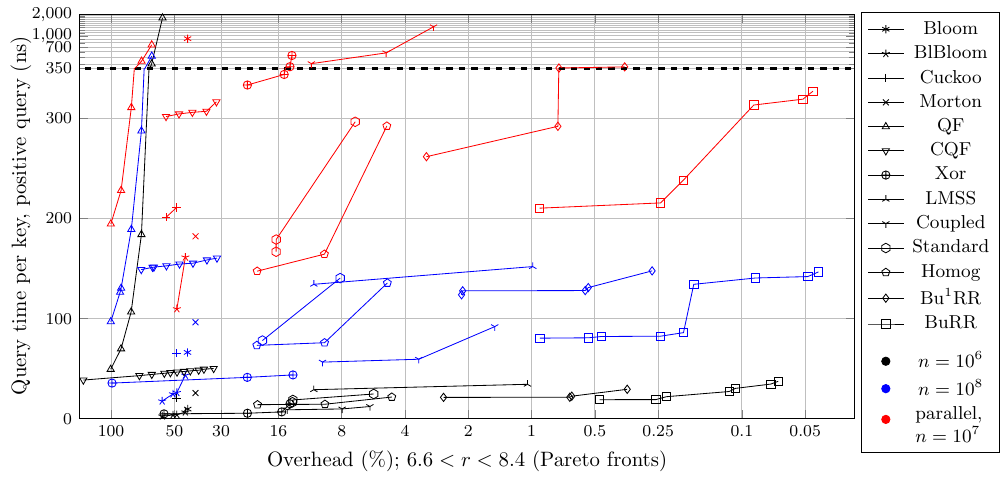}
  \caption{\label{fig:scatterQuery}Query time--overhead trade-off for
    positive queries, false-positive rate between $0.3\,\%$ and $1\,\%$ for different AMQs and
    different inputs. Note that Xor filters have excellent query time
    sequentially where random fetches can be performed in
    parallel but are far from optimal in the parallel setting where
    the total number of memory accesses matters most.  Logarithmic vertical
    axis above 350\,ns.}
\end{figure}
\begin{figure}\centering
  \includegraphics[page=9,width=\textwidth]{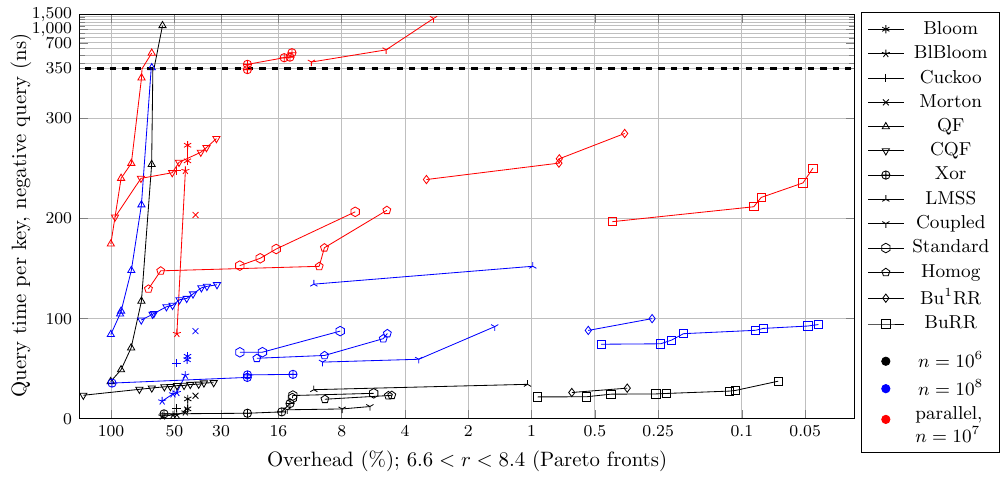}
  \caption{\label{fig:scatterQueryNeg}Query time--overhead trade-off for negative
    queries, false-positive rate between $0.3\,\%$ and $1\,\%$ for different
    AMQs and different inputs. Again, Xor filters perform well sequentially but
    suffer in the parallel case.  Logarithmic vertical axis above 350\,ns.}
\end{figure}
\begin{figure}\centering
  \includegraphics[page=6,width=\textwidth]{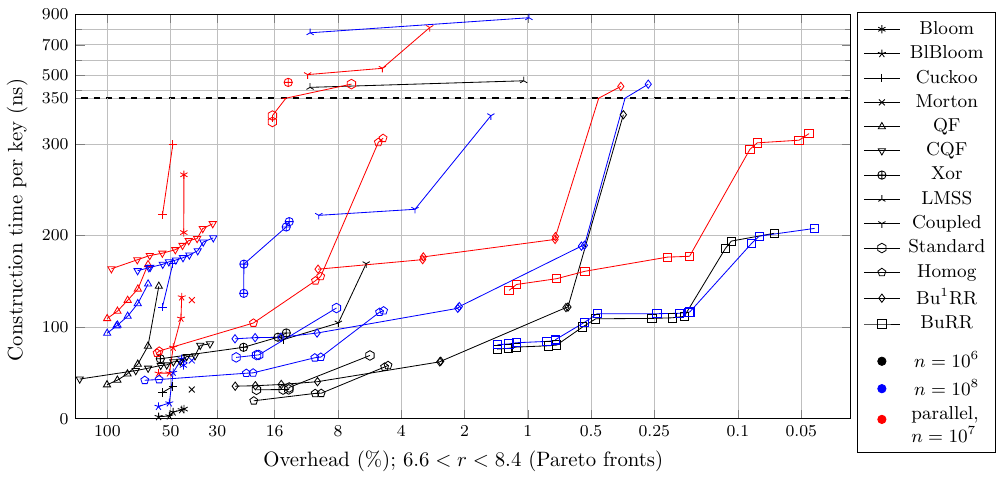}
  \caption{\label{fig:scatterConstruction}Construction time--overhead trade-off for
    false-positive rate between $0.3\,\%$ and $1\,\%$ for different AMQs and
    different inputs.  Compressed vertical axis above 350\,ns.}
\end{figure}

\begin{figure}
  \begin{subfigure}[t]{\linewidth}\centering
    \includegraphics[page=4]{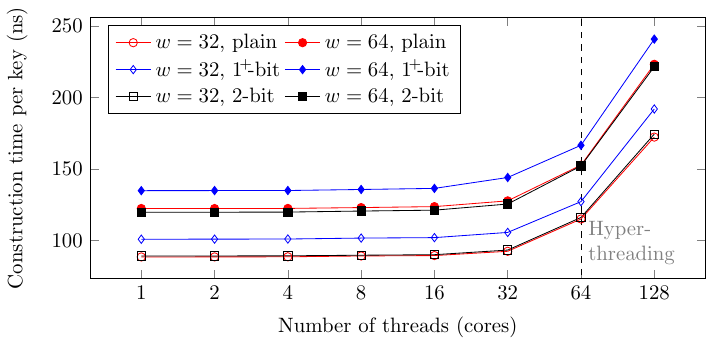}
    \caption{Construction of $10\,000$ filters with $10^6$ keys each
      ($w=32$) or $5\,000$ filters with $2\cdot10^6$ keys each ($w=64$),
      for a total of $10^{10}$ keys; strong scaling.}
  \end{subfigure}
  \begin{subfigure}[t]{\linewidth}\centering
    \vspace*{5mm}
    \includegraphics[page=5]{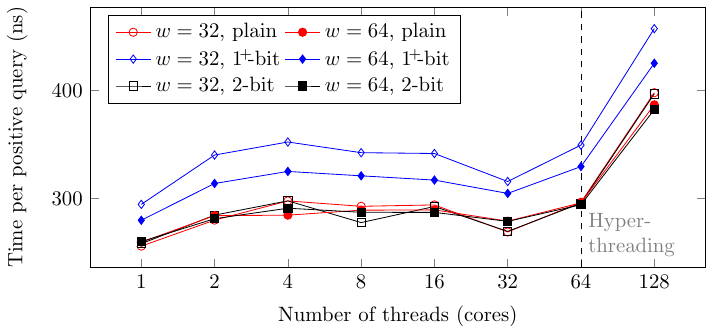}
    \caption{Scaling behavior of positive queries on the filters from
      (a).  Each query accesses a randomly chosen filter.  Tested with
      $10^8$ queries per thread.}
  \end{subfigure}
  \begin{subfigure}[t]{\linewidth}\centering
    \vspace*{5mm}
    \includegraphics[page=6]{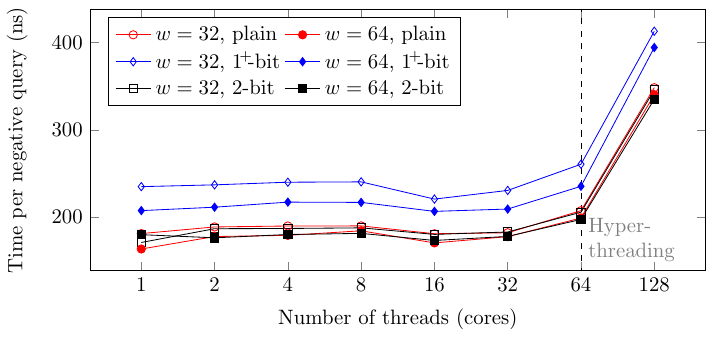}
    \caption{Scaling behavior of negative queries on the filters from
      (a).  Each query accesses a randomly chosen filter.  Tested with
      $10^8$ queries per thread.}
  \end{subfigure}
  \caption{Scaling experiments for parallel construction and querying of BuRR.\label{fig:scaling}}
\end{figure}

\clearpage

%%% Local Variables:
%%% mode: latex
%%% TeX-master: "main"
%%% End:

\iftr
\section{Full Experimental Results}\label{app:fullexp}

\Cref{tbl:full1,tbl:full2} below contain the results of our performance
experiments for all tested configurations for the experiments of \cref{s:exp},
first for retrieval-based approaches in \cref{tbl:full1}, and then for AMQ data
structures in \cref{tbl:full2}.
\vfill

\eject
\thispagestyle{empty}
\pdfpagewidth=10in
\pdfpageheight=55.5in

\begin{table}[h]\flushright
    \caption{Performance of retrieval-based approaches\label{tbl:full1}}
\begin{tabular}{l>{\raggedleft}p{4.5em}rrl|*{3}{>{\raggedleft}p{2em}}|*{3}{>{\raggedleft}p{2em}}|*{2}{>{\raggedleft}p{2.2em}}>{\raggedleft\arraybackslash}p{2.2em}}
   \toprule
&&&& Space  & \multicolumn{3}{c|}{ns/key, $n=10^6$} & \multicolumn{3}{c|}{ns/key, $n=10^8$} & \multicolumn{3}{c}{parallel, $n\!=\!10^7$} \\
                                          Method & Config      & $w$ & $r$ & ovr \%& con  & pos& neg& con  & pos & neg & con  & pos  & neg  \\
\midrule
%
%                                      &            &     &      & seq n=10^6           & seq n=10^8     & par n=10^7*k                                \\
%                       method         & config     & w   & r    & overhead & constr & pos q & neg q & constr & pos q & neg q & constr & pos q & neg q \\\midrule\midrule
\multicolumn{14}{c}{$\downarrow\,\,$\textbf{Standard Ribbon}$\,\,\downarrow$}\\\midrule
Standard Ribbon                                  &             & 64  & 1   & 14;20 & 26   & 4  & 4  & 64   & 16  & 16  & 320  & 77   & 99   \\
                                                 &             & 64  & 3   & 14;20 & 27   & 9  & 9  & 65   & 54  & 51  & 326  & 118  & 138  \\
                                                 &             & 64  & 5   & 14;20 & 32   & 13 & 20 & 70   & 72  & 62  & 319  & 147  & 165  \\
                                                 &             & 64  & 7   & 14;20 & 32   & 16 & 20 & 70   & 78  & 66  & 324  & 234  & 194  \\
                                                 & 10\,\% Pad  & 64  & 7   & 14;20 & 32   & 16 & 20 & 70   & 78  & 66  & 330  & 167  & 177  \\
\bottomrule
 \end{tabular}
 \end{table}

\eject
\clearpage
\thispagestyle{empty}
\pdfpagewidth=10in
\pdfpageheight=25.5in
 \begin{table}\flushright
   \caption{Performance of pure AMQs\label{tbl:full2}}
 \begin{tabular}{rc|*{3}{>{\raggedleft}p{2em}}|*{3}{>{\raggedleft}p{2em}}|*{2}{>{\raggedleft}p{2.2em}}>{\raggedleft\arraybackslash}p{2.2em}}
   \toprule
& Space  & \multicolumn{3}{c|}{ns/key, $n=10^6$} & \multicolumn{3}{c|}{ns/key, $n=10^8$} & \multicolumn{3}{c}{parallel, $n\!=\!10^7$} \\
                       Method & ovr \% & con & pos & neg & con & pos & neg & con & pos  & neg \\\midrule
\multicolumn{11}{c}{$\downarrow\,\,$\textbf{Bloom filter variants}$\,\,\downarrow$}\\\midrule
                      Bloom8  & 44.6    & 9   & 8   & 19  & 53  & 59  & 60  & 199 & 752  & 271 \\
            Bloom8, bulk add  & 44.6    & 21  & 8   & 19  & 43  & 59  & 60  & 175 & 750  & 256 \\
                     Bloom10  & 44.3    & 10  & 9   & 19  & 63  & 66  & 62  & 266 & 913  & 273 \\
                     Bloom12  & 44.4    & 12  & 10  & 19  & 73  & 73  & 63  & 332 & 1058 & 267 \\
           Bloom12, bulk add  & 44.4    & 30  & 10  & 19  & 58  & 73  & 63  & 204 & 1174 & 257 \\
\bottomrule
 \end{tabular}
\end{table}

\fi % \iftr

\end{document}